\documentclass[unpublished,onecolumn, letterpaper, 11pt]{quantumarticle}
\pdfoutput=1

\usepackage[usenames,dvipsnames,svgnames]{xcolor}
\definecolor{mypurple}{RGB}{95,45,145}

\usepackage[colorlinks,citecolor=mypurple,linkcolor=mypurple,urlcolor=mypurple,bookmarks=true,backref=page]{hyperref}

\renewcommand*{\backref}[1]{}

\renewcommand*{\backrefalt}[4]{
 \ifcase #1
 \or
 [p.~#2]
 \else
 [pp.~#2]
 \fi
}

\usepackage{xspace}
\usepackage{amssymb,amsmath,amsthm,graphicx}
\graphicspath{{figures/}{figs/}}
\usepackage{mathtools}

\usepackage{thmtools, thm-restate}
\usepackage[margin=1.in]{geometry}
\allowdisplaybreaks
\usepackage{orcidlink}
\usepackage{microtype}
\usepackage[utf8]{inputenc}
\usepackage[english]{babel}
\usepackage[T1]{fontenc}
\usepackage{braket}
\usepackage{dsfont}
\usepackage{cite}
\usepackage[capitalise,nameinlink]{cleveref}
\usepackage{enumitem}

\theoremstyle{plain}
\newtheorem{theorem}{Theorem}[section]
\newtheorem{proposition}[theorem]{Proposition}
\newtheorem{lemma}[theorem]{Lemma}
\newtheorem{claim}[theorem]{Claim}
\newtheorem{fact}[theorem]{Fact}
\newtheorem{corollary}[theorem]{Corollary}

\theoremstyle{definition}
\newtheorem{definition}[theorem]{Definition}

\theoremstyle{remark}
\newtheorem{remark}[theorem]{Remark}

\theoremstyle{definition}
\newtheorem{setting}[theorem]{Setting}
\theoremstyle{plain}

\def\Reals{{\mathbb{R}}}
\def\Complex{{\mathbb{C}}}
\def\Density{{\mathbb{D}}}
\renewcommand{\Pr}{\mathop{\bf Pr\/}}
\newcommand{\E}{\mathop{\mathbb E\/}}
\newcommand{\eps}{\varepsilon}

\renewcommand{\tilde}{\widetilde}
\newcommand{\ketbra}[2]{\ket{#1}\!\bra{#2}}
\newcommand{\diag}{\mathrm{diag}}
\newcommand{\B}{\{0,1\}}
\newcommand{\dd}{\mathrm{d}}
\newcommand{\nn}{\nonumber\\}
\newcommand{\shadow}{\mathrm{shadow}}
\newcommand{\sgn}{\mathrm{sgn}}
\DeclareMathOperator{\tr}{tr}
\DeclareMathOperator{\me}{\mathrm{e}}
\newcommand{\norm}[1]{\left\lVert#1\right\rVert}
\DeclareMathOperator{\Var}{\mathrm{Var}}
\DeclareMathOperator{\comm}{comm}
\DeclareMathOperator{\Rad}{Rad}
\DeclareMathOperator{\Alt}{Alt}
\DeclareMathOperator{\End}{End}
\DeclareMathOperator{\im}{im}
\DeclareMathOperator{\rank}{rank}
\DeclareMathOperator{\spec}{spec}

\newcommand{\channel}{\mathcal{E}}
\newcommand{\mcl}{\mathcal{L}}
\newcommand{\mcc}{\mathcal{C}}
\newcommand{\mch}{\mathcal{H}}
\newcommand{\mcu}{\mathcal{U}}
\newcommand{\mcm}{\mathcal{M}}
\newcommand{\mct}{\mathcal{T}}
\newcommand{\mcd}{\mathcal{D}}
\newcommand{\mcw}{\mathcal{W}}
\newcommand{\mcv}{\mathcal{V}}

\newcommand{\mbs}{\mathbb{S}}
\newcommand{\mbu}{\mathbb{U}}
\newcommand{\mbo}{\mathbb{O}}
\newcommand{\id}{\mathds{1}}
\newcommand{\om}{\ketbra{\Omega}{\Omega}}
\newcommand{\etal}{\textit{et al.}\xspace}
\newcommand{\Wg}{\mathrm{Wg}}
\newcommand{\Brauer}{\mathfrak{B}}
\newcommand{\sit}{Engineering Cluster, \href{https://ror.org/01v2c2791}{Singapore Institute of Technology}, 1 Punggol Coast Road, Singapore 828608, Republic of Singapore\looseness=-1}

\begin{document}

\title{Real Classical Shadows with Noise}
\author{Atharva Hingane\!
\orcidlink{0009-0004-3848-4060}}
\email[]{atharvah@iisc.ac.in}
\affiliation{Department of Computational and Data Sciences, \href{https://ror.org/04dese585}{Indian Institute of Science}, Bengaluru 560012, India}
\author{Dax Enshan Koh\! \orcidlink{0000-0002-8968-591X}}
\email[]{dax.koh@singaporetech.edu.sg}
\affiliation{\sit}

\begin{abstract}
The real classical shadows protocol of West \etal\ [J.~Phys.~A: Math.~Theor. \textbf{58}, 245304 (2025)] replaces the unitary (Clifford) ensemble of the Huang--Kueng--Preskill scheme by the orthogonal (real Clifford) ensemble, and for symmetric observables achieves strictly smaller estimator variances: a factor approaching two for global evolution and an exponential factor $(3/2)^k$ for $k$-local real Pauli observables. Real hardware, however, never implements the ideal evolution. Building on the noisy classical shadows framework of Koh and Grewal [Quantum \textbf{6}, 776 (2022)], we give a complete theory of the real classical shadows protocol in the presence of a known completely positive trace-preserving noise channel $\channel$ acting after the orthogonal evolution. We derive the noisy global and local orthogonal shadow channels from first principles using the Weingarten calculus of the orthogonal group, prove that each is a depolarizing channel acting on the symmetric (respectively locally symmetric) component of its input, and derive from it the exact single-shot variance in closed form, together with the associated shadow seminorm, two-sided bounds on it, and the resulting sample-complexity guarantees. We prove that the noiseless sample-complexity advantages survive intact under noise. Because the variances are exact rather than bounded, the ratio is controlled by a single dimensionless parameter, which gives a closed-form criterion for when the factor of two is attainable: both the second-moment and the variance ratio reach it exactly when the observable's norm profile grows, and the noise enters that limit only through a factor lying within $2/(d+2)$ of two, so the advantage is uniform in the noise. For rank-one targets it is provably unattainable, saturating strictly below two. The local real-Pauli advantage remains $(3/2)^k$. We treat complex measurement bases through a ``reality'' parameter and a transposed-noise scalar~$\tilde\beta$, recovering unitary shadows in the appropriate limit. We survey the Brauer-algebra and orthogonal-Weingarten machinery in an extended appendix, including a representation-theoretic account of the singularity of the order-three Gram matrix at $d=2$.
\end{abstract}

\maketitle
\tableofcontents

\section{Introduction}
\label{sec:intro}

Estimating expectation values $\tr(O\rho)$ of an unknown $n$-qubit state $\rho$ is a basic subroutine across near-term quantum algorithms, and it is often the dominant cost. The \emph{classical shadows} protocol of Huang, Kueng and Preskill~\cite{huang2020predicting} addresses this by replacing full state tomography with randomized single-copy measurements: one evolves $\rho$ by a random unitary $U$ drawn from an ensemble $\mcu$, measures in a fixed basis $\mcw$, and stores a classical description of the back-evolved measurement projector. A short classical dataset then predicts many observables, with a sample complexity governed by the estimator variance, which in turn depends on the choice of $\mcu$~\cite{huang2020predicting,elben2022randomized}.

Different ensembles are suited to different observables, and what decides the match is how much of operator space the observable occupies rather than any one of its invariants. For a global Clifford ensemble the sample cost is governed by $\tr(O_0^2)$ against $\norm{O_0}_\infty^2$, so the protocol is efficient for observables whose relevant operator support is small compared with the Hilbert space. Rank is a proxy for this and not the criterion itself: it enters only through $\tr(O_0^2)\le\operatorname{rank}(O_0)\,\norm{O_0}_\infty^2$, and the identity has maximal rank while costing nothing at all, since $\id_0=0$ and the protocol never estimates it. Local Clifford unitaries are instead efficient for local Pauli observables, the cost being exponential in Pauli weight. The two are not interchangeable, because a global ensemble that is a $3$-design has the Haar shadow norm and that norm is unitarily invariant, so it is indifferent to how entangled the target is, whereas Pauli weight --- which is what the local cost tracks --- is precisely what entanglement inflates. Fermionic Gaussian ensembles are efficient for low-degree Majorana observables~\cite{zhao2021fermionic,wan2023matchgate,heyraud2025unified}. West, Mele, Larocca and Cerezo~\cite{west2025real} recently closed a conceptual gap by analyzing the orthogonal group equivalently, by the real-Clifford orthogonal $3$-design property~\cite{hashagen2018real}, the \emph{real Clifford} ensemble together with real measurement bases. For symmetric observables whose norm profile grows with the number of qubits they obtain a variance reduction over global unitary shadows approaching a factor of two, and for $k$-local observables whose Pauli support avoids $Y$ they obtain an exponential-in-$k$ improvement, from $4^k$ to $3^k$ for general local operators and from $3^k$ to $2^k$ for individual Pauli strings.

These advantages are derived under the idealization of noiseless evolution. On real hardware, the evolution is never ideal, and the practical question is whether the advantage survives once noise is accounted for. For the unitary protocol this question was answered by Koh and Grewal~\cite{koh2022classical}, who modelled a known completely positive trace-preserving (CPTP) error channel $\channel$ acting after the ideal unitary and showed that (i) for a global Clifford ensemble the noisy shadow channel remains a depolarizing channel, so that inverting the noisy channel in post-processing restores an unbiased estimator, and (ii) the sample complexity degrades only by a controllable factor captured by a \emph{shadow seminorm}. We carry out the analogous program here for the real classical shadows protocol.

\subsection{Contributions}
\label{sec:contributions}
We give a complete and self-contained theory of noisy real classical shadows. Concretely:
\begin{enumerate}[leftmargin=*]
\item \textbf{Global orthogonal channel (\cref{prop:global_noisy_channel}).} For an orthogonal $3$-design and a real basis, and for a CPTP (more generally trace-preserving or unital) channel $\channel$, the noisy shadow channel is $\mcm_{\mbo,\channel}=\mcd_{n,f(\channel)}\circ(\cdot)_{\mathrm{sym}}$, the composition of a depolarizing channel with the symmetrization map $(\cdot)_{\mathrm{sym}}:A\mapsto A_{\mathrm{sym}}=\tfrac12(A+A^\intercal)$, a depolarizing channel of parameter
$f(\channel)=\tfrac{2(\beta-1)}{(d-1)(d+2)}$ acting on the symmetric part, where $\beta=\tr[\channel\circ\diag]$ is the diagonal weight of the noise ($\diag:A\mapsto\sum_b\Pi_bA\Pi_b$ is the completely dephasing map) and $d=2^n$. Inverting on the visible (symmetric) subspace restores unbiasedness.
\item \textbf{Shadow seminorm and sample complexity (\cref{prop:shadow_norm,cor:seminorm_bounds,cor:sample_complexity}).} We compute the shadow seminorm exactly via the orthogonal order-three Weingarten calculus, give two-sided bounds of matching form (they differ by the factor $5$ that the relaxation $\norm{O_0^2}_{\mathrm{sp}}\le\tr(O_0^2)$ costs, and are not tight, see \cref{sec:comparison}), and obtain a sample-complexity guarantee $N_{\mathrm{tot}}\le 170\,(d-1)^2(\beta-1)^{-2}\eps^{-2}\log(2M/\delta)\max_i\tr(O_i^2)$.
\item \textbf{Robustness of the advantage (\cref{sec:comparison}).} We show that for observables with $\kappa\to\infty$ the exact second-moment ratio tends to $2$, uniformly in the noise, since the noise enters only through $\varrho_L$, and $\sup_{1<\beta\le d}|\varrho_L-2|=2/(d+2)$, attained as $\beta\to1^{+}$, because $\varrho_L$ is increasing in $\beta$, running from $2(d+1)/(d+2)$ at $\beta\to1^{+}$ to $2(d+1)(d+4)/(d+2)^2$ at $\beta=d$, so the two endpoints give $|\varrho_L-2|=2/(d+2)$ and $2d/(d+2)^2$ respectively and the former is the larger. \cref{cor:var_ratio} transfers this to the variance ratio for observables with $\kappa\to\infty$, and shows that for bounded $\kappa$ the two ratios have different limits. The local real-Pauli advantage remains $(3/2)^k$. We carefully distinguish this exact factor-of-two from the weaker factor $204/170=1.2$ that compares the loose seminorm upper bounds.
\item \textbf{Local orthogonal channel (\cref{prop:local_noisy}).} For product ensembles and product noise $\channel=\channel_1^{\otimes n}$ the channel factorizes, and for a Pauli string $P$ of weight $\mathrm{wt}(P)$ (the number of qubits on which $P$ acts non-trivially) the locally symmetric $k$-Pauli seminorm equals $(2f_1(\channel_1)^2)^{-\mathrm{wt}(P)}$, where $f_1(\channel_1)=\tfrac12(\tr[\channel_1\circ\diag]-1)$ is the single-qubit depolarizing parameter.
\item \textbf{Complex bases (\cref{prop:complex_basis,prop:local_complex}).} We generalize to bases of arbitrary \emph{reality} $\alpha_{\mathrm r}$, introducing the transposed-noise scalar $\tilde\beta=\tr[\channel\circ\widetilde\diag]$ (the analogue of $\beta$ with the transposed dephasing map $\widetilde\diag:A\mapsto\sum_b\bra bA\ket b\,\Pi_b^\intercal$ of \eqref{eq:diagtilde_def}), treat both the global (Part III) and local (Part IV) complex-basis ensembles, and recover the unitary shadow \emph{channel} in the large-$d$, zero-reality limit. For a complex basis, the visible space becomes the full operator space, so non-symmetric observables become estimable.
\item \textbf{An exact advantage criterion (\cref{sec:comparison}) and many-body case studies (\cref{sec:casestudy}).} We prove that the orthogonal-versus-unitary second-moment ratio depends on the state and observable through a single dimensionless parameter $\kappa$, derive the resulting master curve, identify exactly when its finite-$d$ ceiling exceeds two (a condition on the noise, not on the dimension), and use it to show that the advantage saturates strictly below two for rank-one targets but approaches two exponentially in $n$ for extensive Hamiltonians. We instantiate this for GHZ fidelity and the critical transverse-field Ising energy, and show that the scalar spin chirality, being antisymmetric and hence time-reversal-odd, lies in the kernel of the real protocol and requires the complex bases of Parts~III--IV.
\item \textbf{Design-independence (\cref{sec:optimality}).} All orthogonal $3$-designs give literally identical predictions (\cref{prop:design_indep}), so no result here depends on which design is implemented, and in particular the real Cliffords may be substituted for the Haar measure on $\mbo(d)$ without altering any quantity in this paper.
\item \textbf{Examples (\cref{sec:examples}).} We specialize to depolarizing, dephasing, amplitude-damping, coherent over-rotation, and readout-error noise, showing in particular that arbitrary CPTP noise (including complex coherent errors) preserves the real-basis structure and enters only through $\beta$.
\item \textbf{Weingarten survey (\cref{app:brauer}).} We give a complete, self-contained account of the Brauer algebra, the orthogonal Weingarten/Gram matrix, and, in particular, a representation-theoretic explanation of the singularity of the order-three Gram matrix at $d=2$, together with explicit null vectors and a proof that the physical linear system remains solvable there.
\end{enumerate}

\subsection{Relation to prior work}
\label{sec:related}
Our noise model and post-processing strategy follow Koh and Grewal~\cite{koh2022classical}, who treated the unitary ensemble, whereas the technical novelty here is that the relevant Haar averages are over the orthogonal group, whose commutant is the Brauer algebra rather than the group algebra of the symmetric group~\cite{brauer1937algebras,goodman2009symmetry,collins2006integration}. West \etal~\cite{west2025real} give the noiseless real protocol and its variance analysis. We reproduce their results as the $\channel=\id$ limit of ours and confirm them numerically. 

Robust shadow estimation with unknown noise was studied by Chen \etal~\cite{chen2021robust}. As in~\cite{koh2022classical}, we instead assume the noise is characterized in advance, which allows deterministic inversion. Other works treat noise in shadow estimation from complementary angles: Brieger \etal~\cite{brieger2025stability} establish stability of the classical-shadows estimator under gate-dependent noise, going beyond the gate-independent Markovian model assumed here and in~\cite{koh2022classical}, while Jnane \etal~\cite{jnane2024quantum} combine classical shadows with quantum error mitigation to suppress the bias of noisy snapshots rather than invert a known channel. Closest to the present setting is the error-mitigated fermionic protocol of Wu and Koh~\cite{wu2024error}, which adopts the same gate-independent, time-stationary, Markovian noise model used here, but for the matchgate (fermionic Gaussian) ensemble and with the noise calibrated rather than assumed, so that a trusted noiseless input state supplies $\tilde{\mathcal O}(\sqrt n)$ calibration measurements, after which $k$-RDMs are estimated with $\tilde{\mathcal O}(kn^k)$ copies. That ensemble is a sibling of ours in a precise sense, a fermionic Gaussian unitary acts on Majorana operators as $\gamma_\mu\mapsto\sum_\nu R_{\mu\nu}\gamma_\nu$ with $R\in\mbo(2n)$, so matchgate shadows are also governed by orthogonal rather than unitary Haar averages, though of $\mbo(2n)$ on the Majorana index space rather than of $\mbo(2^n)$ on the Hilbert space itself. 

Where the noise sits relative to the evolution decides what kind of theory is possible. In the model used here and in~\cite{koh2022classical}, a single channel $\channel$ acts once, between the evolution and the measurement, so the twirl sees it as one object and the whole noise dependence collapses onto the scalar $\beta$, identically for every observable (\cref{prop:global_noisy_channel}). If instead a channel acts after each layer of the circuit that implements the evolution, each copy is twirled by only the part of the circuit that follows it, and no single scalar absorbs them all. Yu \etal~\cite{yu2026light} work out that case, giving a microscopic account of how in-circuit gate and readout errors bias randomized-measurement estimates for locally scrambled shallow circuits, showing that independent local twirling reduces the local implementation noise to stochastic Pauli damping, a noise event contributes only where it overlaps the Heisenberg evolution of the measured Pauli operator, and in one dimension the resulting activated noise volume grows linearly in the support of a contiguous observable, so the damping of the estimated Pauli coefficient is exponential in the observable's size. Their bias is therefore observable-dependent where ours is a single number, and their small-string calibration --- predicting large-string observables from small ones without learning the full noisy measurement channel --- is the tool that situation calls for.

Exploiting a symmetry to reduce the sampling cost is a programme in its own right, for which Sauvage \etal~\cite{sauvage2024classical} give criteria for symmetry-adapted shadows and work out the permutation-invariant case in detail. Their symmetry is carried by the \emph{state}: a group under which $\rho$ is invariant, permutation invariance being the case they develop. The symmetry here is carried instead by the \emph{observables and the measurement basis}, and it is the involution $O\mapsto O^\intercal$, transposition in the computational basis. Its fixed points are the symmetric operators --- for Hermitian $O$, exactly the real-valued ones (\cref{sec:observables}) --- and its $(-1)$ eigenspace, the antisymmetric operators, is precisely what the real protocol cannot see (\cref{lem:reduction}). Restricting the ensemble from $\mbu(d)$ to the subgroup that respects this involution is what buys the reduced variance, and it is the same trade Sauvage \etal\ describe, made on the other argument of the protocol. What is particular to the present instance is that the involution is small enough for every quantity to be computed in closed form. Further noise-aware developments treat Pauli-invariant unitary ensembles~\cite{bu2024classical} and shallow quantum circuits~\cite{farias2025robust}, and other practical adaptations of classical shadows to experimental imperfections~\cite{nguyen2023shadow,onorati2024noise,hu2025demonstration,rozon2024optimal}. Readout error specifically, which enters our framework through a single scalar in \cref{sec:readout}, has been treated at the level of correlated confusion matrices by~\cite{guo2026tensor}.

Our approach is orthogonal to these approaches. We retain the known-noise, deterministic-inversion setting of~\cite{koh2022classical} and ask instead how the orthogonal (real-Clifford) ensemble changes the resulting channel, seminorm, and sample complexity. We use standard orthogonal Weingarten calculus~\cite{weingarten1978asymptotic,collins2006integration,collins2009some,hashagen2018real}. Garc\'ia-Mart\'in \etal~\cite{garciamartin2025quantum} set up the same Brauer scaffolding for the orthogonal group, working in the large-$d$ asymptotic regime. Our contribution in \cref{app:brauer} is a unified and explicit treatment tailored to the third moment and to $d=2$, where the order-three Gram matrix is inverted exactly at every $d$ rather than asymptotically, and the statement that the $(2k)!/(2^kk!)$ Brauer diagrams form a basis of $\comm(\mbo(d),k)$---exact for $d\ge k$---fails at $d=2$, $k=3$, where the rank drops from $15$ to $10$. That is the single-qubit case Parts~II and~IV rest on.

\subsection{Structure of the paper}
\cref{sec:background} fixes notation and reviews classical shadows, the real protocol, the noise model, the scalar invariants $\alpha,\beta,\tilde\beta$, and the orthogonal Weingarten toolbox. We then develop the theory in four parallel Parts, one per measurement setting, each self-contained with its own reconstruction, inverse, estimator, seminorm/variance, sample complexity, and representative noise models: \cref{sec:results} (Part~I: real basis, global ensemble); \cref{sec:partII} (Part~II: real basis, local ensemble); \cref{sec:partIII} (Part~III: complex basis, global ensemble); \cref{sec:partIV} (Part~IV: complex basis, local ensemble). We prove the exact criterion for when the factor-of-two advantage is realized in \cref{sec:comparison}, alongside the two comparisons it disentangles, and \cref{sec:casestudy} applies it to three many-body case studies (GHZ certification, transverse-field Ising energy estimation, and the scalar spin chirality). We prove design-independence in \cref{sec:optimality}, where all orthogonal $3$-designs are shown to coincide; we conclude in \cref{sec:discussion}. \cref{app:brauer} is the Brauer/Weingarten survey including the order-three Gram matrix, its singularity at $d=2$, and the separate treatment of the $d=2$ complex-basis third moment that Part~IV needs, \cref{app:global_channel_proof,app:shadow_norm_proof,app:seminorm_bounds,app:complex_basis_proof,app:local_proof} the deferred proofs, \cref{app:complex_third_moment} the closed form of the complex-basis third moment, \cref{app:examples} the example details.

\section{Preliminaries and Notation}
\label{sec:background}

We write $\mch=\Complex^{d}$ with $d=2^n$, and $\mcl(\mch)$ for the linear operators on $\mch$. We write $\langle A,B\rangle=\tr(A^\dagger B)$ for the Hilbert--Schmidt inner product. We take transpose and complex conjugation in the computational basis. For a vector $\ket v$ we write $\ket{v^*}$ for its entrywise complex conjugate in that basis (so $\Pi_v^\intercal=\ketbra{v^*}{v^*}$, and $\ket{v^*}=\ket v$ for a computational-basis state), and $A_{\mathrm{sym}}=\tfrac12(A+A^\intercal)$ denotes the symmetric part. Symmetrization and the traceless projection $A\mapsto A-\tr(A)\id/d$ commute, since $\tr(A^\intercal)=\tr A$ and $\id^\intercal=\id$, so symmetrizing and then removing the trace gives the same operator as removing the trace and then symmetrizing. We write $A_{\mathrm{sym};0}$ for that common value, and $A_{\mathrm{sym};0}=A_0$ whenever $A$ is symmetric. For $b\in\B^n$ we write $\Pi_b=\ketbra{b}{b}$. The depolarizing channel of parameter $f$ (Koh--Grewal convention) is
\begin{equation}\label{eq:depol_def}
\mcd_{n,f}(A) = f\,A + (1-f)\,\tr(A)\,\frac{\id}{d},
\end{equation}
so that $\mcd_{n,f}^{-1}=\mcd_{n,1/f}$ whenever $f\neq0$.

\subsection{Classical shadows}
\label{sec:cs_review}
A classical shadows protocol is fixed by an ensemble $\mcu$ (a set of unitaries with a probability measure) and a measurement basis $\mcw=\{\ket{w}\}_w$. Each copy of $\rho$ is evolved by $U\sim\mcu$, measured in $\mcw$ to give outcome $\ket w$ with probability $\bra w U\rho U^\dagger\ket w$, and the snapshot $U^\dagger\Pi_w U$ is stored. In expectation, this implements the shadow channel
\begin{equation}\label{eq:cs_channel}
\mcm_{\mcu,\mcw}(\rho)=\sum_{w\in\mcw}\int_{U\sim\mcu} \tr[\rho\,U^\dagger\Pi_w U]\,U^\dagger\Pi_w U .
\end{equation}
When $\mcm_{\mcu,\mcw}$ is invertible on the relevant subspace, the classical shadow $\hat\rho=\mcm_{\mcu,\mcw}^{-1}(U^\dagger\Pi_w U)$ satisfies $\E[\hat\rho]=\rho$, and $\hat o=\tr(O\hat\rho)$ is an unbiased estimator of $\tr(O\rho)$ for every $O$ in the visible space $\mathrm{span}\{U^\dagger\Pi_w U\}$. Its variance $\Var[\hat o]$ controls the number of samples through a median-of-means construction~\cite{jerrum1986random,fu2025classical}, recalled in \cref{fact:mom}.
\begin{fact}[Median of means; the construction is~\cite{jerrum1986random} and the constant $34$ is the one used by~\cite{koh2022classical}, which follows from the Chebyshev step and the stated tail rather than from~\cite{jerrum1986random}]\label{fact:mom}
Let $X$ have variance $\sigma^2$. Then $K$ independent means of $N=\lceil 34\sigma^2/\eps^2\rceil$ samples each yield an estimator $\hat\mu(N,K)$ with $\Pr[\,|\hat\mu-\E X|\ge\eps\,]\le 2\me^{-K/2}$ for all $\eps>0$.
\end{fact}

\subsection{Real classical shadows}
\label{sec:rcs_review}
In the real protocol~\cite{west2025real} the ensemble is the Haar measure on the orthogonal group $\mbo(d)$ (globally) or $\mbo(2)^{\otimes n}$ (locally). By the orthogonal $3$-design property of the real Cliffords $\mcc_n\cap\mbo(d)$ and $(\mcc_1\cap\mbo(2))^{\otimes n}$~\cite{hashagen2018real,nebe2006self}, these Haar averages may be realized with real Clifford circuits without changing any moment up to the third. Throughout, we take ``orthogonal evolution'' to mean $U\in\mbo(d)$, so $U^\dagger=U^\intercal$. With a real measurement basis ($\Pi_w=\Pi_w^\intercal$) the noiseless global channel is~\cite{west2025real}
\begin{equation}\label{eq:west_global}
\mcm_{\mbo(d),\mcw}(A)=\frac{\tr[A]\id+A+A^\intercal}{d+2}=\mcd^{\mathrm{W}}_{d/(d+2)}(A_{\mathrm{sym}}),
\end{equation}
where $\mcd^{\mathrm{W}}_p(A)=\tfrac{p\,\tr[A]}{d}\id+(1-p)A$ is the depolarizing channel in the strength convention of~\cite{west2025real}. Since $\mcd^{\mathrm{W}}_p=\mcd_{n,1-p}$, the same channel in the convention~\eqref{eq:depol_def} reads $\mcm_{\mbo(d),\mcw}=\mcd_{n,2/(d+2)}\circ(\cdot)_{\mathrm{sym}}$. We return to this identification in \cref{rem:convention}. Here the visible space is the symmetric subspace, of dimension $d(d+1)/2$; the noisy version is derived in \cref{sec:global_channel}.

\subsection{Observables: symmetry and the traceless part}
\label{sec:observables}
We state two reductions here once, and apply both to every result below. The first fixes the class of observables the real protocol addresses, since for Hermitian $O$, symmetry $O=O^\intercal$ and being real-valued, $O=\bar O$ entrywise, are the same condition.\footnote{Indeed $O=O^\dagger=\bar O^\intercal$, so $O=O^\intercal$ forces $\bar O^\intercal=O^\intercal$ and hence $\bar O=O$, and conversely. We say \emph{symmetric} throughout, since $O=O^\intercal$ is the property the orthogonal commutant sees. \cite{west2025real} use both terms, though not interchangeably in practice, because every formal statement there is written for $A=A^\intercal$, while ``real-valued'' carries the abstract and the introduction, where the equivalence is asserted in a parenthesis and without the Hermiticity hypothesis it requires --- the implication above fails for non-Hermitian $A$.}

Neither condition is basis-independent: both refer to $O$ in the computational basis, in which the transpose is defined. Under a change of basis an operator transforms by similarity, $O\mapsto V^\dagger OV$, and for real $V$ this preserves symmetry,
\begin{equation}\label{eq:sym_basis}
\big(V^\intercal OV\big)^\intercal=V^\intercal O^\intercal V=V^\intercal OV ,\qquad V^\dagger=V^\intercal\ \text{ exactly when }V\ \text{is real},
\end{equation}
whereas a general unitary gives $\big(V^\dagger OV\big)^\intercal=V^\intercal O\bar V$, which is not $V^\dagger OV$. Symmetry is therefore preserved by $\mbo(d)$ and not by $\mbu(d)$. This is the structural reason the ensemble and the measurement basis have to be real together, and it is why the real protocol is not a change of variables inside the unitary one. \cite{west2025real} fix the transpose to the computational basis, but do not draw this consequence for observables.

A second reduction isolates the part of an observable that the protocol has to estimate. Write
\begin{equation}\label{eq:traceless_part}
O=\frac{\tr O}{d}\,\id+O_0,\qquad \tr O_0=0 ,
\end{equation}
for the decomposition of $O$ into its identity component and its traceless part $O_0$.

Both reductions are used in every bound below. Because $\tr O/d$ is a number rather than a random variable, the identity component of $O$ is known before any measurement is made and costs no samples to learn. Everything the protocol must actually estimate therefore sits in $O_0$, and every seminorm, variance and sample-complexity statement in this paper is written for $O_0$ on the strength of the lemma below.

All four shadow channels derived below have the same shape, a depolarizing map composed with one further linear map,
\begin{equation}\label{eq:factored_form}
\mcm=\mcd_{n,f}\circ\Psi .
\end{equation}
Its depolarizing factor carries the scalar $f$ and is what the estimator inverts. The second factor is the \emph{non-depolarizing residue}, whatever the twirl leaves once that scalar is stripped off: the symmetrization $\Psi(A)=A_{\mathrm{sym}}=\tfrac12(A+A^\intercal)$ in Parts~I--II, and in Parts~III--IV the map $\Psi_q(A)=qA+(1-q)A^\intercal$ of which it is the case $q=\tfrac12$. Stating the lemma for \eqref{eq:factored_form} proves it once instead of four times, one per Part.

Three decorated copies of $\mcm$ occur in the lemma, and we fix them here. Here $\mcm^{\dagger}$ is the Hilbert--Schmidt adjoint, defined by $\langle A,\mcm(B)\rangle=\langle\mcm^{\dagger}(A),B\rangle$, exactly as $\channel^{*}$ was in \cref{sec:scalars}. We write a star for the adjoint of a channel and a dagger for that of a superoperator built from one. A reader comparing with~\cite{koh2022classical} should note that they reserve the dagger for the operator adjoint. Next, $\mcm^{-1}$ denotes the inverse of $\mcm$ restricted to $\mcv$, extended by zero on $\mcv^{\perp}$: the restriction is what the lemma's hypothesis grants, and the extension is a choice, made so that $\mcm^{-1}$ is defined on every snapshot. Finally $\mcm^{-1,\dagger}$ is the adjoint of that map, and it annihilates $\mcv^{\perp}$ because $\mcm^{-1}$ does.

We define the visible space through the adjoint, $\mcv=\im\mcm^{\dagger}$, because that is the side the observable enters on. The estimator reaches $O_0$ only through $\hat o=\langle O_0,\mcm^{-1}(\sigma)\rangle=\langle\mcm^{-1,\dagger}(O_0),\sigma\rangle$, so what decides visibility is which observables survive the adjoint, not which operators $\mcm$ can produce. This is the computable form of the definition of~\cite{west2025real}, for whom the visible space is the set of observables whose estimator is unbiased: computable because $\im\mcm^{\dagger}=(\ker\mcm)^{\perp}$ for any linear map, so it is found by identifying $\ker\mcm$, which \cref{sec:global_channel} does.

Everything the lemma needs about $\Psi$, and the identification of $\mcv$ in each Part, follows from one statement.

\begin{claim}[The non-depolarizing residue]\label{claim:residue}
For $q\in\Reals$ let $\Psi_q(A)=qA+(1-q)A^\intercal$ on $\mcl(\Complex^d)$, so that $\Psi_{1/2}=(\cdot)_{\mathrm{sym}}$. Then, for every $f$:
\begin{enumerate}[label=(\alph*),leftmargin=*]
\item\label{it:res_comm} $\Psi_q$ is trace-preserving and fixes the identity, and consequently commutes with $\mcd_{n,f}$;
\item\label{it:res_sa} $\Psi_q$ is Hilbert--Schmidt self-adjoint, and hence so is $\mcm=\mcd_{n,f}\circ\Psi_q$, i.e.\ $\mcm^{\dagger}=\mcm$;
\item\label{it:res_inv} $\Psi_q$ is invertible on all of $\mcl(\Complex^d)$ if and only if $q\neq\tfrac12$, with $\Psi_q^{-1}(X)=\tfrac{qX-(1-q)X^\intercal}{2q-1}$; at $q=\tfrac12$ it is instead the orthogonal projector onto the symmetric operators, which it fixes pointwise.
\end{enumerate}
In both cases, and whenever $f\neq0$, the restriction $\Psi_q|_{\mcv}$ is a trace-preserving bijection of $\mcv=\im\mcm^{\dagger}$: for $q\neq\tfrac12$ because $\mcv=\mcl(\Complex^d)$ and $\Psi_q$ is invertible there, and for $q=\tfrac12$ because $\mcv$ is the symmetric subspace, on which $\Psi_{1/2}$ is the identity.
\end{claim}
\begin{proof}
\ref{it:res_comm} $\tr\Psi_q(A)=q\tr A+(1-q)\tr A^\intercal=\tr A$ and $\Psi_q(\id)=\id$. Hence $\mcd_{n,f}(\Psi_q(A))=f\Psi_q(A)+(1-f)\tr[\Psi_q(A)]\tfrac{\id}{d}=f\Psi_q(A)+(1-f)\tr(A)\tfrac{\id}{d}$, while $\Psi_q(\mcd_{n,f}(A))=f\Psi_q(A)+(1-f)\tr(A)\tfrac{\Psi_q(\id)}{d}$, which is the same. Only these two properties of $\Psi_q$ are used.

\ref{it:res_sa} The transpose is Hilbert--Schmidt self-adjoint, since $\tr[B^\dagger A^\intercal]=\tr[(B^\dagger A^\intercal)^\intercal]=\tr[A\bar B]=\tr[\bar BA]$; as $q$ is real, $\Psi_q$ is a real combination of it and the identity and is therefore self-adjoint too. So is $\mcd_{n,f}$, being a real combination of the identity and the map $A\mapsto\tr(A)\id/d$. Then $\mcm^{\dagger}=(\mcd_{n,f}\circ\Psi_q)^{\dagger}=\Psi_q^{\dagger}\circ\mcd_{n,f}^{\dagger}=\Psi_q\circ\mcd_{n,f}=\mcd_{n,f}\circ\Psi_q=\mcm$, the last step by \ref{it:res_comm}. That $q$ is real is not automatic and is checked where $q_\beta$ is introduced in \eqref{eq:Atilde}.

\ref{it:res_inv} Write $Y=qX+(1-q)X^\intercal$ and transpose it, $Y^\intercal=qX^\intercal+(1-q)X$. Eliminating $X^\intercal$ between the two gives $\big(q^2-(1-q)^2\big)X=qY-(1-q)Y^\intercal$, and $q^2-(1-q)^2=2q-1$, so the stated inverse is the unique solution when $q\neq\tfrac12$. At $q=\tfrac12$ the map is idempotent, $\Psi_{1/2}\circ\Psi_{1/2}=\Psi_{1/2}$, with range the symmetric operators and kernel the antisymmetric ones; being also self-adjoint by \ref{it:res_sa}, it is the orthogonal projector onto its range, and it annihilates the antisymmetric operators, so it has no inverse on $\mcl(\Complex^d)$.

For the last statement, $\mcm=\mcd_{n,f}\circ\Psi_q$ with $f\neq0$ makes $\mcd_{n,f}$ invertible, so $\ker\mcm=\ker\Psi_q$ and $\mcv=(\ker\Psi_q)^{\perp}$ by \ref{it:res_sa}. For $q\neq\tfrac12$ that kernel is trivial and $\mcv=\mcl(\Complex^d)$, on which $\Psi_q$ is a bijection by \ref{it:res_inv}; for $q=\tfrac12$ it is the antisymmetric operators, so $\mcv$ is the symmetric ones and $\Psi_{1/2}$ restricts there to the identity. Trace preservation is \ref{it:res_comm}.
\end{proof}

With $\Psi$ settled, the reduction every Part relies on can be stated once and for all four.

\begin{lemma}[Reduction to the traceless part]\label{lem:reduction}
Suppose the shadow channel factors as $\mcm=\mcd_{n,f}\circ\Psi$ with $f\neq0$, with $\Psi$ linear and trace-preserving, and with $\mcm$ invertible on the visible space $\mcv=\im\mcm^{\dagger}$; write $\mcm^{-1}$ for the inverse of $\mcm|_{\mcv}$, extended by zero on $\mcv^{\perp}$. This holds in each of \cref{prop:global_noisy_channel,prop:local_noisy,prop:complex_basis,prop:local_complex}, where by \cref{claim:residue} $\Psi$ restricts to a trace-preserving bijection of $\mcv$ --- the identity on $\mcv$ in Parts~I--II, and an invertible map that is \emph{not} the identity in Parts~III--IV. Then, for every unitary $U$ and outcome $w$ whose raw snapshot $\sigma\coloneqq U^\dagger\Pi_wU$ lies in $\mcv$:
\begin{enumerate}[label=(\roman*),leftmargin=*]
\item the classical shadow has unit trace, $\tr\hat\rho=1$, shot by shot;
\item consequently $\hat o=\tr(O\hat\rho)=\tfrac{\tr O}{d}+\tr(O_0\hat\rho)$ identically, so the identity component contributes a known constant and no variance:
$\Var[\hat o]=\Var[\tr(O_0\hat\rho)]$;
\item $\hat o$ depends on $O$ only through the component of $O_0$ in the visible space $\mcv=\im\mcm^{\dagger}$, because $\mcm^{-1,\dagger}$ annihilates $\mcv^{\perp}$.
\end{enumerate}
\end{lemma}
\begin{proof}
Write $\sigma=U^\dagger\Pi_wU$ for the raw snapshot and $\hat\rho=\mcm^{-1}(\sigma)$ for the classical shadow it produces.

\emph{(i).} We first check $\tr\sigma=1$. Since $\Pi_w$ is a rank-one orthogonal projector, $\tr\Pi_w=1$, and the trace is invariant under the similarity $\Pi_w\mapsto U^\dagger\Pi_wU$; hence $\tr\sigma=1$. We next check that $\mcm^{-1}$ preserves the trace. By hypothesis $\mcm|_{\mcv}=\mcd_{n,f}\circ\Psi$ is a bijection of $\mcv$, so on $\mcv$ we may write $\mcm^{-1}=(\Psi|_{\mcv})^{-1}\circ\mcd_{n,f}^{-1}$, the inverse being that of the restriction, which is what the hypothesis grants and which $\Psi$ itself does not possess, and it suffices that each factor preserve the trace. For any trace-preserving bijection $\Lambda$ the inverse is trace-preserving, since $\tr\Lambda^{-1}(X)=\tr\Lambda\big(\Lambda^{-1}(X)\big)=\tr X$; this applies to $\Psi|_{\mcv}$. And $\mcd_{n,f}^{-1}=\mcd_{n,1/f}$, which is trace-preserving because $\mcd_{n,g}$ is for every $g\neq0$, which is where $f\neq0$ is used. It is also where invertibility \emph{on $\mcv$} is used rather than on all of $\mcl(\Complex^d)$, because the symmetrization has no global inverse. Combining, $\tr\hat\rho=\tr\mcm^{-1}(\sigma)=\tr\sigma=1$.

\emph{(ii).} Substitute the decomposition \eqref{eq:traceless_part} into $\hat o=\tr(O\hat\rho)$ and use linearity of the trace:
\[
\hat o=\tr\!\Big[\Big(\tfrac{\tr O}{d}\,\id+O_0\Big)\hat\rho\Big] =\tfrac{\tr O}{d}\,\tr\hat\rho+\tr(O_0\hat\rho) =\tfrac{\tr O}{d}+\tr(O_0\hat\rho),
\]
the last equality by (i). Because it is determined by $O$ alone and depends on neither $U$ nor $w$, the first term is a \emph{constant}, the same on every shot. Adding a constant to a random variable shifts its mean and leaves its spread untouched, whence $\Var[\hat o]=\Var[\tr(O_0\hat\rho)]$. We note that $\Psi$ does not appear anywhere in this argument, since its entire contribution to (ii) is through (i), that is, through $\tr\hat\rho=1$.

\emph{(iii).} Move the inverse channel off the state and onto the observable. Writing $\langle A,B\rangle=\tr(A^\dagger B)$ for the Hilbert--Schmidt inner product, and using that $O_0$ is Hermitian,
\[
\tr(O_0\hat\rho)=\langle O_0,\mcm^{-1}(\sigma)\rangle=\langle\mcm^{-1,\dagger}(O_0),\sigma\rangle ,
\]
the second equality being the definition of the adjoint. Now decompose $O_0$ orthogonally against the visible space, $O_0=P_{\mcv}O_0+P_{\mcv^{\perp}}O_0$. Since $\mcm^{-1}$ was extended by zero on $\mcv^{\perp}$, its adjoint annihilates $\mcv^{\perp}$, so the second summand contributes nothing:
\[
\mcm^{-1,\dagger}(O_0)=\mcm^{-1,\dagger}\big(P_{\mcv}O_0\big).
\]
Substituting back into (ii), $\hat o=\tfrac{\tr O}{d}+\langle\mcm^{-1,\dagger}(P_{\mcv}O_0),\sigma\rangle$, which depends on $O$ only through $\tr O$ and $P_{\mcv}O_0$. Two observables with the same trace whose traceless parts agree on $\mcv$ therefore give the same estimator on every shot, not merely in expectation.
\end{proof}

Every seminorm, variance and sample-complexity statement below is therefore stated for $O_0$, and by (iii) only its visible component matters. In Parts~I--II the visible space is the symmetric subspace, so an antisymmetric observable is invisible; Parts~III--IV enlarge $\mcv$ to all of $\mcl(\Complex^d)$, both by \cref{claim:residue}\ref{it:res_inv}.

Two consequences of \cref{claim:residue}\ref{it:res_sa} are worth separating from the lemma, because the lemma is stated without them. Since $\mcm^{\dagger}=\mcm$, the visible space is also $\im\mcm$, and the restriction $\mcm|_{\mcv}$ is an automorphism of $\mcv$. Extending $\mcm^{-1}$ by zero then gives the Moore--Penrose pseudoinverse. The two readings of $\mcm^{-1,\dagger}$, the adjoint of the inverse and the inverse of the adjoint, therefore agree. We nonetheless write the lemma with the dagger. Part~(iii) uses only that $\mcm^{-1,\dagger}$ annihilates $\mcv^{\perp}$, and that holds by construction rather than by self-adjointness. It would therefore survive a channel with $\mcm^{\dagger}\neq\mcm$. None of the four Parts is such a channel, but a gate-dependent noise model need not respect the equality (\cref{sec:discussion}).

\subsection{Noise model}
\label{sec:noise_model}
We adopt the gate-independent, time-stationary, Markovian model of~\cite{koh2022classical}: a known CPTP channel $\channel$ acts on the state after the ideal evolution and before measurement. The \emph{noisy shadow channel} is
\begin{equation}\label{eq:noisy_channel}
\mcm_{\mcu,\channel}(\rho)=\E_{U\sim\mcu}\sum_{b\in\B^n}\bra b\channel(U\rho U^\dagger)\ket b\,U^\dagger\Pi_b U .
\end{equation}
Because $\channel$ is known, the post-processing inverts $\mcm_{\mcu,\channel}$ rather than the noiseless channel, which (when the inverse exists on the visible space) keeps the estimator unbiased. As in~\cite{koh2022classical}, the inverse need not itself be a channel. It is applied as a classical map on a classical description.

\subsection{Scalar invariants}
\label{sec:scalars}
Write $\channel^*$ for the dual of $\channel$ under the bilinear trace pairing, defined by $\tr[\channel^*(X)\,Y]=\tr[X\,\channel(Y)]$ for all $X,Y$. It exists for every linear superoperator and needs no Kraus representation; when $\channel(A)=\sum_i K_i A K_i^\dagger$ it is the familiar $\channel^*(A)=\sum_i K_i^\dagger A K_i$. This is the dual every identity below uses. It agrees with the Hilbert--Schmidt adjoint on Hermitian arguments whenever $\channel$ preserves Hermiticity, which the channels of \cref{sec:noise_model} do, and differs from it otherwise. We keep the dagger for the Hilbert--Schmidt adjoint of a \emph{shadow} channel, as in $\mcm^{\dagger}$; there the two again agree, since $\mcm$ preserves Hermiticity and is applied to Hermitian arguments. We write $\diag(A)=\sum_i \Pi_i A\Pi_i$ for the completely dephasing channel; its transposed variant is
\begin{equation}\label{eq:diagtilde_def}
\widetilde\diag(A)=\sum_i \bra i A\ket i\,(\ketbra ii)^{\intercal}=\sum_i \bra i A\ket i\,\ketbra{i^*}{i^*}.
\end{equation}
The three scalar invariants that will govern every result are
\begin{equation}\label{eq:scalars}
\alpha=\tr[\channel(\id)],\qquad
\beta=\tr[\channel\circ\diag]=\sum_b\bra b\channel(\Pi_b)\ket b,\qquad
\tilde\beta=\tr[\channel\circ\widetilde\diag]=\sum_b\bra b\channel(\Pi_b^\intercal)\ket b .
\end{equation}
Both right-hand sides of \eqref{eq:scalars} follow from the superoperator traces on their left by evaluating a superoperator trace in the Hilbert--Schmidt orthonormal basis $\{\ketbra ij\}$ gives $\tr\Phi=\sum_{i,j}\bra i\Phi(\ketbra ij)\ket j$. Both $\diag$ and $\widetilde\diag$ annihilate every off-diagonal $\ketbra ij$ with $i\neq j$, leaving $\diag(\ketbra ii)=\Pi_i$ and $\widetilde\diag(\ketbra ii)=\Pi_i^\intercal$, so only the $d$ diagonal terms survive and
\begin{equation}\label{eq:scalars_reduction}
\tr[\channel\circ\diag]=\sum_i\bra i\channel(\Pi_i)\ket i,\qquad
\tr[\channel\circ\widetilde\diag]=\sum_i\bra i\channel(\Pi_i^\intercal)\ket i ,
\end{equation}
which is \eqref{eq:scalars}. For a real basis $\Pi_b^\intercal=\Pi_b$ and hence $\tilde\beta=\beta$.

Both scalars are confined to the same interval, and the hypothesis that confines them is complete positivity rather than the linearity that \cref{prop:global_noisy_channel} assumes. The bound on $\beta$ is stated for a quantum channel by~\cite{koh2022classical} in the course of bounding their depolarizing parameter; the companion bound on $\tilde\beta$, the equality condition, and the proof are ours.
\begin{claim}[Range of the diagonal weights]\label{claim:beta_range}
Let $\channel$ be completely positive and either trace-preserving or unital. Then
\begin{equation}\label{eq:beta_range}
0\;\le\;\beta\;\le\;d,\qquad 0\;\le\;\tilde\beta\;\le\;d ,
\end{equation}
and $\beta=d$ if and only if $\channel(\Pi_b)=\Pi_b$ for every $b$, which is the \emph{inconsequential} case of \cref{sec:examples}.
\end{claim}
\begin{proof}
Each $\Pi_b$ is a state and $\Pi_b^\intercal=\ketbra{b^*}{b^*}$ is again one, so it is enough to bound $\bra b\channel(\Pi)\ket b$ for a state $\Pi$. Complete positivity gives $\channel(\Pi)\succeq0$, whence $\bra b\channel(\Pi)\ket b\ge0$ and the lower bounds follow by summing over $b$. For the upper bound, if $\channel$ is trace-preserving then $\bra b\channel(\Pi)\ket b\le\tr\channel(\Pi)=\tr\Pi=1$, because a positive semidefinite operator has every diagonal entry at most its trace; summing the $d$ terms gives $\beta\le d$. If instead $\channel$ is unital, then $\sum_b\channel(\Pi_b)=\channel(\id)=\id$ and every summand is positive semidefinite, so $\beta=\sum_b\bra b\channel(\Pi_b)\ket b\le\sum_{b,b'}\bra b\channel(\Pi_{b'})\ket b=\tr\id=d$. Both arguments apply verbatim to $\tilde\beta$ with $\Pi_b^\intercal$ in place of $\Pi_b$. Equality in the trace-preserving case forces $\bra b\channel(\Pi_b)\ket b=\tr\channel(\Pi_b)$ for every $b$, and a positive semidefinite operator of unit trace whose $\ket b$ diagonal entry equals its trace is $\Pi_b$ itself.
\end{proof}
\noindent No ordering between $\beta$ and $\tilde\beta$ holds in general; for depolarizing noise \eqref{eq:depol_complex} gives $\beta-\tilde\beta=p(d-\alpha_{\mathrm r})\ge0$, but that is a statement about that channel and that basis. One further scalar belongs to the measurement basis rather than to the noise. Following~\cite{west2025real} we set $\alpha_w=|\braket{w}{w^*}|^2$ and write $\alpha_{\mathrm r}=\sum_w\alpha_w$ for the \emph{basis reality}, so $\alpha_{\mathrm r}=d$ exactly when the basis is real. Normalizing it gives the \emph{reality fraction}
\begin{equation}\label{eq:reality_fraction}
\varsigma\;=\;\frac{\alpha_{\mathrm r}}{d}\;\in[0,1] ,
\end{equation}
so $\varsigma=1$ is a real basis and $\varsigma\to0$ the opposite extreme. Parts~I--II fix $\varsigma=1$; Parts~III--IV vary it. We use the following invariance repeatedly; it makes the adjoint bookkeeping harmless.
\begin{lemma}[Adjoint invariance of $\alpha,\beta$]\label{lem:adjoint}
For any linear superoperator $\channel$, $\ \sum_b\tr[\channel^*(\Pi_b)]=\tr[\channel(\id)]=\alpha$ and $\ \sum_b\bra b\channel^*(\Pi_b)\ket b=\sum_b\bra b\channel(\Pi_b)\ket b=\beta$.
\end{lemma}
\begin{proof} We use only the defining property of $\channel^*$, $\tr[\channel^*(X)\,Y]=\tr[X\,\channel(Y)]$ for all $X,Y$, so no positivity or Kraus assumption is needed. For the first identity, linearity and $\sum_b\Pi_b=\id$ give
\begin{equation}
\sum_b\tr[\channel^*(\Pi_b)]=\tr\!\Big[\channel^*\!\Big(\sum_b\Pi_b\Big)\Big]=\tr[\channel^*(\id)]=\tr[\channel^*(\id)\,\id]=\tr[\id\,\channel(\id)]=\tr[\channel(\id)]=\alpha,
\end{equation}
where the middle equality applies the adjoint property with $X=\id$, $Y=\id$. For the second, applying the adjoint property with $X=Y=\Pi_b$,
\begin{equation}
\bra b\channel^*(\Pi_b)\ket b=\tr[\channel^*(\Pi_b)\,\Pi_b]=\tr[\Pi_b\,\channel(\Pi_b)]=\bra b\channel(\Pi_b)\ket b;
\end{equation}
Summing over $b$ gives the claim.
\end{proof}
For a trace-preserving or unital channel $\alpha=d$, and $\beta$ is confined by \cref{claim:beta_range}. Invertibility of the shadow channel will require only $\beta\neq1$.

\subsection{The orthogonal Weingarten toolbox}
\label{sec:weingarten_toolbox}
Every average below is over the orthogonal group, for which the relevant object is the $k$th moment operator (twirl) $\mct^{(k)}_{\mbo(d)}(A)=\int_{\mbo(d)}U^{\otimes k}A\,U^{\intercal\otimes k}\dd U$, equal to the orthogonal projector onto the commutant $\comm(\mbo(d),k)$. Brauer diagrams span the commutant: for $k=2$ it is $\mathrm{span}\{\id,\mbs,\om\}$, and for $k=3$ it is $15$-dimensional when $d\ge3$. We use the Weingarten formula $\mct^{(k)}(A)=\sum_{ij}\Wg_{ij}\langle x_i,A\rangle x_j$ with $\Wg=G(d)^{+}$ the Moore--Penrose pseudoinverse of the Gram matrix $G(d)_{ij}=\langle x_i,x_j\rangle$. This machinery, including the order-three Gram determinant and its singularity at $d=2$, is developed from first principles in \cref{app:brauer}.

Throughout we compute moments over the Haar measure on $\mbo(d)$, but the protocol is implemented with the finite \emph{real Clifford} group $\mcc_n\cap\mbo(d)$ (globally) or $(\mcc_1\cap\mbo(2))^{\otimes n}$ (locally). This is justified by the orthogonal $3$-design property~\cite{hashagen2018real,nebe2006self}, i.e., the group average of $U^{\otimes k}(\cdot)U^{\intercal\otimes k}$ over the real Cliffords equals the Haar-$\mbo(d)$ twirl for all $k\le3$, hence every quantity in this paper, all of which depend only on the second and third moments, is unchanged.

\section{Part I: Real Basis, Global Orthogonal Ensemble}
\label{sec:results}

\begin{setting}[Real basis, global ensemble]\label{set:I}
Fix $n\ge1$ and $d=2^n$, and assume throughout this Part:
\begin{enumerate}[label=(S\arabic*),leftmargin=*]
\item\label{it:I_ens} \emph{Ensemble.} $\mbo$ is the Haar measure on $\mbo(d)$ or any orthogonal $3$-design, for instance $\mcc_n\cap\mbo(d)$. By \cref{prop:design_indep} the choice is immaterial.
\item\label{it:I_basis} \emph{Basis.} The measurement basis $\mcw$ is real, $\Pi_w^\intercal=\Pi_w$ for every $w$, equivalently $\varsigma=1$ in \eqref{eq:reality_fraction}.
\item\label{it:I_noise} \emph{Noise.} $\channel$ is the completely positive map of \cref{sec:noise_model}, applied after the evolution, with $d\beta\neq\alpha$ so that the shadow channel is invertible on its visible space, and with $\beta>1$. The latter is what every seminorm and variance bound of this Part needs, and \cref{sec:shadow_norm} says where. Complete positivity is what bounds $\beta$ (\cref{claim:beta_range}) and what makes the outcome weights of \eqref{eq:seminorm_def} non-negative; the channel formula \eqref{eq:global_general} itself needs only linearity, and is stated that way. Statements labelled ``trace-preserving or unital'' additionally use $\alpha=d$.
\item\label{it:I_obs} \emph{Observable.} $O$ is Hermitian and symmetric, and only its traceless part $O_0$ of \eqref{eq:traceless_part} enters, by \cref{lem:reduction}. Symmetry alone is what the \emph{channel} needs --- \cref{prop:global_noisy_channel} holds for every $A\in\mcl(\Complex^d)$ --- and Hermiticity is used in three specific places: to make symmetry and real-valuedness the same condition (\cref{sec:observables}), to identify $\langle O_0,X\rangle$ with $\tr(O_0X)$ in the proof of \cref{lem:reduction}(iii), and to make $\max_\sigma\tr[\sigma O_0^2]=\norm{O_0^2}_{\mathrm{sp}}$ in \eqref{eq:seminorm_general}. It is also what makes $\tr(O\rho)$ real, which is the quantity being estimated.
\end{enumerate}
\end{setting}

\noindent Parts~II--IV each alter exactly one of \ref{it:I_ens}--\ref{it:I_basis} and state which results change. Proofs are in \cref{app:global_channel_proof,app:shadow_norm_proof,app:seminorm_bounds}.

\subsection{Reconstruction: the noisy global shadow channel}
\label{sec:global_channel}

What the protocol needs first is the channel it actually implements. Averaging
\eqref{eq:noisy_channel} over an orthogonal $3$-design gives it in closed form:

\begin{restatable}[Noisy global orthogonal shadow channel]{proposition}{propglobal}
\label{prop:global_noisy_channel}
Let $\mbo$ be an orthogonal $3$-design on $\Complex^d$, let $\channel$ be a linear superoperator, and let the measurement basis be real. Then for all $A\in\mcl(\Complex^d)$,
\begin{equation}\label{eq:global_general}
\mcm_{\mbo,\channel}(A)=\frac{(d+1)\alpha-2\beta}{d(d-1)(d+2)}\,\tr(A)\,\id
+\frac{2\,(d\beta-\alpha)}{d(d-1)(d+2)}\,A_{\mathrm{sym}},
\end{equation}
with $\alpha,\beta$ as in \eqref{eq:scalars}. If $\channel$ is trace-preserving or unital, then $\alpha=d$ and
\begin{equation}\label{eq:global_depol}
\mcm_{\mbo,\channel}(A)=\mcd_{n,f(\channel)}\!\big(A_{\mathrm{sym}}\big),
\qquad
f(\channel)=\frac{2(\beta-1)}{(d-1)(d+2)} .
\end{equation}
\end{restatable}

The channel annihilates the antisymmetric part $(A-A^\intercal)/2$, and $\im\mcm^{\dagger}=(\ker\mcm)^{\perp}$ for any linear map, so the visible space is exactly the symmetric operators,
\begin{equation}
\mathsf{Vis}(\mbo,\mcw)=\{A: A=A^\intercal\},\qquad \dim\mathsf{Vis}=\tfrac{d(d+1)}{2},
\end{equation}
the dimension counting the free entries on and above the diagonal. Both statements hold independently of $\channel$, since noise does not enlarge or shrink the visible space but only rescales the depolarizing parameter. This is the noisy analogue of the observation of~\cite{west2025real} that orthogonal shadows cannot estimate non-symmetric observables.

\begin{remark}[Convention, resolving a discrepancy]\label{rem:convention}
In the noiseless limit $\channel=\id$ we have $\beta=\tr[\diag]=d$, hence
\[
f(\id)=\frac{2(d-1)}{(d-1)(d+2)}=\frac{2}{d+2},
\]
and \eqref{eq:global_depol} reproduces \eqref{eq:west_global} exactly, since $\mcd_{n,2/(d+2)}(A_{\mathrm{sym}})=\tfrac{2}{d+2}A_{\mathrm{sym}}+\tfrac{d}{d+2}\tr(A)\tfrac{\id}{d}=\tfrac{A+A^\intercal+\tr(A)\id}{d+2}$. Note that $f(\id)=2/(d+2)$ in the convention~\eqref{eq:depol_def}, where the value $d/(d+2)$ is the depolarizing \emph{strength} $p$ of the complementary convention $\mcd^{\mathrm{W}}_p$ used in~\cite{west2025real}, related by $f=1-p$. We fix $f=1-p$ throughout.

A second, independent clash with the notation of~\cite{west2025real} follows from this choice, and we settle it here too. Because their $f$ is not committed to a depolarizing parameter, \cite{west2025real} are free to write the basis reality of \cref{sec:complex} as $\alpha_{\mathrm r}=fd$. For us $f$ is already taken, so we write $\alpha_{\mathrm r}=\varsigma d$ with the reality fraction $\varsigma\in[0,1]$ of \eqref{eq:reality_fraction} throughout Parts~III--IV. Thus $\varsigma=1$ is a real basis and $\varsigma\to0$ the unitary limit, and no statement of the form ``$f\to0$'' refers to the depolarizing parameter.

A third discrepancy is arithmetic, and we flag it because a reader comparing the two papers will meet it at the visible space. It is not a disagreement about the dimension of the same object: the object is the same in both papers, and only its name differs. \cite{west2025real} describe the visible space correctly as $\{A:A=A^\intercal\}$, then identify it with $\Complex\mathfrak{o}(d)$, the complexified orthogonal Lie algebra, and quote the dimension $d(d-1)/2$ of \emph{that} object. But $\mathfrak{o}(d)$ is the \emph{antisymmetric} algebra --- the complement of the set they described, and precisely what the channel annihilates. The set has dimension $d(d+1)/2$, as counted above. Their asymptotic conclusion is untouched, since $d(d\pm1)/2\big/d^2\to1/2$ either way, but the identification and the exact dimension are. We use $d(d+1)/2$ throughout.

The same misidentification recurs for the local ensemble, and there it is harmless. \cite{west2025real} write the local visible space as $\mathrm{span}_{\Complex}\{\id,X,Z\}^{\otimes n}=\bigotimes_j\Complex\mathfrak{o}(2)$, where the span is correct and $\Complex\mathfrak{o}(2)$ is one-dimensional, so the right-hand side is one-dimensional for every $n$. Their quoted ratio $(3/4)^n$ is nevertheless right, because it is computed from the span rather than from the name --- which is precisely what went wrong globally, where the dimension was read off the label instead. We agree with their span, and \cref{prop:local_noisy} states it.
\end{remark}

\begin{claim}[Invertibility and the classical shadow]\label{claim:invert}
For $\channel$ trace-preserving or unital, so that $\alpha=d$, the channel $\mcm_{\mbo,\channel}=\mcd_{n,f(\channel)}\circ(\cdot)_{\mathrm{sym}}$ is invertible on the symmetric subspace iff $f(\channel)\neq0$, i.e.\ iff $\beta\neq1$; for general $\alpha$ the condition is $d\beta\neq\alpha$, as in \cref{prop:shadow_norm}. In that case the inverse is $\mcd_{n,1/f(\channel)}\circ(\cdot)_{\mathrm{sym}}$, and the classical shadow for a snapshot $U^\intercal\Pi_b U$ (which is symmetric) is
\begin{equation}\label{eq:classical_shadow}
\hat\rho=\mcm_{\mbo,\channel}^{-1}(U^\intercal\Pi_b U)=\frac{1}{f(\channel)}\,U^\intercal\Pi_b U+\Big(1-\frac{1}{f(\channel)}\Big)\frac{\id}{d}.
\end{equation}
\end{claim}
\begin{proof}
On symmetric $Y$ the image $\mcd_{n,f}(Y)=fY+(1-f)\tr(Y)\id/d$ is again symmetric and has the same trace, $\tr[\mcd_{n,f}(Y)]=f\tr(Y)+(1-f)\tr(Y)=\tr(Y)$. Solving $\mcd_{n,f}(Y)=X$ for $Y$ gives $Y=\mcd_{n,1/f}(X)$, well defined iff $f\neq0$. As $U^\intercal\Pi_b U$ is symmetric and has unit trace, \eqref{eq:classical_shadow} follows from \eqref{eq:depol_def} with $f\mapsto1/f$. Unbiasedness $\E[\hat\rho]=\rho_{\mathrm{sym}}$ is immediate from $\mcm^{-1}\mcm=\mathrm{id}$ on the visible space.
\end{proof}

Since $\beta\le d$ we have $f(\channel)\le f(\id)=2/(d+2)$, so noise can only decrease the depolarizing parameter. That it can never \emph{improve} the protocol is a statement about the variance rather than about $f$, and is established in \eqref{eq:var_monotone}. Writing $f_{\mbu}(\channel)=\dfrac{\beta-1}{(d-1)(d+1)}$ for the corresponding global unitary parameter~\cite{koh2022classical}, we record the exact, noise-independent relation
\begin{equation}\label{eq:fO_fU}
\frac{f(\channel)}{f_{\mbu}(\channel)}
=\frac{2(\beta-1)/[(d-1)(d+2)]}{(\beta-1)/[(d-1)(d+1)]}
=\frac{2(d+1)}{d+2}\ \xrightarrow[d\to\infty]{}\ 2 ,
\end{equation}
which underlies the factor-of-two comparison below (at $d=2$ it equals $3/2$, matching the local single-qubit ratio of \cref{sec:local}). The ratio is in fact free of $\alpha$ as well as of $\channel$: the general forms are $f(\channel)=2(\beta-\alpha/d)/[(d-1)(d+2)]$ from \eqref{eq:global_general} and $f_{\mbu}(\channel)=(\beta-\alpha/d)/(d^2-1)$ from~\cite{koh2022classical}, and the common factor $(\beta-\alpha/d)$ cancels. So a single calibration of the noise serves both ensembles, without assuming the channel is trace-preserving or unital.

Suppose the snapshots are post-processed with the noiseless inverse $\mcm_{\mbo,\id}^{-1}$ while the true channel is $\channel$, the ``noise-blind'' protocol. That estimator is then biased by the multiplicative factor $f(\channel)/f(\id)=(\beta-1)/(d-1)$ computed in \cref{rem:blind}.

\begin{remark}[Cost of noise-blind post-processing]\label{rem:blind}
By \cref{lem:reduction} the estimator carries information only through the traceless symmetric part, and there $\mcm_{\mbo,\channel}(O_0)=f(\channel)\,O_0$, so applying $\mcm_{\mbo,\id}^{-1}$ instead of $\mcm_{\mbo,\channel}^{-1}$ yields
\begin{equation}\label{eq:blind_bias}
\E\big[\hat o_{\mathrm{blind}}\big]
=\frac{f(\channel)}{f(\id)}\,\tr(O_0\rho)
=\frac{\beta-1}{d-1}\;\tr(O_0\rho).
\end{equation}
The bias is multiplicative, scaling with the signal, so it does not average away with more shots and cannot be detected by checking self-consistency of repeated runs; the left panel of \cref{fig:noise_blind_monotone} runs the biased protocol and measures exactly this factor. For depolarizing noise $\beta-1=p(d-1)$ and the factor is exactly $p$. For amplitude damping it is $\big((1+p)^n-1\big)/(d-1)$. For any $\channel$ with $\beta=d$ (\cref{sec:examples}) it is $1$ and noise-blind post-processing is already correct.
\end{remark}

\noindent Only the ratio $f(\channel)/f(\id)$ enters, so \eqref{eq:blind_bias} is the exact price of ignoring known noise, for every trace-preserving or unital $\channel$. What happens when $\channel$ is not known, and must be calibrated from data as in the robust shadow estimation of~\cite{chen2021robust}, lies outside the known-noise setting of this paper and of~\cite{koh2022classical}. We return to it in \cref{sec:discussion}.

\subsection{Shadow seminorm, bounds, and sample complexity}
\label{sec:shadow_norm}

We bound the variance of $\hat o=\tr(O\hat\rho)$ by the shadow seminorm of \cref{lem:seminorm_props}, taken on the traceless part $O_0=O-\tr(O)\id/d$, where the seminorm is defined exactly as in~\cite{koh2022classical} with the unitary average replaced by the orthogonal one,
\begin{equation}\label{eq:seminorm_def}
\norm{O_0}^2_{\shadow,\mbo,\channel}=\max_{\sigma\in\Density_d}\ \E_{U\sim\mbo}\sum_{b\in\B^n}\bra b\channel(U\sigma U^\intercal)\ket b\Big(\bra b U\,\mcm_{\mbo,\channel}^{-1,\dagger}(O_0)\,U^\intercal\ket b\Big)^2 .
\end{equation}
\begin{lemma}[Seminorm properties]\label{lem:seminorm_props}
$\norm{\cdot}_{\shadow,\mbo,\channel}$ defined by \eqref{eq:seminorm_def} is a seminorm on $\mcl(\Complex^d)$: it is absolutely homogeneous and satisfies the triangle inequality. It is not a norm, its kernel containing every antisymmetric operator.
\end{lemma}
\begin{proof}
For a state $\sigma$ put
\begin{equation}\label{eq:Phisigma}
\Phi_\sigma(X)=\E_{U}\sum_b\bra b\channel(U\sigma U^\intercal)\ket b\big(\bra bU\mcm^{-1,\dagger}_{\mbo,\channel}(X)U^\intercal\ket b\big)^2 ,\qquad \norm{X}^2_{\shadow,\mbo,\channel}=\max_{\sigma\in\Density_d}\Phi_\sigma(X).
\end{equation}
\emph{Step 1: each $\Phi_\sigma$ is a non-negative quadratic form.} Because $\channel$ is completely positive (\ref{it:I_noise}), $\channel(U\sigma U^\intercal)\succeq0$ and the weights $\bra b\channel(U\sigma U^\intercal)\ket b$ are non-negative; they do not depend on $X$, and $X\mapsto\bra bU\mcm^{-1,\dagger}(X)U^\intercal\ket b$ is linear. For a linear $\channel$ that is not positive the weights can be negative, and $\sqrt{\Phi_\sigma}$ is then undefined. Hence $\sqrt{\Phi_\sigma}$ is a seminorm, by Cauchy--Schwarz.
\emph{Step 2: both properties survive the maximum.} Absolute homogeneity is immediate, $\Phi_\sigma(cX)=|c|^2\Phi_\sigma(X)$ for every $\sigma$. For the triangle inequality, for each $\sigma$
\begin{equation}
\sqrt{\Phi_\sigma(X+Y)}\;\le\;\sqrt{\Phi_\sigma(X)}+\sqrt{\Phi_\sigma(Y)}\;\le\;\max_{\sigma'}\sqrt{\Phi_{\sigma'}(X)}+\max_{\sigma'}\sqrt{\Phi_{\sigma'}(Y)} ,
\end{equation}
and the right-hand side no longer depends on $\sigma$, so it bounds the maximum of the left, which is the claim. The maximum is attained: $\Phi_\sigma(X)$ is linear in $\sigma$, hence continuous, and $\Density_d$ is compact.
\emph{Step 3: the kernel.} By \cref{lem:reduction}(iii) $\mcm^{-1,\dagger}_{\mbo,\channel}$ annihilates $\mcv^\perp$, which in \cref{set:I} contains every antisymmetric operator. Such an $X$ therefore has $\norm{X}_{\shadow,\mbo,\channel}=0$ and the map is not a norm. On the symmetric subspace it \emph{is} a norm whenever $\channel$ is unital. Take $\sigma=\id/d$, for which every outcome weight is $\bra b\channel(\id/d)\ket b=1/d>0$, so $\Phi_{\id/d}(X)=0$ forces $\bra bU Y U^\intercal\ket b=0$ for almost every $U\in\mbo(d)$ and every $b$, where $Y=\mcm^{-1,\dagger}_{\mbo,\channel}(X)$. Given any unit eigenvector $y$ of the symmetric $Y$ with eigenvalue $\lambda$, choose $U$ with $U^\intercal\ket b=y$, then $\bra bUYU^\intercal\ket b=\lambda$, so every eigenvalue vanishes and $Y=0$, hence $X=0$. For non-unital $\channel$ some weight may vanish and we fall back on the non-degeneracy condition of~\cite{koh2022classical}. Only the seminorm property is used below.
\end{proof}

That \eqref{eq:seminorm_def} is a seminorm says nothing about its value. The orthogonal
Weingarten calculus of \cref{app:brauer} evaluates it exactly:

\begin{restatable}[Global shadow seminorm]{proposition}{propseminorm}
\label{prop:shadow_norm}
Let $\mbo$ be an orthogonal $3$-design, let $\channel$ be a linear superoperator with $d\beta\neq\alpha$, and let $O_0$ be a traceless symmetric observable. Then, with $\alpha,\beta$ as in \eqref{eq:scalars},
\begin{equation}\label{eq:seminorm_general}
\norm{O_0}^2_{\shadow,\mbo,\channel}
=\frac{d(d-1)(d+2)}{(d+4)(d\beta-\alpha)^2}\left[\frac{(d+3)\alpha-4\beta}{2}\,\tr(O_0^2)+2(d\beta-\alpha)\,\norm{O_0^2}_{\mathrm{sp}}\right].
\end{equation}
If in addition $\channel$ is trace-preserving or unital, so that $\alpha=d$ and the condition reads $\beta\neq1$, this specializes to
\begin{equation}\label{eq:seminorm_exact}
\norm{O_0}^2_{\shadow,\mbo,\channel}
=\frac{(d-1)(d+2)}{(d+4)(\beta-1)}\left[\frac{d(d+3)-4\beta}{2d(\beta-1)}\,\tr(O_0^2)+2\,\norm{O_0^2}_{\mathrm{sp}}\right].
\end{equation}
\end{restatable}
Proved in \cref{app:shadow_norm_proof}. As an identity \eqref{eq:seminorm_general} holds for any linear $\channel$, the unital hypothesis being used only to set $\alpha=d$; that the quantity it evaluates is a seminorm is \cref{lem:seminorm_props}, which needs positivity as well.

\begin{restatable}[Two-sided bounds]{corollary}{corbounds}
\label{cor:seminorm_bounds}
Under the hypotheses of \cref{prop:shadow_norm} with $\channel$ trace-preserving or unital,
\begin{equation}\label{eq:seminorm_bounds}
\frac{(d-1)^2}{2(\beta-1)^2}\,\tr(O_0^2)\ \le\ \norm{O_0}^2_{\shadow,\mbo,\channel}\ \le\ \frac{5(d-1)^2}{2(\beta-1)^2}\,\tr(O^2).
\end{equation}
\end{restatable}

The upper bound is what a sample-complexity guarantee needs, and \cref{fact:mom} turns it into
one.

\begin{restatable}[Sample complexity]{corollary}{corsample}
\label{cor:sample_complexity}
Let $\{O_i\}_{i=1}^M$ be observables and let $\channel$ be trace-preserving or unital with $\beta\neq1$, as in \cref{cor:seminorm_bounds}, whose bound the proof inserts. To estimate all $\tr(O_i\rho)$ to additive error $\eps$ with failure probability $\delta$, orthogonal shadows of total size
\begin{equation}\label{eq:sample_complexity}
N_{\mathrm{tot}}\le \frac{170\,(d-1)^2\,\log(2M/\delta)}{(\beta-1)^2\,\eps^2}\,\max_{1\le i\le M}\tr(O_i^2)
\end{equation}
suffice. The corresponding global unitary bound of~\cite{koh2022classical} has prefactor $204$ in place of $170$.
\end{restatable}
Both corollaries are proved in \cref{app:seminorm_bounds}, and both prefactors come from \cref{fact:mom} applied to the respective seminorm bounds of \cref{cor:seminorm_bounds}.

\subsection{The factor of two, exactly, and what the bound $170/204$ does and does not say}
\label{sec:comparison}

The exact variance of the estimator for a symmetric observable $A$ (traceless symmetric part $A_{\mathrm{sym};0}$) is, from \eqref{eq:seminorm_exact} with $\norm{O_0^2}_{\mathrm{sp}}\mapsto\tr(\rho O_0^2)$ and subtracting the squared mean,
\begin{multline}\label{eq:varO}
\Var_{\mbo}(\hat a)=\frac{(d-1)(d+2)}{(d+4)(\beta-1)}\left[\frac{d(d+3)-4\beta}{2d(\beta-1)}\tr(A_{\mathrm{sym};0}^2)+2\tr(\rho A_{\mathrm{sym};0}^2)\right]-\tr(A_{\mathrm{sym};0}\rho)^2 ,
\end{multline}
At $\beta=d$ this collapses to the noiseless orthogonal variance $\tfrac{d+2}{2d+8}\big[\tr(A_{\mathrm{sym};0}^2)+4\tr(\rho A_{\mathrm{sym};0}^2)\big]-\tr(A_{\mathrm{sym};0}\rho)^2$ of~\cite{west2025real}, which fixes every constant above. It is to be compared with the unitary result of~\cite{koh2022classical}, which likewise needs only a unitary $3$-design and so is realized by the multiqubit Clifford group~\cite{zhu2017multiqubit},
\begin{multline}\label{eq:varU}
\Var_{\mbu}(\hat a)=\frac{d^2-1}{(d+2)(\beta-1)}\left[\frac{d+d^2-2\beta}{d(\beta-1)}\tr(A_0^2)+2\tr(\rho A_0^2)\right]-\tr(A_0\rho)^2 .
\end{multline}
For symmetric $A$ we have $A_{\mathrm{sym};0}=A_0$. When the $O_0$-only terms dominate the $\rho$-weighted ones,
\begin{equation}\label{eq:factor2}
\mathbb E_{\mbu}[\hat a^2]\;\sim\;\frac{d^2}{(\beta-1)^2}\tr(A_0^2),
\qquad
\mathbb E_{\mbo}[\hat a^2]\;\sim\;\frac{d^2}{2(\beta-1)^2}\tr(A_0^2),
\end{equation}
so the second-moment ratio tends to $2$, as in the noiseless case~\cite{west2025real}. That is an asymptotic statement, and both the condition under which the $O_0$-only terms dominate, and the exact ratio at finite $d$ follow from a stronger result, which we therefore prove first.

Both variances have the same shape once we abbreviate the three scalars they depend on. Fix a state $\rho$ and a traceless symmetric observable $O_0$, and set
\begin{equation}\label{eq:tr_scalars}
t=\tr(O_0^2)=\norm{O_0}_2^2,\qquad r=\tr(\rho\,O_0^2),\qquad m=\tr(O_0\rho),
\end{equation}
so that
\begin{equation}\label{eq:var_shape}
\Var_{\mbo}[\hat o]+m^2=C_{\mbo}\big(a_{\mbo}t+2r\big),\qquad
\Var_{\mbu}[\hat o]+m^2=C_{\mbu}\big(a_{\mbu}t+2r\big),
\end{equation}
with
\begin{equation}\label{eq:CaCa}
C_{\mbo}=\frac{(d-1)(d+2)}{(d+4)(\beta-1)},\quad
a_{\mbo}=\frac{d(d+3)-4\beta}{2d(\beta-1)},\quad
C_{\mbu}=\frac{d^2-1}{(d+2)(\beta-1)},\quad
a_{\mbu}=\frac{d+d^2-2\beta}{d(\beta-1)} .
\end{equation}
Since the estimator is unbiased in both ensembles, $m$ is ensemble-independent and the left-hand sides of \eqref{eq:var_shape} are the single-shot \emph{second} moments $\mathbb E[\hat o^2]$. This shape settles the claim of \cref{sec:global_channel} that noise never improves the protocol, which the bound $f(\channel)\le f(\id)$ recorded there does not: that bound constrains the depolarizing parameter, and the protocol is judged by its variance. The argument is that more noise means a smaller $\beta$ (\cref{claim:beta_range}, with $\beta=d$ the noiseless value), that the single-shot second moment is strictly decreasing in $\beta$, and hence that any channel lowering $\beta$ strictly raises it. Differentiating the orthogonal case,
\begin{equation}\label{eq:var_monotone}
\frac{\partial}{\partial\beta}\,\mathbb E_{\mbo}[\hat o^2]\;=\;-\,\frac{(d-1)(d+2)\big[\,2dr(\beta-1)+t\big(d^2+3d-2\beta-2\big)\big]}{d\,(\beta-1)^3(d+4)}\;<\;0 ,
\end{equation}
since on $1<\beta\le d$ the bracket is positive: $2dr(\beta-1)\ge0$, and $d^2+3d-2\beta-2\ge d^2+d-2=(d-1)(d+2)>0$ at the largest admissible $\beta=d$, growing as $\beta$ falls. The second moment is therefore strictly decreasing in $\beta$, so any noise that lowers $\beta$ strictly raises it. The same computation with $C_{\mbu}$ and $a_{\mbu}$ gives the unitary statement,
\begin{equation}\label{eq:var_monotone_U}
\frac{\partial}{\partial\beta}\,\mathbb E_{\mbu}[\hat o^2]\;=\;-\,\frac{2(d^2-1)\big[\,t\big(d^2+d-\beta-1\big)+d\,r(\beta-1)\big]}{d\,(d+2)(\beta-1)^3}\;<\;0 ,
\end{equation}
negative on $1<\beta\le d$ for the same reason: $d\,r(\beta-1)\ge0$, and $d^2+d-\beta-1\ge d^2-1>0$ at the largest admissible $\beta=d$. Neither ensemble is helped by noise. Their ratio turns out to depend on $\rho$ and $O_0$ through a single dimensionless combination.

\noindent Both panels of \cref{fig:noise_blind_monotone} measure what this subsection and \cref{rem:blind} derive; the right-hand one is the monotonicity above, seen directly.

\begin{proposition}[Exact advantage criterion]
\label{prop:criterion}
Let $\channel$ be trace-preserving or unital with $\beta\neq1$, let $O_0$ be symmetric and traceless with $r>0$, and set
\begin{equation}\label{eq:kappa_def}
\kappa\;=\;\frac{a_{\mbo}t}{2r}\;=\;\frac{\big(d(d+3)-4\beta\big)\,\tr(O_0^2)}{4d(\beta-1)\,\tr(\rho O_0^2)} .
\end{equation}
Then the single-shot second moments obey exactly
\begin{equation}\label{eq:ratio_identity}
\frac{\mathbb E_{\mbu}[\hat o^2]}{\mathbb E_{\mbo}[\hat o^2]}
=\frac{\varrho_L\,\kappa+\varrho_S}{\kappa+1},
\qquad
\varrho_S=\frac{(d+1)(d+4)}{(d+2)^2},\qquad
\varrho_L=\frac{2(d+1)(d+4)\big(d^2+d-2\beta\big)}{(d+2)^2\big(d^2+3d-4\beta\big)} .
\end{equation}
The right-hand side is strictly increasing in $\kappa$ for every $d\ge2$, and interpolates between $\varrho_S$ (as $\kappa\to0$) and $\varrho_L$ (as $\kappa\to\infty$). In the noiseless limit $\beta=d$ we have $\varrho_L=2\varrho_S$ exactly, and for large $d$, $\varrho_S\to1$ and $\varrho_L\to2$, so that
\begin{equation}\label{eq:ratio_master}
\frac{\mathbb E_{\mbu}[\hat o^2]}{\mathbb E_{\mbo}[\hat o^2]}\;\xrightarrow[\ d\to\infty\ ]{}\;\frac{2\kappa+1}{\kappa+1}=2-\frac{1}{1+\kappa}.
\end{equation}
Two readings of \eqref{eq:ratio_master} have to be resisted. It is a ratio of \emph{second} moments, not of variances --- \cref{cor:var_ratio} converts it --- and the limit it displays is the one in $d$, taken at fixed $\kappa$. In particular the value is not $2$ in the noiseless case: at $\beta=d$ and any finite $\kappa$ it is still $2-1/(1+\kappa)$, since the noise has already left through $\varrho_L\to2$, and the remaining shortfall is a property of the observable. The factor of two is the $\kappa\to\infty$ corner, not the noiseless one. At finite $d$ the limiting value $\varrho_L$ is strictly increasing in $\beta$, and $\varrho_L>2$ if and only if
\begin{equation}\label{eq:betastar}
\beta\;>\;\beta^{*}(d)\;=\;\frac{d\,(d^2+7d+8)}{2\,(d^2+3d+4)} ,
\end{equation}
a threshold that always lies strictly below $d$.
\end{proposition}

\begin{proof}
Divide numerator and denominator of the ratio of \eqref{eq:var_shape} by $2rC_{\mbo}$:
\begin{equation}
\frac{C_{\mbu}(a_{\mbu}t+2r)}{C_{\mbo}(a_{\mbo}t+2r)}
=\frac{C_{\mbu}}{C_{\mbo}}\cdot\frac{\tfrac{a_{\mbu}t}{2r}+1}{\tfrac{a_{\mbo}t}{2r}+1}
=\frac{C_{\mbu}}{C_{\mbo}}\cdot\frac{\tfrac{a_{\mbu}}{a_{\mbo}}\kappa+1}{\kappa+1},
\end{equation}
where the last step used $\tfrac{a_{\mbu}t}{2r}=\tfrac{a_{\mbu}}{a_{\mbo}}\cdot\tfrac{a_{\mbo}t}{2r}=\tfrac{a_{\mbu}}{a_{\mbo}}\kappa$. Setting $\varrho_S=C_{\mbu}/C_{\mbo}$ and $\varrho_L=C_{\mbu}a_{\mbu}/(C_{\mbo}a_{\mbo})$ gives \eqref{eq:ratio_identity}. Explicitly, the $(\beta-1)$ factors cancel:
\begin{equation}
\varrho_S=\frac{d^2-1}{(d+2)(\beta-1)}\cdot\frac{(d+4)(\beta-1)}{(d-1)(d+2)}=\frac{(d+1)(d+4)}{(d+2)^2},
\end{equation}
and, using $C_{\mbu}a_{\mbu}=\tfrac{(d^2-1)(d^2+d-2\beta)}{d(d+2)(\beta-1)^2}$ and $C_{\mbo}a_{\mbo}=\tfrac{(d-1)(d+2)(d^2+3d-4\beta)}{2d(d+4)(\beta-1)^2}$,
\begin{equation}
\varrho_L=\frac{(d^2-1)(d^2+d-2\beta)}{d(d+2)}\cdot\frac{2d(d+4)}{(d-1)(d+2)(d^2+3d-4\beta)}
=\frac{2(d+1)(d+4)(d^2+d-2\beta)}{(d+2)^2(d^2+3d-4\beta)} .
\end{equation}
Differentiating \eqref{eq:ratio_identity}, $\partial_\kappa[\,\cdot\,]=(\varrho_L-\varrho_S)/(\kappa+1)^2$, so monotonicity is equivalent to $\varrho_L>\varrho_S$, i.e.\ to $2(d^2+d-2\beta)>d^2+3d-4\beta$, i.e.\ to $d^2-d>0$, which holds for $d\ge2$. Setting $\beta=d$ gives $d^2+d-2\beta=d^2-d$ and $d^2+3d-4\beta=d^2-d$, whence $\varrho_L=2\varrho_S$. Finally $\varrho_S=1+d/(d+2)^2\to1$ and, since $\beta\le d$ is $O(d)$, $\varrho_L\to2$.

For \eqref{eq:betastar}, note first that $\beta\le d$ and $d\ge2$ give $4\beta\le4d<d^2+3d$, so the denominator of $\varrho_L$ is positive. Clearing it,
\begin{equation}\label{eq:rhoL_minus2}
\varrho_L-2=\frac{2\big[2\beta(d^2+3d+4)-d(d^2+7d+8)\big]}{(d+2)^2(d^2+3d-4\beta)} ,
\end{equation}
whose numerator is strictly increasing in $\beta$ and vanishes at $\beta=\beta^{*}(d)$. Since $\beta^{*}(d)/d=(d^2+7d+8)/[2(d^2+3d+4)]<1$, the threshold always lies inside the admissible range $\beta\in(1,d]$, and $\varrho_L>2$ if and only if $\beta>\beta^{*}(d)$.
\end{proof}

\eqref{eq:ratio_identity} compares second moments. The operationally relevant quantity is the ratio of \emph{variances}, and the two differ by a term that is controlled by the same parameter $\kappa$.
\begin{corollary}[From second moments to variances]\label{cor:var_ratio}
Under the hypotheses of \cref{prop:criterion} write $R=\mathbb E_{\mbu}[\hat o^2]/\mathbb E_{\mbo}[\hat o^2]$ for the second-moment ratio and $x=m^2/\mathbb E_{\mbo}[\hat o^2]\in[0,1)$. Then
\begin{equation}\label{eq:var_from_moment}
\frac{\Var_{\mbu}[\hat o]}{\Var_{\mbo}[\hat o]}=\frac{R-x}{1-x}\;\ge\;R,\qquad\text{with}\qquad x\;\le\;\frac{1}{2\,C_{\mbo}\,(\kappa+1)} ,
\end{equation}
and equality in the bound exactly when $\rho$ is an eigenstate of $O_0$. Hence the variance ratio exceeds the second-moment ratio, and the two agree to $O(1/\kappa)$: in particular they share the limit $2$ as $\kappa\to\infty$, $d\to\infty$, while for bounded $\kappa$ their limits differ.
\end{corollary}
\begin{proof}
Both estimators are unbiased, so $m=\tr(O_0\rho)$ is common and $\Var=\mathbb E[\hat o^2]-m^2$. Dividing numerator and denominator by $\mathbb E_{\mbo}[\hat o^2]$ gives the identity, and $R\ge1$ makes $\tfrac{R-x}{1-x}-R=\tfrac{x(R-1)}{1-x}\ge0$. For the bound, non-negativity of the variance of $O_0$ in the state $\rho$ gives $m^2=\tr(O_0\rho)^2\le\tr(\rho O_0^2)=r$, with equality iff $\rho$ is an eigenstate of $O_0$. Then, by \eqref{eq:var_shape} and \eqref{eq:kappa_def},
\begin{equation}
x=\frac{m^2}{C_{\mbo}(a_{\mbo}t+2r)}\;\le\;\frac{r}{C_{\mbo}(a_{\mbo}t+2r)}=\frac{1}{C_{\mbo}\big(2\kappa+2\big)} .
\end{equation}
\end{proof}

\begin{remark}[The comparison survives noise-blind post-processing]\label{rem:blind_ratio}
The proof of \cref{cor:var_ratio} invokes unbiasedness only to make $m=\tr(O_0\rho)$ common to the two ensembles, and that conclusion holds under a weaker hypothesis than unbiasedness. By \eqref{eq:global_depol} the shadow channel acts on the visible traceless subspace as the scalar $f(\channel)$, so $\mcm_{\mbo,\channel}^{-1}$ acts there as $f(\channel)^{-1}$ and noise-blind post-processing multiplies the estimator by the constant $f(\channel)/f(\id)$ of \eqref{eq:blind_bias} \emph{shot by shot}, not merely in mean --- which is what lets it cancel from a ratio. By \eqref{eq:fO_fU} the ratio $f_{\mbo}/f_{\mbu}=2(d+1)/(d+2)$ does not depend on the channel, so it cancels between $\channel$ and $\id$ and the blind factor is the \emph{same} $(\beta-1)/(d-1)$ in both ensembles. A common deterministic rescaling $\hat o\mapsto\kappa\hat o$ sends $m\mapsto\kappa m$ and $\mathbb E[\hat o^2]\mapsto\kappa^2\,\mathbb E[\hat o^2]$, so both $R$ and $x$ of \eqref{eq:var_from_moment} are unchanged, and with them the variance ratio itself. Because the two ensembles are biased identically rather than unbiased, the factor-of-two comparison does not depend on the post-processing being noise-aware. What the argument does require is that \emph{both} ensembles invert with the same assumed channel, since assuming $\channel_{\mbo}$ for one and $\channel_{\mbu}$ for the other rescales $R$ by $(\kappa_{\mbu}/\kappa_{\mbo})^{2}$ and the cancellation fails.
\end{remark}
Equation \eqref{eq:factor2} is thus the $\kappa\to\infty$, $d\to\infty$ corner of \eqref{eq:ratio_identity}, and both limits are needed. The finite-$d$ ceiling $\varrho_L$ can sit on either side of $2$, and which side is a question about the noise rather than about the dimension, since by \eqref{eq:betastar} it exceeds $2$ exactly when $\beta>\beta^{*}(d)$, which for depolarizing noise reads $p>\tfrac34$ at $d=4$ and $p>0.65$ at $d=8$, relaxing to $p>\tfrac12$ as $d\to\infty$. At $d=4$ and $p=0.9$, for instance, $\varrho_L=2.121$. At the same dimension and $p=0.5$ it is $1.852$. The factor-of-two advantage is in either case robust to noise, and the ratio climbs monotonically toward $2$ with $n$ (\cref{fig:variance_ratio}).

\begin{figure}[t]
\centering
\includegraphics[width=0.59\linewidth]{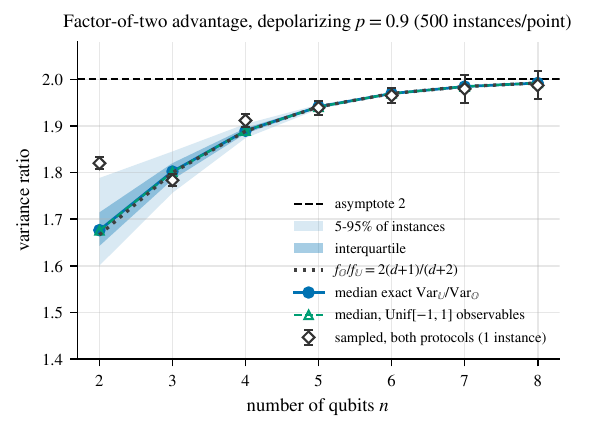}
\caption{Exact variance ratio $\Var_{\mbu}/\Var_{\mbo}$ under depolarizing noise ($p=0.9$) from \eqref{eq:varO}--\eqref{eq:varU}, over $500$ independent $(\rho,O)$ instances per point: median (circles), interquartile and $5$--$95\%$ bands. Observables are random symmetric matrices made traceless and normalized; states are random mixed. The ratio rises monotonically toward the noise-independent asymptote $2$ (dashed) and concentrates with $n$. Triangles: medians for the observable distribution of~\cite{west2025real}. Dotted: the depolarizing-parameter ratio \eqref{eq:fO_fU}. Diamonds: the ratio measured by simulating both protocols on one instance per $n$, with $95\%$ confidence intervals; $2.5\times10^5$ shots per protocol for $n\le6$ and $6.8\times10^4$ for $n=7,8$, the confidence interval at fixed shot count being independent of $n$.}
\label{fig:variance_ratio}
\end{figure}

By contrast, the ratio $204/170=1.2$ of the \emph{sample-complexity prefactors} in \cref{cor:sample_complexity} is a comparison of loose upper bounds, obtained after the relaxation $\norm{O_0^2}_{\mathrm{sp}}\le\tr(O_0^2)$ that discards precisely the structure responsible for the factor of two. The ratio $1.2$ is therefore a conservative worst-case guarantee, not the typical improvement, and the two numbers are easily interchanged. The primary result is the exact factor of two of \eqref{eq:factor2}.

Neither ratio is the operationally meaningful quantity; that is the number of shots required to reach a target $(\varepsilon,\delta)$. We therefore ran the median-of-means estimator of \cref{cor:sample_complexity} directly (\cref{fig:sample_complexity}; $n=3$, $\varepsilon=0.1$, $\delta=0.05$, $M=1$, so $K=8$). At the prescribed batch size the measured failure rate is $0/1200$ for both ensembles, so the guarantee holds with a wide margin. It is also loose by a factor of about $40$ in shots, the smallest $N$ meeting $(\varepsilon,\delta)$ being $68$ against a prescribed $2721$ for the orthogonal protocol and $122$ against $4910$ for the unitary one. Here the prescribed counts are $\lceil 34\,\Var/\varepsilon^2\rceil$ evaluated at the \emph{exact} variances of \eqref{eq:varO} and \eqref{eq:varU}, not at the seminorm relaxation ---so the factor of about $40$ measures the looseness of the Chebyshev-type constant $34$ and the median-of-means argument alone, and not that of the norm bound. Chaining \eqref{eq:sample_complexity} in full, which does relax the variance to the seminorm, is looser again. Both are loose by the same factor, so the measured shot ratio $122/68=1.79$ tracks the exact variance ratio $\Var_{\mbu}/\Var_{\mbo}=1.805$ and not the prefactor ratio $1.2$.

\begin{figure}[t]
\centering
\includegraphics[width=\linewidth]{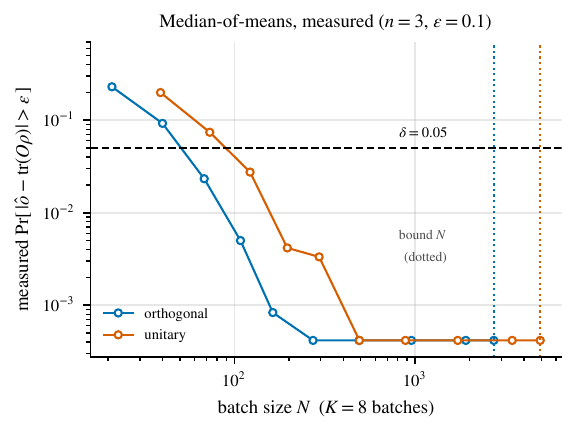}
\caption{The median-of-means estimator of \cref{cor:sample_complexity}, run: measured $\Pr[\lvert\hat o-\tr(O\rho)\rvert>\varepsilon]$ against batch size $N$ at fixed $K=8$, for $\varepsilon=0.1$, $\delta=0.05$, $M=1$, $n=3$ and the unit-$2$-norm transverse-field Ising energy; $1200$ repetitions per point, the floor being zero observed failures. Dotted verticals mark the $N$ the guarantee prescribes. The guarantee holds with wide margin and is loose by a factor of about $40$ in $N$ for both ensembles.}
\label{fig:sample_complexity}
\end{figure}

\begin{figure}[t]
\centering
\includegraphics[width=\linewidth]{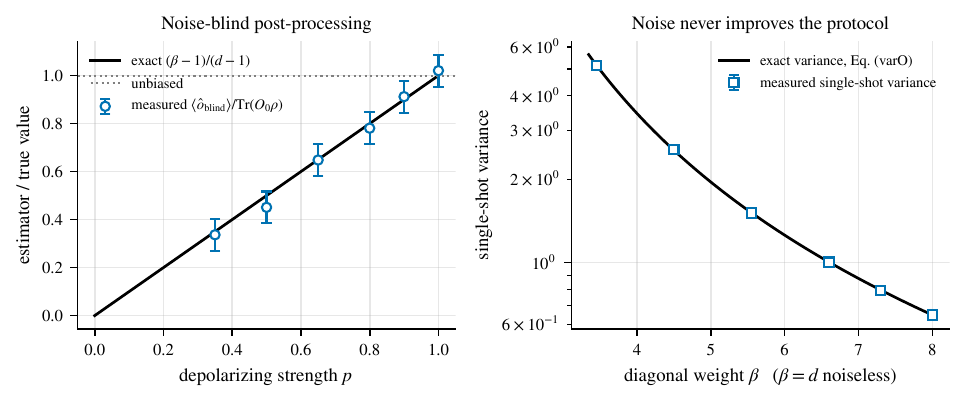}
\caption{Two statements of this Part, run rather than only derived ($n=3$, $4\times10^5$ shots per point, $95\%$ confidence intervals). \emph{Left:} the noise-blind protocol of \cref{rem:blind}, in which snapshots produced under $\mcd_{n,p}$ are post-processed with the \emph{noiseless} inverse $\mcm_{\mbo,\id}^{-1}$. The measured ratio $\E[\hat o_{\mathrm{blind}}]/\tr(O_0\rho)$ follows the predicted multiplicative bias $(\beta-1)/(d-1)$ of \eqref{eq:blind_bias}, which for depolarizing noise is exactly $p$. Because the bias multiplies the signal rather than shifting it, the line passes through the origin and no amount of averaging removes it. \emph{Right:} the same protocol post-processed correctly. The measured single-shot variance tracks \eqref{eq:varO} and rises monotonically as $\beta$ falls from its noiseless value $d=8$, by a factor $8.0$ down to $\beta=3.45$, which is \eqref{eq:var_monotone} in the form one would observe. All twelve measurements sit within $1.4$ of their own confidence interval.}
\label{fig:noise_blind_monotone}
\end{figure}

\subsection{Noise models}
\label{sec:examples}
We specialize the global results to standard channels (real basis; derivations for depolarizing, amplitude damping and dephasing in \cref{app:examples}, and for the coherent and readout models inline below).

\subsubsection{Depolarizing noise}
For $\channel=\mcd_{n,p}$ we find $\alpha=d$ and $\beta=pd+1-p$, hence
\begin{equation}\label{eq:depol_f}
f(\mcd_{n,p})=\frac{2\big(pd+1-p-1\big)}{(d-1)(d+2)}=\frac{2p}{d+2},
\end{equation}
so that $\mcm_{\mbo,\mcd_{n,p}}=\mcd_{n,2p/(d+2)}\circ(\cdot)_{\mathrm{sym}}$. This is simply the noiseless channel with $f$ scaled by $p$. Since $\beta-1=p(d-1)$, the factor $(d-1)^2/(\beta-1)^2=1/p^2$ in \cref{cor:sample_complexity} gives the clean bound
\begin{equation}
N_{\mathrm{tot}}\le\frac{170\,\log(2M/\delta)}{p^2\,\eps^2}\,\max_i\tr(O_i^2),
\end{equation}
mirroring the depolarizing scaling of the unitary case ($204/p^2$) with the improved prefactor.

\subsubsection{Amplitude damping}
For $\channel=\mathrm{AD}_{n,p}=\mathrm{AD}_{1,p}^{\otimes n}$ with $\bra0\mathrm{AD}_{1,p}(\Pi_0)\ket0=1$ and $\bra1\mathrm{AD}_{1,p}(\Pi_1)\ket1=p$, the diagonal weight is multiplicative,
\begin{equation}\label{eq:ad_beta}
\beta=\tr[\mathrm{AD}_{n,p}\circ\diag]=\sum_{b\in\B^n}\prod_{i}\bra{b_i}\mathrm{AD}_{1,p}(\Pi_{b_i})\ket{b_i}=\sum_{b}p^{\,\ell_1(b)}=(1+p)^n ,
\end{equation}
where $\ell_1(b)$ is the Hamming weight. Hence $f(\mathrm{AD}_{n,p})=\tfrac{2((1+p)^n-1)}{(d-1)(d+2)}$ and
\begin{equation}
N_{\mathrm{tot}}\le\frac{170\,(2^n-1)^2\log(2M/\delta)}{\big((1+p)^n-1\big)^2\,\eps^2}\max_i\tr(O_i^2).
\end{equation}
On the computational basis $\Pi_b^\intercal=\Pi_b$, so $\tilde\beta=\beta=(1+p)^n$. A complex basis gives a basis-angle-dependent $\tilde\beta$ (\cref{app:examples}). \cref{fig:noise_dependence} compares the two channels, showing that depolarizing noise degrades $f$ linearly in $p$, whereas amplitude damping degrades it through $(1+p)^n-1$, and the right panel shows the resulting sample-complexity blow-up $\propto f^{-2}$.

\begin{figure}[t]
\centering
\includegraphics[width=\linewidth]{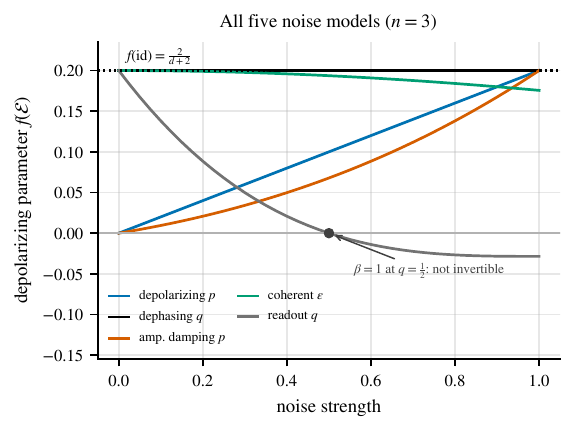}
\caption{The depolarizing parameter $f(\channel)=\tfrac{2(\beta-1)}{(d-1)(d+2)}$ for all five noise models of \cref{sec:examples} ($n=3$), against the noiseless value $f(\id)=2/(d+2)$ (dotted). Because dephasing is diagonal-preserving, $\beta=d$ at every strength and $f$ is flat at $f(\id)$. Readout bit-flip at rate $q$ has $\beta=d(1-q)^n$, which vanishes to $\beta=1$ at $q=\tfrac12$ for every $n$; past that point $f<0$ and the inverse carries a sign flip.}
\label{fig:noise_dependence}
\end{figure}

\subsubsection{Dephasing noise is inconsequential}
For the completely dephasing channel $\channel=\diag$ we have $\beta=\tr[\diag\circ\diag]=\tr[\diag]=d$, so $f(\diag)=f(\id)=2/(d+2)$ and $\mcm_{\mbo,\diag}=\mcm_{\mbo,\id}$, so dephasing before a computational-basis measurement leaves the shadow channel unchanged. More generally, any $\channel$ with $\beta=d$ is inconsequential, exactly as in the unitary case~\cite{koh2022classical}.

\subsubsection{Coherent over-rotations}
\label{sec:coherent}
A coherent control error is a unitary channel $\channel(A)=VAV^\dagger$ with $V\in\mathrm{U}(d)$ not necessarily real, and its Kraus operator is complex, so this is the natural test of whether the real structure survives non-conjugation-symmetric noise. Being unitary, $\channel$ is trace-preserving and unital, hence $\alpha=d$, and \cref{prop:global_noisy_channel} applies verbatim: the shadow channel is $\mcd_{n,f}\circ(\cdot)_{\mathrm{sym}}$ with the entire noise dependence in
\begin{equation}\label{eq:coherent_beta}
\beta=\sum_b\bra b\channel(\Pi_b)\ket b=\sum_b\bra b V\Pi_bV^\dagger\ket b=\sum_b\lvert V_{bb}\rvert^2 .
\end{equation}
Since each column of $V$ is a unit vector, $\lvert V_{bb}\rvert\le1$ and $\beta\le d$, with $\beta=d$ (the inconsequential case) precisely when $V$ is diagonal, i.e.\ a phase error in the measurement basis---consistent with dephasing. Write $V=e^{i\varepsilon H}$ for Hermitian $H$ and expand $V_{bb}=1+i\varepsilon H_{bb}-\tfrac{\varepsilon^2}{2}(H^2)_{bb}+O(\varepsilon^3)$. Hermiticity enters twice: it makes $H_{bb}$ and $(H^2)_{bb}$ real, so the expansion splits into a real part $1-\tfrac{\varepsilon^2}{2}(H^2)_{bb}$ and an imaginary part $\varepsilon H_{bb}$, and it gives the row-norm identity $(H^2)_{bb}=\sum_{b'}H_{bb'}H_{b'b}=\sum_{b'}\lvert H_{bb'}\rvert^2$ from $H_{b'b}=\overline{H_{bb'}}$. Taking the modulus squared,
\begin{equation}\label{eq:coherent_expansion}
\lvert V_{bb}\rvert^2=1-\varepsilon^2\big[(H^2)_{bb}-H_{bb}^2\big]+O(\varepsilon^3)
=1-\varepsilon^2\!\!\sum_{b'\neq b}\lvert H_{bb'}\rvert^2+O(\varepsilon^3),
\end{equation}
using $(H^2)_{bb}=\sum_{b'}\lvert H_{bb'}\rvert^2$. Summing over $b$,
\begin{equation}\label{eq:coherent_dminusbeta}
d-\beta=\varepsilon^2\!\!\sum_{b\neq b'}\lvert H_{bb'}\rvert^2+O(\varepsilon^3)=\varepsilon^2\,\lVert H_{\mathrm{off}}\rVert_{\mathrm F}^2+O(\varepsilon^3),
\end{equation}
where $H_{\mathrm{off}}$ is the off-diagonal part of $H$. Only the basis-rotating (off-diagonal) part of the generator degrades the shadow, since a generator diagonal in the measurement basis leaves $\beta=d$ and the protocol unchanged, whereas an off-diagonal coupling reduces $f$ to $\tfrac{2}{d+2}\big[1-\tfrac{\varepsilon^2}{d-1}\lVert H_{\mathrm{off}}\rVert_{\mathrm F}^2\big]+O(\varepsilon^3)$ and inflates the variance. That factor is $1+\tfrac{2\varepsilon^2}{d-1}\lVert H_{\mathrm{off}}\rVert_{\mathrm F}^2+O(\varepsilon^3)$, and it arises as follows. In \eqref{eq:seminorm_exact} the $\tr(O_0^2)$ term carries $(\beta-1)^{-2}\propto f^{-2}$ and the $\norm{O_0^2}_{\mathrm{sp}}$ term only $(\beta-1)^{-1}$, while the coefficient $a_{\mbo}$ of \eqref{eq:CaCa} depends on $\beta$ as well; so $f^{-2}$ is the leading behaviour in the regime where the $\tr(O_0^2)$ term dominates --- the same $\kappa\to\infty$ regime that carries the factor of two --- and not an exact scaling of the variance. Expanding the exact \eqref{eq:varO} to $O(\varepsilon^2)$ returns the same factor. Thus coherent errors never break the block structure. They act only through $\beta$, and are dangerous only insofar as they drive $\beta$ toward the invertibility threshold $\beta=1$.

\subsubsection{Readout error}
\label{sec:readout}
Classical measurement confusion is described by a column-stochastic matrix $R$ ($R_{b'b}\ge0$, $\sum_{b'}R_{b'b}=1$), where $R_{b'b}$ is the probability of recording $b'$ when the true outcome is $b$. As a channel applied before the ideal computational measurement, it reads out populations and relabels them,
\begin{equation}\label{eq:readout_channel}
\channel(A)=\sum_{b,b'}R_{b'b}\,\bra b A\ket b\,\Pi_{b'},\qquad K_{b'b}=\sqrt{R_{b'b}}\,\ketbra{b'}{b},
\end{equation}
and $\sum_{b,b'}K_{b'b}^\dagger K_{b'b}=\sum_b\big(\sum_{b'}R_{b'b}\big)\Pi_b=\id$ confirms trace preservation, so $\alpha=d$. Then $\channel(\Pi_b)=\sum_{b'}R_{b'b}\Pi_{b'}$ is diagonal, and
\begin{equation}\label{eq:readout_beta}
\beta=\sum_b\bra b\channel(\Pi_b)\ket b=\sum_b R_{bb}=\tr R,
\end{equation}
the expected number of correctly recorded outcomes. Readout error is inconsequential ($\beta=d$) iff $R=\id$, and otherwise degrades the shadow through the total misread weight $d-\tr R=\sum_b(1-R_{bb})$, giving $f=\tfrac{2(\tr R-1)}{(d-1)(d+2)}$. Because $\tr R\ge1$ for any detector better than fully randomizing, readout error only rescales $f$ and never destroys invertibility. This is the diagonal counterpart of the coherent case \eqref{eq:coherent_beta}, since both are unital with $\alpha=d$ and enter solely through $\beta$, one via the squared diagonal of a unitary, the other via the diagonal of a stochastic matrix.

Two consequences of \eqref{eq:readout_beta} delimit the model. First, only the diagonal of $R$ enters, so the protocol is insensitive to \emph{which} wrong outcome a misread produces: two detectors with equal per-outcome fidelities and entirely different crosstalk give the same shadow channel. Second, $R$ is $d\times d$, so $\beta=\tr R$ presumes it has been characterized, and characterizing a correlated $R$ costs exponentially many calibration experiments. Practical calibration therefore assumes an uncorrelated confusion model, and that assumption --- not the noise model of \cref{sec:noise_model} --- is where this treatment meets an experimental limit. Tensor-network calibration relaxes it: Guo and Yang~\cite{guo2026tensor} represent the readout process as a matrix product operator trained by likelihood optimization on calibration data, capture spatial correlations at a sample cost near-linear in the number of qubits, and demonstrate it up to $20$ qubits on a superconducting processor. Their construction returns the object $R$ that \eqref{eq:readout_beta} contracts to a scalar, so an MPO-calibrated $R$ feeds straight into $\beta=\tr R$ with nothing else in \cref{prop:global_noisy_channel} changed.

\section{Part II: Real Basis, Local Orthogonal Ensemble}
\label{sec:partII}

\begin{setting}[Real basis, local ensemble]\label{set:II}
\cref{set:I} with \ref{it:I_ens} replaced by: $\mbo=\mbo(2)^{\otimes n}$, or any product of single-qubit orthogonal $3$-designs such as $(\mcc_1\cap\mbo(2))^{\otimes n}$, and $\channel=\channel_1^{\otimes n}$ a product channel. \ref{it:I_basis}, \ref{it:I_noise} and \ref{it:I_obs} are unchanged, except that the observable class narrows to $\mbo(2)^{\otimes n}$-visible operators, spanned by $\{\id,X,Z\}^{\otimes n}$.
\end{setting}

\noindent We replace the global ensemble by $\mbo(2)^{\otimes n}$ and the noise by product noise $\channel=\channel_1^{\otimes n}$, keeping the real basis. Now the twirl factorizes across qubits, which replaces the symmetric visible space by the locally symmetric one and turns the polynomial seminorm into an exponential in Pauli weight. Proof in \cref{app:local_proof}.

\subsection{Reconstruction: the local shadow channel}
\label{sec:local}

For product ensembles and product noise, the analysis factorizes across qubits.

\begin{restatable}[Local noisy channel]{proposition}{proplocal}
\label{prop:local_noisy}
Let $\mbo=\mbo(2)^{\otimes n}$ with the computational basis, and $\channel=\channel_1^{\otimes n}$ a product of single-qubit channels. Then
\begin{equation}\label{eq:local_channel}
\mcm_{\mbo,\channel}(A)=\mcd_{1,f_1(\channel_1)}^{\otimes n}\big(A_{\mathrm{L.S.}}\big),\qquad
f_1(\channel_1)=\frac{\tr[\channel_1\circ\diag]-1}{2},
\end{equation}
where $A_{\mathrm{L.S.}}$ is the projection of $A$ onto locally symmetric operators, i.e.\ Pauli strings drawn from $\{\id,X,Z\}^{\otimes n}$. The classical shadow factorizes as $\hat\rho=\bigotimes_{j}\mcd_{1,1/f_1}(U_j^\intercal\Pi_{b_j}U_j)$, and the visible space is $\mathrm{span}_{\Complex}\{\id,X,Z\}^{\otimes n}$.
\end{restatable}

\subsection{Locally symmetric Pauli seminorm and sample complexity}
\label{sec:local_seminorm}

Definition \eqref{eq:seminorm_def} contains a maximization over all states of the full
$d$-dimensional system, so its factorization across qubits is not a formal consequence of
the channel factorizing: we need to know that the maximizer may be taken to be a product
state. It may, for the following reason, which we isolate because both Part~II and Part~IV
use it.

\begin{lemma}[Factorization of the maximization]\label{lem:tensormax}
Let $S_1,\dots,S_n$ be positive semidefinite operators on $\Complex^2$. Then
\begin{equation}\label{eq:tensormax}
\max_{\sigma\in\Density_{2^n}}\tr\Big[\sigma\bigotimes_{j=1}^nS_j\Big]
=\prod_{j=1}^n\lambda_{\max}(S_j).
\end{equation}
\end{lemma}
\begin{proof}
For Hermitian $A$, $\max_{\sigma\in\Density}\tr[\sigma A]=\lambda_{\max}(A)$, the maximum
being attained at a top eigenprojector. Now the spectrum of $\bigotimes_jS_j$ is the set of
products $\prod_j\mu_j$ with $\mu_j\in\spec S_j$, and all $\mu_j\ge0$ since $S_j\succeq0$;
hence the largest such product is $\prod_j\lambda_{\max}(S_j)$, attained at the tensor
product of the corresponding eigenvectors.
\end{proof}
\noindent Positivity makes the largest product the product of the largest, and it
holds here because each $S_j$ is a single-shot second moment. We write $S$ rather than $R$ because $R$ is already the readout confusion matrix of \cref{sec:readout} and the second-moment ratio of \cref{cor:var_ratio}.

\begin{restatable}[Locally symmetric Pauli seminorm]{proposition}{propPauli}
\label{prop:pauli_seminorm}
Let $\channel=\channel_1^{\otimes n}$ be a product of single-qubit channels in the sense of \cref{sec:noise_model}, so that each $\channel_1$ is completely positive with $\alpha_1=2$, and write $f_1=f_1(\channel_1)$ as in \eqref{eq:local_channel}. For a locally symmetric Pauli operator $P\in\{\id,X,Z\}^{\otimes n}$ of weight $\mathrm{wt}(P)$,
\begin{equation}\label{eq:pauli_seminorm}
\norm{P}^2_{\shadow,\mbo,\channel_1^{\otimes n}}=\Big(2\,f_1(\channel_1)^2\Big)^{-\mathrm{wt}(P)}
=\Big(\tfrac{1}{\sqrt2\,f_1(\channel_1)}\Big)^{2\,\mathrm{wt}(P)} .
\end{equation}
Consequently, for $M$ locally symmetric Pauli targets,
\begin{equation}\label{eq:local_sample}
N_{\mathrm{tot}}\le\frac{68\log(2M/\delta)}{\eps^2}\max_{1\le i\le M}\big(2f_1(\channel_1)^2\big)^{-\mathrm{wt}(P_i)}.
\end{equation}
\end{restatable}
\Cref{fig:local_advantage} validates \eqref{eq:pauli_seminorm} numerically against exact enumeration.
Because the single-qubit orthogonal and unitary parameters satisfy $f_1^{\mbo}=(\beta_1-1)/2$ and $f_1^{\mbu}=(\beta_1-1)/3$, their ratio is exactly $3/2$ independent of the noise, so the local unitary bound $(3f_1^{\mbu\,2})^{-\mathrm{wt}(P)}$ exceeds the orthogonal bound by a factor
\begin{equation}
\frac{N^{\mbu}_{\mathrm{tot}}}{N^{\mbo}_{\mathrm{tot}}}=\frac{\big(3 f_1^{\mbu\,2}\big)^{-\mathrm{wt}(P)}}{\big(2 f_1^{\mbo\,2}\big)^{-\mathrm{wt}(P)}}=\Big(\frac{2 f_1^{\mbo\,2}}{3 f_1^{\mbu\,2}}\Big)^{\mathrm{wt}(P)}=\Big(\tfrac32\Big)^{\mathrm{wt}(P)} .
\end{equation}
So the exponential-in-weight advantage of local real shadows persists verbatim under arbitrary single-qubit product noise, where the noiseless bound $2^{\mathrm{wt}(P)}$ is recovered at $\channel_1=\id$ ($f_1=1/2$).

\begin{figure}[t]
\centering
\includegraphics[width=\linewidth]{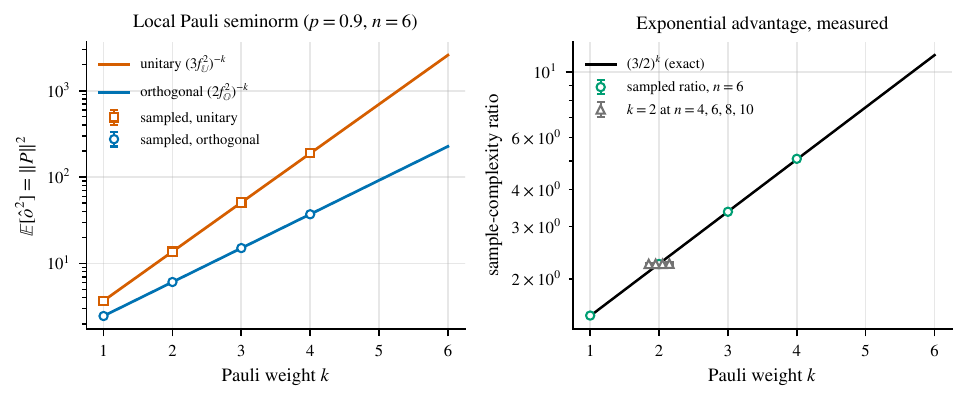}
\caption{Local Pauli shadow seminorm under single-qubit depolarizing noise ($p=0.9$). \emph{Left:} the orthogonal $(2f_1^2)^{-\mathrm{wt}}$ of \eqref{eq:pauli_seminorm} and the unitary $(3f_{1,\mbu}^2)^{-\mathrm{wt}}$ (lines) against the measured single-shot second moment at $n=6$ (markers, $95\%$ CI). The prediction is exact rather than an upper bound, because at $d=2$ the contracted third moment is a multiple of the identity (for a real basis this is the computation $R_j=\tfrac12\id$ in the proof of \cref{prop:pauli_seminorm}; \cref{prop:d2complex} is its complex-basis counterpart, and hypothesises $\beta_1\neq\tilde\beta_1$). \emph{Right:} the measured ratio against $(3/2)^k$; triangles are the weight-$2$ ratio at $n=4,6,8,10$, offset horizontally, showing independence of system size.}
\label{fig:local_advantage}
\end{figure}

\subsection{Noise models}
\label{sec:local_examples}
Every single-qubit channel enters \eqref{eq:local_channel} through a single number, which the following identifies. Write $\Lambda$ for the Pauli transfer matrix of $\channel_1$, $\Lambda_{\mu\nu}=\tfrac12\tr[\sigma_\mu\channel_1(\sigma_\nu)]$ with $\sigma_\mu\in\{\id,X,Y,Z\}$.
\begin{claim}[The local noise parameter]\label{claim:local_lambda}
For any single-qubit linear map $\channel_1$,
\begin{equation}\label{eq:beta1_lambda}
\beta_1=1+\Lambda_{zz},\qquad f_1(\channel_1)=\frac{\Lambda_{zz}}{2} ,
\end{equation}
so the local shadow channel, its inverse, and the seminorm \eqref{eq:pauli_seminorm} depend on the noise only through $\Lambda_{zz}=\tfrac12\tr[Z\channel_1(Z)]$. The protocol is inconsequential exactly when $\Lambda_{zz}=1$ and non-invertible exactly when $\Lambda_{zz}=0$.
\end{claim}
\begin{proof}
In the Bloch parametrization a state with Bloch vector $v$ has $\bra0\cdot\ket0=(1+v_z)/2$ and $\bra1\cdot\ket1=(1-v_z)/2$, while $\Pi_0$ and $\Pi_1$ have Bloch vectors $\pm\hat z$ and $\channel_1$ acts as $v\mapsto\Lambda v+t$. Hence $\bra0\channel_1(\Pi_0)\ket0=\tfrac12(1+\Lambda_{zz}+t_z)$ and $\bra1\channel_1(\Pi_1)\ket1=\tfrac12(1+\Lambda_{zz}-t_z)$, whose sum is $\beta_1=1+\Lambda_{zz}$; the translation $t_z$ cancels, which is why a non-unital channel is no worse here than a unital one with the same $\Lambda_{zz}$. Then $f_1=\tfrac{\beta_1-1}{2}=\tfrac{\Lambda_{zz}}{2}$ by \eqref{eq:local_channel}, and $f_1=f_1(\id)=\tfrac12$ iff $\Lambda_{zz}=1$, $f_1=0$ iff $\Lambda_{zz}=0$.
\end{proof}
\noindent The reason a single entry suffices is that the noise acts between the evolution and the measurement, so the twirl sees it only through $\channel_1^{*}(\Pi_b)$, that is only through how populations map to populations. \Cref{tab:local_noise} collects the five models; the derivations are in \cref{app:local_examples}.

\begin{table}[t]
\centering
\begin{tabular}{lccc}
\hline
single-qubit channel $\channel_1$ & $\Lambda_{zz}$ & $f_1$ & $\norm{P}^2_{\shadow}$, weight $k$\\
\hline
depolarizing $\mcd_{1,p}$ & $p$ & $p/2$ & $(p^2/2)^{-k}$\\
amplitude damping $\mathrm{AD}_{1,p}$ & $p$ & $p/2$ & $(p^2/2)^{-k}$\\
dephasing $\diag$ & $1$ & $1/2$ & $2^{-k}$\\
rotation by $\theta$ about $\hat n$ & $1-2(n_x^2+n_y^2)\sin^2\tfrac\theta2$ & $\Lambda_{zz}/2$ & $(\Lambda_{zz}^2/2)^{-k}$\\
readout, flip rates $q_0,q_1$ & $1-q_0-q_1$ & $(1-q_0-q_1)/2$ & $\big((1-q_0-q_1)^2/2\big)^{-k}$\\
\hline
\end{tabular}
\caption{The five noise models of \cref{sec:examples} at the single-qubit level, each entering the local protocol only through $\Lambda_{zz}$ (\cref{claim:local_lambda}). Derivations in \cref{app:local_examples}; the noiseless column is $\Lambda_{zz}=1$, $f_1=\tfrac12$, $\norm{P}^2=2^{k}$. Dephasing is inconsequential, and a $\theta=\pi/2$ rotation about any axis in the $x$--$y$ plane drives $f_1$ to zero exactly. Every entry is derived in \cref{app:local_examples}.}
\label{tab:local_noise}
\end{table}

In every case the local sample complexity \eqref{eq:local_sample} is exponential in the Pauli weight, with base improved by $3/2$ per qubit over the unitary protocol, since $f_1^{\mbo}/f_1^{\mbu}=3/2$ independently of $\Lambda_{zz}$.

\section{Part III: Complex Basis, Global Orthogonal Ensemble}
\label{sec:partIII}

\begin{setting}[Complex basis, global ensemble]\label{set:III}
\cref{set:I} with \ref{it:I_basis} replaced by: $\mcw$ is an arbitrary orthonormal basis, of reality fraction $\varsigma=\alpha_{\mathrm r}/d\in[0,1]$ as in \eqref{eq:reality_fraction}, with $\varsigma=1$ recovering \cref{set:I}. The invertibility condition of \ref{it:I_noise} becomes $\beta+\tilde\beta\neq2$, and \ref{it:I_obs} widens: $O$ need not be symmetric, because the visible space is all of $\mcl(\Complex^d)$ whenever the basis is not real (\cref{claim:complex_visible}).
\end{setting}

\noindent Returning to the global ensemble, we relax the reality of the basis. Its new ingredient is the transposed-diagonal scalar $\tilde\beta$ of \eqref{eq:scalars}, and the channel stays depolarizing but acts on a $\beta,\tilde\beta$-weighted combination of $A$ and $A^\intercal$, and becomes invertible on the full operator space, so non-symmetric observables become estimable. Proof in \cref{app:complex_basis_proof}.

\subsection{Reconstruction: the complex-basis global channel}
\label{sec:complex}

Relaxing the reality of the basis tunes the protocol continuously between real and unitary shadows. Both the basis reality $\alpha_{\mathrm r}$ and the reality fraction $\varsigma$ of \eqref{eq:reality_fraction} now come into play, and the relevant noise scalar is now the transposed-diagonal invariant $\tilde\beta$ of \eqref{eq:scalars}.

\begin{restatable}[Complex-basis noisy channel]{proposition}{propcomplex}
\label{prop:complex_basis}
For an orthogonal $3$-design, a trace-preserving or unital channel $\channel$ (so that $\alpha=d$), and a basis $\mcw$,
\begin{equation}\label{eq:complex_channel}
\mcm_{\mbo,\channel,\mcw}(A)=\mcd_{n,f(\channel)}\big(\tilde A\big),
\qquad
f(\channel)=\frac{\beta+\tilde\beta-2}{(d-1)(d+2)},
\end{equation}
where
\begin{equation}\label{eq:Atilde}
\begin{aligned}
\tilde A&=\frac{\big(\beta(d+1)-d-\tilde\beta\big)A+\big(\tilde\beta(d+1)-d-\beta\big)A^\intercal}{d\,(\beta+\tilde\beta-2)}\\
&=q_\beta A+(1-q_\beta)A^\intercal ,\\
q_\beta&\coloneqq\frac{\beta(d+1)-d-\tilde\beta}{d\,(\beta+\tilde\beta-2)} ,
\end{aligned}
\end{equation}
For a real basis $\tilde\beta=\beta$, whence $\tilde A=A_{\mathrm{sym}}$ and \eqref{eq:complex_channel} reduces to \cref{prop:global_noisy_channel}.
\end{restatable}
At $\channel=\id$, where $\beta=d$ and $\tilde\beta=\alpha_{\mathrm r}$, \eqref{eq:complex_channel} reduces coefficient by coefficient to the noiseless general-reality channel of~\cite{west2025real}, $\mcm_{\mbo(d);\alpha_{\mathrm r}}(A)=\big[\tr(A)(d^2-\alpha_{\mathrm r})\id+(d^2-\alpha_{\mathrm r})A+(\alpha_{\mathrm r}d+\alpha_{\mathrm r}-2d)A^\intercal\big]/[d(d-1)(d+2)]$. The reality-weighted mixture of $A$ and $A^\intercal$ is therefore theirs; what \cref{prop:complex_basis} adds is the channel, which replaces the single scalar $\alpha_{\mathrm r}$ by the pair $\beta,\tilde\beta$. In the large-$d$ limit at fixed reality fraction $\varsigma$ of \eqref{eq:reality_fraction} one recovers, as $\varsigma\to0$, the unitary shadow channel acting on $A$ itself. \Cref{fig:complex_crossover} shows the crossover at each $n$ against the unitary parameter of~\cite{koh2022classical}, with $f$ measured from simulated snapshots as well as computed, and beside it the variance. The two behave differently. One reaches its unitary value at the rate \eqref{eq:complex_relgap}. The other overshoots it at every finite $d$ and returns only as $d$ grows, which is \cref{cor:unitary_limit} and the reason the $\varsigma\to0$ statement is a joint limit. (We use $\varsigma$ rather than the symbol $f$ of~\cite{west2025real}, who write $\alpha_{\mathrm r}=fd$, because $f$ is reserved here for the depolarizing parameter of \eqref{eq:depol_def}; see \cref{rem:convention}.) The invertibility condition becomes $\beta+\tilde\beta\neq2$.

\begin{figure}[t]
\centering
\includegraphics[width=\linewidth]{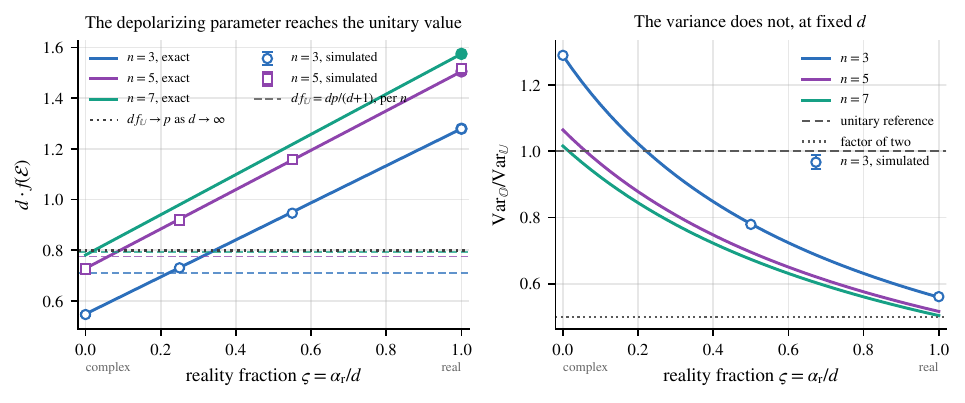}
\caption{Complex-basis crossover for depolarizing noise ($p=0.8$) at $n=3,5,7$, against the reality fraction $\varsigma$ of \eqref{eq:reality_fraction}. The two panels do not tell the same story, which is the point of showing them together. \emph{Left:} the rescaled depolarizing parameter $d\,f(\channel)$ (solid, \eqref{eq:depol_complex}) against the unitary parameter $d\,f_{\mbu}=dp/(d+1)$ of~\cite{koh2022classical} (dashed, one per $n$ and colour-matched; dotted is its $d\to\infty$ limit $p$). The parameter reaches its unitary value as $\varsigma\to0$, at the rate \eqref{eq:complex_relgap}. Markers: $f$ measured by running the protocol end to end, $1.5\times10^6$ shots per point, assuming neither $f$ nor $q_\beta$ --- the raw snapshot $X$ is probed with a traceless operator \emph{before} any inversion, so $\E[\tr(AX)]=f\,\tr(A\,\Psi_{q_\beta}(\rho))$, and $\rho$ is real symmetric so that $\Psi_{q_\beta}(\rho)=\rho$ for every $q_\beta$ (\cref{claim:residue}), leaving $\hat f=\E[\tr(\rho_0X)]/\norm{\rho_0}_2^2$. All eight points agree with the closed form within $0.84$ of their own confidence interval. \emph{Right:} the single-shot variance of the same protocol as a ratio to the unitary one on the same $(\rho,O_0)$. It does \emph{not} reach the unitary value at fixed $d$. At $\varsigma=1$ the ratio is the factor of two, $0.559$, $0.516$, $0.504$ at $n=3,5,7$, falling toward $\tfrac12$; at $\varsigma=0$ it is $1.289$, $1.064$, $1.016$, \emph{above} the unitary reference at every finite $d$ and approaching it only as $d$ grows. Recovering unitary shadows is therefore a joint limit in $(\varsigma,d)$, not a limit in $\varsigma$ --- the finite-$d$ statement \cref{fig:complex_noise} makes at a single $n$, here as a trend. Both curves are closed form via \cref{prop:complex_variance}, which is what makes $n=5,7$ reachable at all: $\mcm^{\dagger}=\mcm$ (\cref{claim:residue}) removes the $d^2\times d^2$ pseudoinverse, and the result agrees with the exact twirl at $d=4$ to $3.2\times10^{-15}$. Open markers: an independent end-to-end simulation at $n=3$, $6\times10^5$ shots, within $0.81$ of its own confidence interval.}
\label{fig:complex_crossover}
\end{figure}

\subsection{Inverse, estimator, and the enlarged visible space}
\label{sec:complex_est}
By \cref{prop:complex_basis} the channel is $\mcd_{n,f}\circ\Psi_{q_\beta}$ with $q_\beta$ as in \eqref{eq:Atilde}, so it is an instance of \cref{claim:residue} and the enlargement of the visible space is decided by a single scalar. \Cref{claim:invert} is the real-basis case, where that scalar takes the value that makes $\Psi$ a projector.

\begin{claim}[The complex basis enlarges the visible space]\label{claim:complex_visible}
Let $\channel$ be completely positive with $\beta+\tilde\beta\neq2$, so that $f\neq0$. Then
\begin{equation}\label{eq:qbeta_half}
q_\beta-\tfrac12=\frac{(\beta-\tilde\beta)(d+2)}{2d\,(\beta+\tilde\beta-2)} ,
\end{equation}
so $q_\beta=\tfrac12$ if and only if $\tilde\beta=\beta$, which by \eqref{eq:scalars} holds exactly on a real basis. Consequently the visible space is the symmetric subspace when the basis is real, and all of $\mcl(\Complex^d)$ otherwise, and in the latter case
\begin{equation}\label{eq:complex_inverse}
\hat\rho=\Psi_{q_\beta}^{-1}\big(\mcd_{n,1/f}(U^\intercal\Pi_wU)\big),\qquad
\Psi_{q_\beta}^{-1}(X)=\frac{q_\beta X-(1-q_\beta)X^\intercal}{2q_\beta-1} ,
\end{equation}
is unbiased on all of $\mcl(\Complex^d)$.
\end{claim}
\begin{proof}
Subtracting $\tfrac12$ from $q_\beta=\tfrac{\beta(d+1)-d-\tilde\beta}{d(\beta+\tilde\beta-2)}$ and clearing the common denominator gives numerator $2\beta(d+1)-2d-2\tilde\beta-d\beta-d\tilde\beta+2d=(\beta-\tilde\beta)(d+2)$, which is \eqref{eq:qbeta_half}; it vanishes exactly when $\tilde\beta=\beta$, the denominator being nonzero by hypothesis. Both scalars are real by \cref{claim:beta_range}, so \cref{claim:residue} applies with $q=q_\beta$: for $q_\beta\neq\tfrac12$ it makes $\Psi_{q_\beta}$ invertible on $\mcl(\Complex^d)$ with the stated inverse, hence $\ker\mcm=0$ and $\mcv=\mcl(\Complex^d)$; for $q_\beta=\tfrac12$ it is the projector onto the symmetric operators, whose orthogonal complement is $\ker\mcm$. Unbiasedness is then $\mcm^{-1}\mcm=\mathrm{id}$ on $\mcv$.
\end{proof}

This enlargement is bought at the price of the variance prefactor set by $f=\tfrac{\beta+\tilde\beta-2}{(d-1)(d+2)}$. Equation \eqref{eq:qbeta_half} says how the two are traded. That gap $\beta-\tilde\beta$ which opens the visible space is the numerator of $q_\beta-\tfrac12$, while $\beta+\tilde\beta-2$ is the denominator of both.

\subsection{Variance, sample complexity, and the reality-tuned crossover}
\label{sec:complex_var}
We obtain the single-shot second moment exactly as in \cref{sec:shadow_norm}, now with the third-moment twirl of $\channel^*(\Pi_w)\otimes\Pi_w\otimes\Pi_w$ evaluated for a complex $\Pi_w$ (so that the transposed diagrams $\Omega$ contribute through $\tilde\beta$ rather than collapsing onto $\beta$). Writing $\hat O=\mcm_{\mbo,\channel}^{-1\dagger}(O)$, the estimator variance is
\begin{equation}\label{eq:complex_var}
\Var[\hat o]=\sum_w\mathbb E_{U}\,\bra w\channel(U\rho U^\intercal)\ket w\,\bra w U\hat O U^\intercal\ket w^{2}-\tr(O\rho)^2 ,
\end{equation}
a $\beta,\tilde\beta$-weighted generalization of \eqref{eq:varO} which reduces to it when $\tilde\beta=\beta$ (real basis). It does not reduce to the unitary variance of~\cite{koh2022classical} at $\varsigma=0$ for fixed $d$, since what \eqref{eq:depol_complex} establishes is the convergence of the depolarizing \emph{parameter}, while the second moment matches its unitary counterpart only as $d\to\infty$, coefficient by coefficient and from a direction that depends on the observable (\cref{cor:unitary_limit}). A closed form exists for the fifteen-term third-moment twirl. Beyond $d,\beta,\tilde\beta$ it requires three reality-sensitive invariants of the noisy measurement,
\begin{equation}\label{eq:complex_invariants}
\gamma=\sum_w \tr\!\big[\channel^*(\Pi_w)\big]\,\lvert\braket{w}{w^{*}}\rvert^{2},\quad
\delta=\sum_w \braket{w}{w^{*}}\bra{w^{*}}\channel^*(\Pi_w)\ket{w},
\end{equation}
with $\gamma\in\Reals$ and $\delta\in\Complex$. On a real basis $\gamma=d$ and $\delta=\beta$. We derive the following in \cref{app:complex_third_moment} and verify it numerically below.

\begin{restatable}[Complex-basis second moment]{proposition}{propcomplexvar}
\label{prop:complex_variance}
Let $d\ge3$, let $\channel$ be completely positive and either trace-preserving or unital (so that $\alpha=d$), let $\hat O=\mcm_{\mbo,\channel}^{-1\dagger}(O)$, and write $D=d(d-2)(d-1)(d+2)(d+4)$. Then $\mathbb E[\hat o^2]=\tr[\rho R]$ with the Hermitian operator
\begin{equation}\label{eq:complex_R}
R=\big(s\,\tr[\hat O^2]+\tilde s\,\tr[\hat O\hat O^\intercal]\big)\mathds 1
+A\,\hat O^{2}+B\,\hat O^\intercal\hat O+\bar B\,\hat O\hat O^\intercal+E\,(\hat O^\intercal)^{2},
\end{equation}
whose coefficients are
\begin{align}
A&=\tfrac{2}{D}\big[(d^2{+}d{-}4)\beta-(d{-}4)\tilde\beta-d^2-2d\,\mathrm{Re}\,\delta+2\gamma\big],\qquad
E=A\big|_{\beta\leftrightarrow\tilde\beta},\notag\\
B&=\tfrac{2}{D}\big[-2d(\beta{+}\tilde\beta)+(d^2{+}2d{-}4)\delta+4\bar\delta-(d{+}2)\gamma+4d\big],\label{eq:complex_coeffs}\\
s&=\tfrac{1}{D}\big[d^3{+}2d^2{-}4d-2d(\beta{+}\tilde\beta)-(d{+}2)\gamma+4(\delta{+}\bar\delta)\big],\notag\\
\tilde s&=\tfrac{1}{D}\big[(d^2{+}3d{-}2)\gamma-2d(d{+}2)+8(\beta{+}\tilde\beta)-2(d{+}2)(\delta{+}\bar\delta)\big].\notag
\end{align}
The seminorm is $\lVert\hat O\rVert_{\mathrm{sh}}^2=\max_\sigma\tr[\sigma R]=\lambda_{\max}(R)$, which equals $\lVert R\rVert_\infty$ because $R\succeq0$: $\tr[\sigma R]=\mathbb E[\hat o^2]\ge0$ for every state $\sigma$. Also $\Var[\hat o]=\tr[\rho R]-\tr(O\rho)^2$. On a real basis ($\tilde\beta=\beta,\ \gamma=d,\ \delta=\beta$) with symmetric $O$, \eqref{eq:complex_R} collapses to the real seminorm \eqref{eq:seminorm_general}, where the four operator coefficients coincide and the operator-to-scalar ratio is exactly $4d(\beta-1)/[d(d+3)-4\beta]$. The restriction $d\ge3$ is necessary and not cosmetic, since $D$ carries the factor $(d-2)$ inherited from the Gram determinant \eqref{eq:detfac}, so \eqref{eq:complex_coeffs} is singular at $d=2$. The single-qubit case, which Part~IV needs, is computed separately in \cref{app:d2complex}. There the answer is finite and in fact simpler.
\end{restatable}

That $R$ is Hermitian follows from \eqref{eq:complex_coeffs}, in which $A$, $E$, $s$ and $\tilde s$ are real, since $\delta$ enters them only through $\mathrm{Re}\,\delta$ or $\delta+\bar\delta$, and $\hat O^2$ and $(\hat O^\intercal)^2$ are Hermitian, while the remaining pair contributes $B\,\hat O^\intercal\hat O+\bar B\,\hat O\hat O^\intercal$, which is Hermitian because $(\hat O^\intercal\hat O)^\dagger=\hat O\hat O^\intercal$ and the two coefficients are complex conjugates. Hence $\Var[\hat o]=\tr[\rho R]-\tr(O\rho)^2$ is real for every state. Separately, $R$ carries a transpose symmetry, $A\leftrightarrow E$ under $\beta\leftrightarrow\tilde\beta$ and $B\leftrightarrow\bar B$ under $\delta\leftrightarrow\bar\delta$, which expresses the fact that replacing the basis by its complex conjugate exchanges $\hat O$ with $\hat O^\intercal$. We verified \eqref{eq:complex_R} against the exact twirl for Haar-random complex CPTP channels and complex bases at $d=4,\dots,7$, and the real-basis reduction symbolically.

Qualitatively this is a bias--variance crossover controlled by the basis reality. Decreasing $\varsigma$ from $1$ enlarges the visible space from the symmetric subspace toward all of $\mcl(\Complex^d)$, while the variance rises monotonically away from the real-basis value, which is the value carrying the factor-of-two advantage. \cref{fig:complex_noise} shows this for $n=2$ under depolarizing noise ($p=0.85$): the exact variance of a fixed symmetric observable climbs from $3.47$ at the real end to $10.86$ at a fully complex basis.

On that plot the unitary reference does not bound the curve. They cross at $\varsigma=0.40$, and the fully complex value lies a factor $1.85$ above the unitary reference $5.86$, so at $d=4$ a fully complex basis measured with orthogonal unitaries is appreciably worse than unitary shadows. This is consistent with the $\varsigma\to0$ statement accompanying \eqref{eq:reality_fraction}, which is a joint limit $d\to\infty$, and \cref{cor:unitary_limit} makes the approach exact. At $\varsigma=0$ the second moment is again a combination of $\tr(O_0^2)$ and $\tr(\rho O_0^2)$, and the ratio of each coefficient to its unitary counterpart is a rational function of $d$ and $p$, the first being $1+2/d+O(d^{-2})$ and the second $1-2/d+O(d^{-2})$. Both tend to $1$, which is the sense in which unitary shadows are recovered, but from opposite sides. Consequently, the sign of the finite-$d$ deviation is a property of the observable rather than of the protocol. For the symmetric observable plotted here, which is $\tr(O_0^2)$-dominated, the ratio to the unitary variance is $1.85$ at $d=4$, $1.29$ at $d=8$, $1.03$ at $d=64$ and $1.002$ at $d=1024$, approaching $1$ from above, while the ratio to the real-basis value falls from $3.13$ towards $2$, the factor of two surrendered in full. By contrast the rank-one fidelity projector of \cref{sec:ghz}, which weights the $\rho$-term far more heavily, instead gives $1.21$ at $d=4$ and $1-0.70/d+O(d^{-2})$ asymptotically, so from $d=16$ onwards a fully complex basis is marginally better than the unitary reference there, and its ratio to the real-basis value tends to $4/3$ rather than $2$. This is the same class dependence that \cref{sec:comparison} identifies for the factor of two itself.

For depolarizing noise the reality-sensitive invariants are functions of $\alpha_{\mathrm r}$ alone, $\tilde\beta=p\alpha_{\mathrm r}+1-p$, $\gamma=\alpha_{\mathrm r}$ and $\delta=\big(p+\tfrac{1-p}{d}\big)\alpha_{\mathrm r}$, so the variance depends on the measurement basis only through $\varsigma$. So the plotted curve is a function of the reality fraction rather than one path among many, and a structurally unrelated family of bases reproduces it. Finally, the sample-complexity bound retains the form \eqref{eq:sample_complexity} with $(\beta-1)$ replaced by $\tfrac12(d-1)(d+2)f=\tfrac12(\beta+\tilde\beta-2)$, so that at fixed noise the price of the enlarged visible space is exactly the loss of the factor of two.

\begin{figure}[t]
\centering
\includegraphics[width=0.59\linewidth]{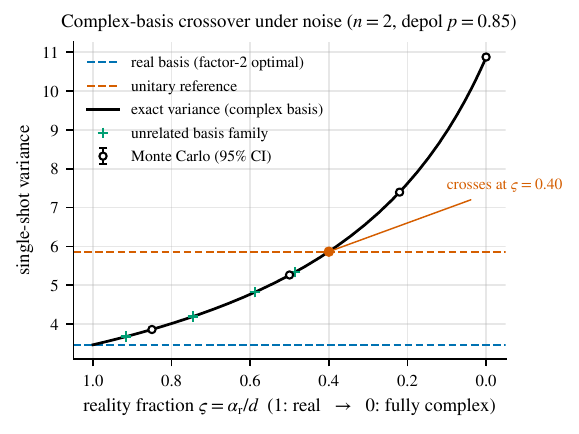}
\caption{Reality-tuned crossover under noise ($n=2$, depolarizing $p=0.85$). The exact single-shot variance of a fixed symmetric observable (solid; open circles are an independent end-to-end simulation with $95\%$ CI) rises monotonically as $\varsigma$ falls from $1$ (real basis, lower dashed) to $0$, from $3.47$ to $10.86$. At this fixed $d=4$ the unitary reference (upper dashed) is not a ceiling: the curve crosses it at $\varsigma=0.40$ and ends a factor $1.85$ above it. This does not contradict the recovery of unitary shadows as $\varsigma\to0$, which is a joint limit with $d\to\infty$; the finite-$d$ deviation and its sign are worked out in the text above. Separately, and for depolarizing noise only, the variance depends on the basis through $\varsigma$ alone --- the crosses are a structurally unrelated family of bases, plotted to show exactly that.}
\label{fig:complex_noise}
\end{figure}

\subsection{Noise models}
\label{sec:complex_examples}
For a complex basis of reality $\alpha_{\mathrm{r}}$ we have $\tilde\beta=p\alpha_{\mathrm{r}}+1-p$ and
\begin{equation}\label{eq:depol_complex}
f(\mcd_{n,p})=\frac{p\,(\alpha_{\mathrm{r}}+d-2)}{(d-1)(d+2)},\qquad
\beta-\tilde\beta=p(d-\alpha_{\mathrm{r}})\ge0 ,
\end{equation}
which for $\alpha_{\mathrm{r}}=d$ (reality fraction $\varsigma=1$) returns \eqref{eq:depol_f}, and for $\varsigma\to0$ with $d\to\infty$ tends to $p/d$, matching the unitary depolarizing parameter $p/(d+1)$ to leading order. The relative gap between the two is exactly
\begin{equation}\label{eq:complex_relgap}
\frac{f(\mcd_{n,p})-f_{\mbu}}{f_{\mbu}}=-\frac{2d}{(d-1)(d+2)}=-\frac2d+O(d^{-2}) ,
\end{equation}
independent of $p$, and every number quoted here is \eqref{eq:complex_relgap} evaluated at the stated $n$ rather than a fit or a measurement: at $\varsigma=0$ and $n=9$ it is $3.9\times10^{-3}$, on a parameter of size $f\approx1.9\times10^{-3}$, so the absolute deviation is $7.6\times10^{-6}$. So the gap closes exponentially in $n$, and the depolarizing parameter of the complex-basis orthogonal protocol is indistinguishable from the unitary one to five decimal places by nine qubits.

\section{Part IV: Complex Basis, Local Orthogonal Ensemble}
\label{sec:partIV}

\begin{setting}[Complex basis, local ensemble]\label{set:IV}
\cref{set:I} with both \ref{it:I_ens} and \ref{it:I_basis} replaced: $\mbo=\mbo(2)^{\otimes n}$, $\channel=\channel_1^{\otimes n}$, and $\mcw=\mcw_1^{\otimes n}$ a product basis with single-qubit reality $\alpha_{\mathrm r,1}\in[0,2]$. The single-qubit invariants of \eqref{eq:complex_invariants} are taken at $d=2$, where \cref{prop:complex_variance} does not apply and \cref{prop:d2complex} is used instead.
\end{setting}

\noindent The final setting (\cref{set:IV}) combines locality (\cref{set:II}) with a complex basis (\cref{set:III}): a product ensemble $\mbo(2)^{\otimes n}$ with a product complex basis and product noise. Everything factorizes, and each factor is the single-qubit complex-basis channel of Part~III.

\subsection{Reconstruction: the complex-basis local channel}
\label{sec:local_complex}

Combining the local factorization of \cref{sec:local} with the complex-basis analysis of \cref{sec:complex} gives the fourth and last setting: a product ensemble $\mbo(2)^{\otimes n}$ with a \emph{product complex} basis $\mcw=\bigotimes_j\mcw_j$ and product noise $\channel=\channel_1^{\otimes n}$.

\begin{restatable}[Complex-basis local channel]{proposition}{proplocalcomplex}
\label{prop:local_complex}
For $\mbo=\mbo(2)^{\otimes n}$, a product basis $\mcw=\bigotimes_j\mcw_j$, and product noise $\channel=\channel_1^{\otimes n}$, the shadow channel factorizes,
\begin{equation}\label{eq:local_complex_channel}
\mcm_{\mbo,\channel,\mcw}=\bigotimes_{j=1}^{n}\mcm_{\mbo(2),\channel_1,\mcw_j},\qquad
\mcm_{\mbo(2),\channel_1,\mcw_j}(A)=\mcd_{1,f_1}\big(q_{\beta_1}A+(1-q_{\beta_1})A^\intercal\big),
\end{equation}
with $f_1=\tfrac{\beta_1+\tilde\beta_1-2}{4}$ and $q_{\beta_1}=\tfrac{3\beta_1-2-\tilde\beta_1}{2(\beta_1+\tilde\beta_1-2)}$ the single-qubit specialization ($d=2$) of \cref{prop:complex_basis}, where $\beta_1,\tilde\beta_1$ are the single-qubit scalars for basis $\mcw_j$. For a complex per-qubit basis ($\beta_1\neq\tilde\beta_1$, i.e.\ $q_{\beta_1}\neq\tfrac12$), each single-qubit factor is invertible on all of $\mcl(\Complex^2)$, so the visible space is the full local operator space and non-symmetric (e.g.\ $Y$-supported) Pauli observables become estimable.
\end{restatable}
\begin{proof}
Factorization is immediate: the Haar measure on $\mbo(2)^{\otimes n}$, the product basis, and the product channel all factor across qubits, so the integrand of \eqref{eq:noisy_channel} factors and $\mcm=\bigotimes_j\mcm_{\mbo(2),\channel_1,\mcw_j}$; each factor is \cref{prop:complex_basis} at $d=2$. Each single-qubit map $A\mapsto q_{\beta_1}A+(1-q_{\beta_1})A^\intercal$ has inverse $X\mapsto\tfrac{q_{\beta_1}X-(1-q_{\beta_1})X^\intercal}{2q_{\beta_1}-1}$ for $q_{\beta_1}\neq\tfrac12$, so $\mcm_{\mbo(2),\channel_1,\mcw_j}$ is invertible whenever $f_1\neq0$ and $q_{\beta_1}\neq\tfrac12$. A real basis has $\tilde\beta_1=\beta_1$, hence $q_{\beta_1}=\tfrac12$ by \eqref{eq:Atilde}, and the map reduces to $A\mapsto A_{\mathrm{sym}}$, so only the symmetric $\{\id,X,Z\}$ block survives, since $(Y)_{\mathrm{sym}}=\tfrac12(Y+Y^\intercal)=0$.
\end{proof}
This completes the parallel treatment of the four measurement settings.

\subsection{Seminorm and sample complexity}
\label{sec:local_complex_seminorm}
A tensor product of invertible single-qubit maps, \eqref{eq:local_complex_channel} makes the operator whose largest eigenvalue \eqref{eq:seminorm_def} computes again a tensor product of per-qubit second moments, and \cref{lem:tensormax} lets the maximization be performed qubit by qubit. Hence for a Pauli string $P=\bigotimes_j P_j$,
\begin{equation}\label{eq:local_complex_seminorm}
\norm{P}^2_{\shadow,\mbo,\channel,\mcw}=\prod_{j=1}^{n}c_1\!\big(P_j;\mcw_j\big),
\end{equation}
where $c_1(P_j;\mcw_j)$ is the single-qubit factor. That factor is a third-moment quantity, so it is not determined by \cref{prop:local_complex}, which fixes only the channel, and it cannot be read off from \cref{prop:complex_variance} either, whose coefficients are singular at $d=2$. We therefore derive it directly in \cref{app:d2complex}, where two things are proved. First, at $d=2$ the contracted third moment $R$ is a multiple of the identity (\cref{prop:d2complex}), so $\mcm_{\mbo(2),\channel_1,\mcw_j}^{-1,\dagger}$ maps a Pauli to a scalar multiple of itself, and every quadratic word in $\{P,P^\intercal\}$ is $\pm\id$ because $P^2=\id$ and $P^\intercal=\pm P$. So the maximization over states in \eqref{eq:seminorm_def} is vacuous, and each $c_1$ is an exact scalar rather than an optimized bound. Second, that scalar is
\begin{equation}\label{eq:c1_main}
c_1(P_j;\mcw_j)=\frac{\alpha_{\mathrm r}^{(j)}}{4f_1^2}\quad\text{for }P_j\in\{X,Z\},
\qquad
c_1(Y;\mcw_j)=\frac{2-\alpha_{\mathrm r}^{(j)}}{2f_1^2\,(2q_{\beta_1}-1)^2},
\end{equation}
with $\alpha_{\mathrm r}^{(j)}$ the reality of the $j$th basis, where the identity qubits contribute $c_1(\id)=1$.

What differs from Part~II is that for a complex per-qubit basis all single-qubit Paulis have a finite factor, so $Y$-supported strings---invisible to the real local protocol---become estimable. For the real (computational) basis, $\alpha_{\mathrm r}^{(j)}=2$ and one recovers Part~II exactly: $c_1(X)=c_1(Z)=(2f_1^2)^{-1}$ and $c_1(Y)=\infty$ (invisible). Under single-qubit depolarizing noise, the two factors are fully explicit: $c_1(X)=c_1(Z)=4/(p^2\alpha_{\mathrm r})$ and $c_1(Y)=2/(p^2(2-\alpha_{\mathrm r}))$ by \eqref{eq:c1_depol}. For instance at $p=0.9$ on a basis of reality $\alpha_{\mathrm r}=1.6775$ they are $c_1(X)=c_1(Z)=2.94$ and $c_1(Y)=7.66$, against the real-basis values $(2f_1^2)^{-1}=2.47$ (with $f_1=p/2=0.45$) and $\infty$. These are exact, and are reproduced by an independent evaluation of the commutant twirl. Thus the complex-local protocol trades a modestly larger per-qubit factor on $X,Z$ for a finite, and hence estimable, factor on $Y$, so that the sample complexity remains exponential in the weight, $\prod_j c_1(P_j;\mcw_j)$, while the visible Pauli alphabet is enlarged from $\{\id,X,Z\}$ to all of $\{\id,X,Y,Z\}$. Note the shape of the trade-off in \eqref{eq:c1_depol}, where the symmetric factors diverge as $\alpha_{\mathrm r}\to0$ and the $Y$ factor as $\alpha_{\mathrm r}\to2$, so no single per-qubit basis is uniformly best and the reality should be tuned to the target Pauli support.

\subsection{The enlarged local visible space}
\label{sec:local_complex_numerics}
At $n=2$ the shadow superoperator has rank $16$ (full) for a complex product basis against $9=3^2$, the $\{\id,X,Z\}^{\otimes2}$ block, for a real one, and $Y$ is recovered unbiasedly in the complex case while lying in the kernel in the real case. This is the local analogue of the enlargement of Part~III.

\section{Many-body case studies: what controls the advantage}
\label{sec:casestudy}

\cref{prop:criterion} decides, for a given state and observable, how much of the factor of two is actually realized, but as stated it is a statement about a scalar $\kappa$ rather than about physics. Here we make it usable (\cref{sec:crit}), read off three qualitatively different regimes (\cref{sec:regimes}), and instantiate them in three many-body settings: GHZ certification (\cref{sec:ghz}), transverse-field Ising energy estimation (\cref{sec:tfim}), and the scalar spin chirality (\cref{sec:chir}), the last of which is invisible to the real protocol and motivates the complex bases of Parts~III--IV. We evaluate all benchmarks in this section in closed form or by exact commutant twirls, so they carry no sampling error; an independent Monte Carlo cross-check is reported in \cref{sec:tfim}.

\subsection{Reading the criterion}
\label{sec:crit}

Two properties make the parameter \eqref{eq:kappa_def} of \cref{prop:criterion} usable. First, $\kappa$ is controlled by the observable's \emph{norm profile}. Bounding $r=\tr(\rho O_0^2)\le\norm{O_0^2}_{\mathrm{sp}}=\norm{O_0}_\infty^2$ gives
\begin{equation}\label{eq:kappa_bound}
\kappa\;\ge\;\frac{d(d+3)-4\beta}{4d(\beta-1)}\;Q,\qquad
Q\;=\;\frac{\norm{O_0}_2^2}{\norm{O_0}_\infty^2}\in[1,d],
\end{equation}
which for depolarizing noise ($\beta-1=p(d-1)$) reads $\kappa\gtrsim Q/(4p)$ at large $d$. So the factor of two is a large-$\norm{O}_2/\norm{O}_\infty$ phenomenon in a precise sense, in that $Q$ is the quantity that must be large, and it must be large compared with $4p$, not with $1$. Second, \eqref{eq:kappa_bound} holds with equality precisely when $\rho$ is supported on the top eigenspace of $O_0^2$, since the only inequality used is $r=\tr(\rho O_0^2)\le\norm{O_0^2}_{\mathrm{sp}}$. That is the generic situation in many-body physics, where one estimates a Hamiltonian in its own eigenstate, and the same configuration saturates the bound of \cref{cor:var_ratio}, where $\rho$ an eigenstate of $O_0$ gives $m^2=r$ exactly.

\subsection{Three regimes}
\label{sec:regimes}

\cref{prop:criterion} sorts observables into three behaviours, all of which occur naturally.

\paragraph{(i) Rank-one observables: the ratio saturates below two.} For a pure-state projector $O=\ketbra{\psi}{\psi}$ evaluated in that same state, $\rho=O$, we compute $t=\tr(O_0^2)=1-1/d$ and, from $O_0^2=(1-2/d)O+\id/d^2$, $r=\tr(\rho O_0^2)=(1-1/d)^2$, so $Q=(1-1/d)^{-1}\to1$, and the norm profile is as small as it can be. Substituting in \eqref{eq:kappa_def} with $\beta=1+p(d-1)$,
\begin{equation}\label{eq:kappa_rank1}
\kappa_{\text{rank-1}}=\frac{d^2+3d-4-4p(d-1)}{4p(d-1)^2}\;\xrightarrow[d\to\infty]{}\;\frac{1}{4p},
\end{equation}
a constant. So the second-moment ratio does not approach two but saturates at the finite ceiling
\begin{equation}\label{eq:ceiling}
\frac{\mathbb E_{\mbu}[\hat o^2]}{\mathbb E_{\mbo}[\hat o^2]}\ \longrightarrow\ \frac{2\cdot\frac1{4p}+1}{\frac1{4p}+1}=\frac{2(1+2p)}{1+4p},
\end{equation}
equal to $1.217$ at $p=0.9$ and $1.2$ in the noiseless limit $p=1$. By \cref{cor:var_ratio} the variance ratio for such a target converges to a strictly larger value, since $\kappa$ stays bounded and $m^2/\E_{\mbo}[\hat o^2]$ does not vanish, and both remain below two. Note the direction of the noise dependence: \eqref{eq:ceiling} increases as $p$ decreases, reaching $1.455$ at $p=0.3$ and $1.714$ at $p=0.1$. Noise inflates the $O_0$-only term relative to the $\rho$-weighted one, pushing even rank-one observables toward the factor-of-two regime, and the advantage of real shadows is thus never degraded by noise, and for low-rank targets it is mildly enhanced.

\paragraph{(ii) Extensive Pauli sums: the ratio approaches two exponentially fast.} Let $O_0=\sum_\mu c_\mu P_\mu$ be a sum of $L$ distinct non-identity Pauli strings, as every local Hamiltonian is. Pauli orthogonality $\tr(P_\mu P_\nu)=d\,\delta_{\mu\nu}$ gives exactly
\begin{equation}\label{eq:pauli_2norm}
\norm{O_0}_2^2=d\sum_\mu c_\mu^2 ,
\end{equation}
whereas extensivity gives only $\norm{O_0}_\infty=\Theta(n)$. Hence
\begin{equation}\label{eq:Q_extensive}
Q=\frac{d\sum_\mu c_\mu^2}{\norm{O_0}_\infty^2}=\Theta\!\Big(\frac{d}{n}\Big),
\end{equation}
since a local Hamiltonian has $\sum_\mu c_\mu^2=\Theta(n)$, being a sum of $O(n)$ terms of bounded coefficient, and $\norm{O_0}_\infty=\Theta(n)$, bounded above by the triangle inequality and below by the extensivity of the ground energy. Both halves are needed; the coefficient sum alone fixes only the numerator. So the norm profile of an extensive observable is exponentially large in $n$, $\kappa=\Theta(d/(pn))$, and by \eqref{eq:ratio_master} the deficit from the factor of two closes exponentially,
\begin{equation}\label{eq:deficit}
2-\frac{\mathbb E_{\mbu}[\hat o^2]}{\mathbb E_{\mbo}[\hat o^2]}=\frac{1}{1+\kappa}+O(d^{-1})=\Theta\!\Big(\frac{p\,n}{d}\Big).
\end{equation}
This is the regime relevant to essentially all Hamiltonian-estimation workloads.

\paragraph{(iii) Dense observables: immediately at two.} For a dense symmetric $O_0$ with a flat spectrum, $\norm{O_0}_2^2=\Theta(d)\norm{O_0}_\infty^2$, so $Q=\Theta(d)$, $\kappa=\Theta(d/p)$ and the ratio is within $O(1/d)$ of two already at modest $n$.

\subsection{GHZ certification revisited}
\label{sec:ghz}

Consider certifying the $n$-qubit GHZ state $\ket{\mathrm{GHZ}}=(\ket0^{\otimes n}+\ket1^{\otimes n})/\sqrt2$ under global depolarizing noise ($p=0.9$) with the fidelity observable $O=\ketbra{\mathrm{GHZ}}{\mathrm{GHZ}}$. This is regime (i) in its purest form, with a projector that is symmetric, hence visible, but whose norm profile is minimal. Evaluating \eqref{eq:kappa_def} exactly gives $\kappa=0.407,0.304,0.284,0.279,0.278$ for $n=2,4,6,8,10$, converging to the predicted $1/(4p)=0.2778$ of \eqref{eq:kappa_rank1}, and the exact second-moment ratios are $1.404,1.288,1.238,1.223,1.219$, converging to the second-moment ceiling $2(1+2p)/(1+4p)=1.217$ of \eqref{eq:ceiling} rather than to two. Variance ratios are larger, $1.70,1.47,1.37,1.34,1.34$, and converge to a different value, because $\rho$ is here an eigenstate of $O_0$, so $m^2=r$ and the bound of \cref{cor:var_ratio} is saturated at $x=0.353$, giving $(R-x)/(1-x)=1.338$ from $R=1.219$. Rank-one targets are exactly where the two ceilings separate, and \eqref{eq:ceiling} bounds the second moment, not the variance. Real shadows still reduce the sample count by roughly a quarter to a third for GHZ fidelity, but the clean factor of two is provably unavailable for a rank-one target, since by \eqref{eq:ceiling} no choice of $n$ can produce it.

\subsection{Transverse-field Ising energy estimation}
\label{sec:tfim}

For the complementary regime, take the critical transverse-field Ising chain
\begin{equation}\label{eq:tfim}
H=-J\sum_{i=1}^{n-1}Z_iZ_{i+1}-h\sum_{i=1}^{n}X_i,\qquad h=J=1,
\end{equation}
and estimate $\tr(H\rho)$ in its own ground state $\rho=\ketbra{\psi_0}{\psi_0}$. Both $X$ and $Z$ are symmetric, so $H=H^\intercal$ lies in the visible space of \cref{prop:global_noisy_channel} and no complex basis is needed; $H$ is traceless, so $O_0=H$. This is regime (ii), and the constants can be pinned down. There are $L=2n-1$ Pauli terms with unit coefficients, so \eqref{eq:pauli_2norm} gives $\norm{H}_2^2=d(2n-1)$ exactly, while criticality fixes the ground-state energy density $\norm{H}_\infty/n\to4/\pi$, whence
\begin{equation}\label{eq:Q_tfim}
Q=\frac{d(2n-1)}{\norm{H}_\infty^2}\;\longrightarrow\;\frac{2d}{(4/\pi)^2n}=\frac{\pi^2}{8}\frac{d}{n},
\qquad
\kappa\;\simeq\;\frac{Q}{4p}\;\simeq\;\frac{\pi^2}{32}\frac{d}{p\,n}.
\end{equation}
Exact evaluation is consistent with both, with a visible $1/n$ finite-size correction: $Q\cdot n/d=1.2399,1.2396,1.2394,1.2391$ for $n=8,9,10,11$, approaching $\pi^2/8=1.2337$ from above, and at $p=0.9$, $n=10$ the exact $\kappa=35.27$ against $\pi^2d/(32pn)=35.09$. Second-moment ratios come out as $1.538,1.648,1.808,1.923,1.974$ for $n=2,4,6,8,10$ and the variance ratios $1.82,1.82,1.90,1.96,1.99$, so the factor of two is essentially saturated by ten qubits, with the deficit closing like $n/d$ as predicted by \eqref{eq:deficit}. \cref{fig:ghz_case_study} shows both families collapsing onto the single master curve \eqref{eq:ratio_master}, which is the content of \cref{prop:criterion}.

As an independent end-to-end check of the analytics, the numbers above are closed-form and carry no sampling error. We also simulated the full protocol for \eqref{eq:tfim}, drawing $U\sim\mbo(d)$, applying the depolarizing channel, sampling an outcome, and inverting. Averaging $M=40$ independent trajectories of $2.5\times10^5$ shots each ($10^7$ shots), the empirical mean of $\hat o$ agrees with $\tr(H\rho)$ within $1.3$ standard errors and the empirical variance matches \eqref{eq:varO} to relative deviations of $1.2\times10^{-4}$ ($n=3$) and $2.4\times10^{-4}$ ($n=4$).

\begin{figure}[t]
\centering
\includegraphics[width=\linewidth]{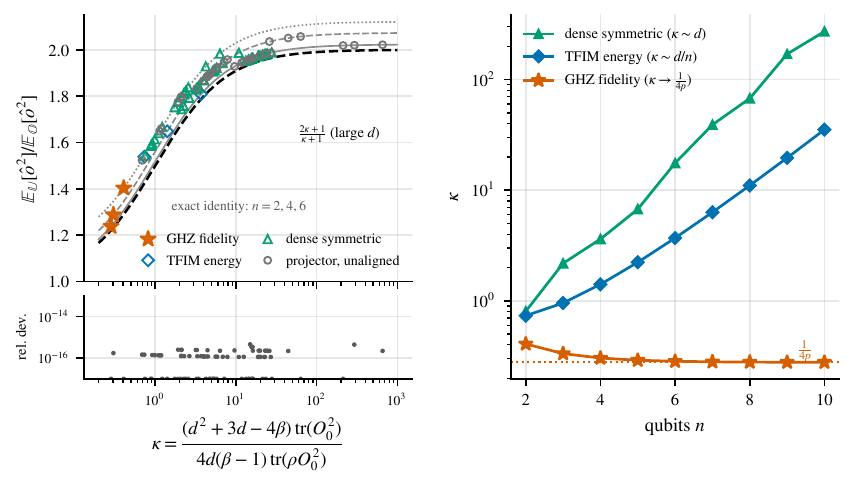}
\caption{The advantage is governed by the single parameter $\kappa$ of \eqref{eq:kappa_def}. \emph{Left:} exact second-moment ratios for four instance families at $n=2,4,6$ and depolarizing $p=0.9$, spanning three decades in $\kappa$. Grey curves are the exact identity \eqref{eq:ratio_identity} for each dimension, while the dashed curve is its large-$d$ limit, which is not an upper envelope here because $\varrho_L>2$ at this noise strength, by \eqref{eq:betastar}. The strip below shows each instance's deviation from the identity for its own dimension. \emph{Right:} $\kappa$ saturates at $1/(4p)$ for a rank-one projector but grows like $d/n$ for an extensive Hamiltonian and like $d$ for a dense observable.}
\label{fig:ghz_case_study}
\end{figure}

\subsection{Chirality: antisymmetric observables and time reversal}
\label{sec:chir}

Our last case study is an observable the real protocol cannot estimate at all, which locates precisely where Parts~III--IV become necessary. Consider the scalar spin chirality on three sites,
\begin{equation}\label{eq:chirality}
\chi_{ijk}=\mathbf S_i\cdot(\mathbf S_j\times\mathbf S_k)=\tfrac18\sum_{abc}\epsilon_{abc}\,\sigma_i^a\sigma_j^b\sigma_k^c ,
\end{equation}
the order parameter of chiral phases and the observable underlying chirality-based qubit encodings in frustrated magnets. Every term of \eqref{eq:chirality} contains each of $X,Y,Z$ exactly once, hence exactly one factor of $Y$; since $P^\intercal=(-1)^{\#Y}P$ for a Pauli string, every term is antisymmetric and
\begin{equation}\label{eq:chi_antisym}
\chi_{ijk}^\intercal=-\chi_{ijk},\qquad\text{equivalently}\qquad (\chi_{ijk})_{\mathrm{sym}}=0 .
\end{equation}
By \cref{prop:global_noisy_channel} the noisy real shadow channel annihilates $\chi$ for every CPTP $\channel$, so the chirality lies in the kernel rather than in a poorly conditioned direction, and no amount of data and no noise characterization can recover it from a real-basis protocol. This has a physical reading. Combining Hermiticity, $\chi=\chi^\dagger$ with $\dagger$ the Hermitian adjoint, with \eqref{eq:chi_antisym} gives $\bar\chi=-\chi$, so $\chi$ has purely imaginary entries. (Here $\bar{\,\cdot\,}$ is entrywise complex conjugation in the computational basis, as in \cref{sec:observables}; we avoid the star, which is reserved throughout for the Hilbert--Schmidt adjoint of a channel.) Hence for any real $\rho$,
\begin{equation}\label{eq:chi_vanishes}
\overline{\tr(\chi\rho)}=\tr(\bar\chi\,\bar\rho)=\tr(-\chi\rho)=-\tr(\chi\rho),
\end{equation}
so $\tr(\chi\rho)$ is purely imaginary; but $\chi$ and $\rho$ are Hermitian, so $\tr(\chi\rho)\in\Reals$, and therefore $\tr(\chi\rho)=0$. Observables the real protocol cannot see are those whose expectation values vanish unless time-reversal symmetry is broken.

We illustrate this with the $J_1$--$J_2$ Heisenberg triangle $H=J_1\sum_i\mathbf S_i\cdot\mathbf S_{i+1}+J_2\sum_i\mathbf S_i\cdot\mathbf S_{i+2}$ ($J_1=1$, $J_2=0.5$, periodic, $n=3$). Because $X\otimes X$, $Y\otimes Y$ and $Z\otimes Z$ are all symmetric, $H$ is real, and its ground state may be chosen real, and the exact chirality expectation in that state vanishes identically, as the argument above requires. Adding a chiral (Dzyaloshinskii--Moriya-like) term, $H_\lambda=H+\lambda\chi_{012}$ with $\lambda=0.8$, breaks time reversal and yields a ground state with $\tr(\chi\rho)=0.4330$. Estimating this value under depolarizing noise ($p=0.9$):
\begin{itemize}[leftmargin=*]
\item \emph{Real basis.} The exact estimator mean is $0$, i.e.\ the bias equals the entire signal, and $\norm{\mcm_{\mbo,\channel}(\chi)}_2=6\times10^{-17}$ confirms kernel membership. Here the shadow superoperator has rank $36=d(d+1)/2$ out of $d^2=64$, the symmetric block of \cref{prop:global_noisy_channel}.
\item \emph{Complex basis.} Following \cref{prop:complex_basis}, a basis of reality $\alpha_{\mathrm r}<d$ makes the superoperator full rank ($64$ of $64$) and the estimator unbiased to $\lesssim3\times10^{-15}$. Single-shot variances of $\hat\chi$ are $9.40$, $2.07$ and $0.98$ at $\alpha_{\mathrm r}/d=0.900$, $0.577$ and $0.178$, so the more the basis departs from real, the better conditioned the previously invisible direction becomes.
\item \emph{The price.} The same complex basis raises the variance of the \emph{visible} energy observable from $6.55$ (real basis) to $9.03$ at $\alpha_{\mathrm r}/d=0.577$, a $38\%$ increase, the loss of the factor of two quantified in \cref{sec:complex_var}.
\end{itemize}
These are exact statements about the twirl; to confirm that the \emph{sampled} protocol realizes them, we also ran the full estimator end to end (draw $U\sim\mbo(d)$, apply $\channel$, sample an outcome in the complex basis $\mcw$, invert) for $2\times10^6$ shots at each of five bases. Monte Carlo mean and variance of $\hat\chi$ agree with the exact values to within $1.3$ and $1.0$ confidence intervals, respectively, at a relative precision of $\sim0.2\%$ on the variance. \cref{fig:chirality} displays both effects, with the Monte Carlo overlaid. In practice, then, use the real ensemble with a real basis for energies, structure factors and any $\{\id,X,Z\}$-supported correlator, where the factor of two and the local $(3/2)^k$ apply, and to switch to a complex basis only for the $Y$-supported, time-reversal-odd observables, currents and chiralities, paying a bounded variance penalty for access to a strictly larger observable algebra.

\begin{figure}[t]
\centering
\includegraphics[width=0.878\linewidth]{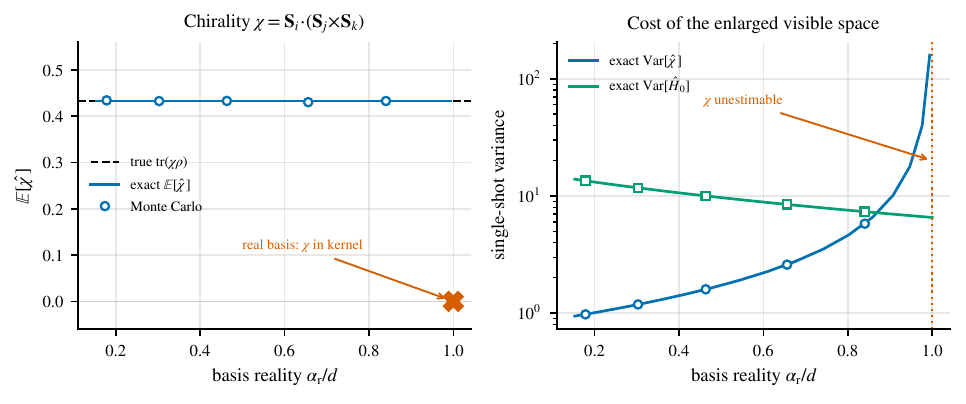}
\caption{An observable outside the real visible space ($n=3$, depolarizing $p=0.9$). Curves are exact; open symbols are an independent end-to-end simulation with $95\%$ CI. \emph{Left:} the scalar chirality \eqref{eq:chirality} in a time-reversal-broken $J_1$--$J_2$ ground state. A real basis is maximally biased---$\chi$ lies in the kernel, so $\hat\chi\equiv0$, while every complex basis is unbiased. \emph{Right:} the trade-off. $\Var[\hat\chi]$ (circles) is finite only for a complex basis and improves as the basis becomes less real, while $\Var[\hat H_0]$ for the visible energy (squares) degrades.}
\label{fig:chirality}
\end{figure}

\section{Design-independence: all orthogonal 3-designs coincide}
\label{sec:optimality}
Our construction rests on the orthogonal $3$-design property, which raises the question of whether its consequences are specific to the Haar measure on $\mbo(d)$ or shared by every $3$-design. They are shared, and exactly rather than approximately. We prove this, and distinguish it from optimality, which is a different and weaker statement.

\subsection{The equivalence}
\label{sec:opt_thm}
Fix a measurement basis and a known channel $\channel$. Every quantity in this paper, the shadow channel \eqref{eq:global_depol}, the estimator, the variance \eqref{eq:varO}, and the seminorm \eqref{eq:seminorm_exact}, is built from the integrals \eqref{eq:noisy_channel} and \eqref{eq:complex_var}, which involve the ensemble $\nu$ only through its moments of order $\le3$. This yields the exact equivalence on which the use of real Cliffords rests.
\begin{proposition}[Design-independence]
\label{prop:design_indep}
Let $\nu,\nu'$ be ensembles on $\mbo(d)$, and fix a state, an observable, and a known channel $\channel$.
\begin{enumerate}[label=(\roman*),leftmargin=*]
\item\label{it:di_second} If $\nu$ and $\nu'$ agree as orthogonal $2$-designs, the noisy shadow channel, the classical shadow, the estimator and its bias are identical for the two.
\item\label{it:di_third} If they agree as orthogonal $k$-designs for every $k\le3$ --- in particular if both are orthogonal $3$-designs, for instance the Haar measure on $\mbo(d)$ and the real Cliffords $\mcc_n\cap\mbo(d)$ --- then the single-shot variance and the shadow seminorm agree as well.
\end{enumerate}
The split is not an artefact of the proof: \cref{sec:opt_thm} exhibits two ensembles satisfying \ref{it:di_second} but not \ref{it:di_third} whose variances genuinely differ.
\end{proposition}
\begin{proof}
Each of the listed quantities is, by construction, a function of the ensemble only through the integrals \eqref{eq:noisy_channel} and \eqref{eq:complex_var}, and the two integrals use different moments. The shadow channel and the classical shadow are built from the second moment $\E_{U\sim\nu}[U^{\otimes2}(\cdot)U^{\intercal\otimes2}]$ alone; the estimator is built from them, and its bias is a difference of two such quantities, so all four agree as soon as the second moments do, which is \ref{it:di_second}. The single-shot variance and the seminorm are built from the third moment, which agrees only under the stronger hypothesis, giving \ref{it:di_third}. In each case the moments coincide by definition of a design.
\end{proof}
\noindent The same freedom is exploited in the unitary case to shrink the ensemble rather than to justify it: Zhang \etal~\cite{zhang2024minimal} reduce the Clifford ensemble to $2^n+1$ circuits drawn from mutually unbiased bases while leaving the post-processing channel unchanged.
In practice one therefore implements the finite real Clifford group rather than a continuous Haar average, without altering any statistical prediction. We do not prove the design property here; we take it from Refs.~\cite{hashagen2018real,nebe2006self}, and do not attempt to confirm it by enumeration, the ambient Clifford group already containing $2^{15}\prod_{j=1}^{3}(4^j-1)=9.3\times10^7$ elements modulo phases at $n=3$.

\paragraph{On optimality.} It is tempting to go further and claim that the orthogonal $3$-design minimizes the variance among all orthogonal $2$-designs. That claim is false, and the reason separates exactly what \cref{prop:design_indep} buys from what it does not. Shadow channel, estimator and unbiasedness are built from the second moment alone, so a $2$-design already fixes them; the variance is built from the third, so a $2$-design does not fix it.

For a subgroup $\Gamma\le\mbo(d)$ the commutant of $U^{\otimes k}$ always contains the Brauer commutant, so $\Gamma$ is a $k$-design exactly when the two have equal dimension, and by Schur orthogonality that dimension is a sum of traces,
\begin{equation}\label{eq:dimcomm_char}
\dim\comm(\Gamma,k)=\frac{1}{\lvert\Gamma\rvert}\sum_{U\in\Gamma}\big(\tr U\big)^{2k} ,
\end{equation}
since the character of $U^{\otimes k}$ is $(\tr U)^k$. This makes the search tractable, the targets being $\dim\comm(\mbo(d),2)=3$ and $\dim\comm(\mbo(d),3)=15$ for $d\ge3$.

At $d=4$, take $\Gamma_{288}$, the subgroup of index $4$ in the real two-qubit Clifford group $\mcc_2\cap\mbo(4)$ generated by
\begin{equation}\label{eq:g288}
\tfrac12\begin{psmallmatrix}1&-1&-1&-1\\1&-1&1&1\\-1&-1&1&-1\\-1&-1&-1&1\end{psmallmatrix},\qquad
\begin{psmallmatrix}0&0&0&-1\\1&0&0&0\\0&0&-1&0\\0&-1&0&0\end{psmallmatrix},\qquad
\tfrac12\begin{psmallmatrix}1&-1&1&1\\1&-1&-1&-1\\1&1&1&-1\\-1&-1&1&-1\end{psmallmatrix} .
\end{equation}
By \eqref{eq:dimcomm_char} it has $\dim\comm(\Gamma_{288},2)=3$ and $\dim\comm(\Gamma_{288},3)=21$, so it is an orthogonal $2$-design and not a $3$-design, while the ambient group of order $1152$ has $3$ and $15$ and is a $3$-design. Its noisy shadow channel therefore agrees with the $3$-design's \emph{exactly}, by \cref{prop:design_indep}\,\ref{it:di_second}; evaluated numerically the two differ by less than $10^{-14}$, which is the floating-point floor rather than a residual discrepancy. Its single-shot second moment does not agree: under depolarizing noise at $p=0.9$, for the traceless symmetric observable obtained by normalizing
\begin{equation}\label{eq:g288_witness}
O=\begin{psmallmatrix}-1&1&0&1\\1&-1&0&0\\0&0&-1&0\\1&0&0&1\end{psmallmatrix},
\qquad \rho=\diag(0.55,0.25,0.15,0.05),
\end{equation}
the ensemble $\Gamma_{288}$ gives a second moment $20.6\%$ \emph{lower} than the exact $3$-design does. On any Pauli string the two agree exactly, which locates where the freedom lives.

The same search at $d=2$ is a complete enumeration rather than a sample, the finite subgroups of $\mbo(2)$ being exactly the cyclic $C_m$ and dihedral $D_m$; of the $24$ conjugacy classes tested, $D_3$ and $D_6$ are $2$-designs that are not $3$-designs, and both reproduce the exact $3$-design variance regardless. The reason is the rank collapse of \cref{app:brauer}: at $d=2$ a traceless symmetric $O_0$ has $O_0^2\propto\id$, so the contracted third moment is a multiple of the identity (\cref{prop:d2complex}) and the directions in which their third moments differ are directions the protocol never probes. The single-qubit case is thus insensitive to the distinction, and a counter-example must live at $d\ge3$.

Tailoring the ensemble to a target observable class in this way is the mechanism behind biased and locally-scrambled shadow schemes~\cite{hu2023classical,west2025real}. What \cref{prop:design_indep} establishes is the weaker but useful invariance that within the orthogonal $3$-design class the protocol is canonical, so the factor-of-two and $(3/2)^k$ advantages are properties of the class as a whole rather than artefacts of a particular design. We state the optimality question that remains in \cref{sec:discussion}.

\section{Discussion}
\label{sec:discussion}

We have shown that the real classical shadows protocol admits a clean and complete theory of known Markovian noise, entirely parallel to the unitary case of Koh and Grewal~\cite{koh2022classical} but with the orthogonal Weingarten calculus in place of the unitary one. The central practical message is that the sample-complexity advantages of real shadows are \emph{robust}, since the local $(3/2)^k$ survives arbitrary known CPTP noise, and so does the global factor of two whenever the observable's norm profile grows---for the second moment by \eqref{eq:ratio_identity}, for the variance by \cref{cor:var_ratio}---with the noise entering that limit only through a factor within $2/d$ of two. For rank-one targets both saturate strictly below two. The mechanism is that noise only rescales a depolarizing parameter and never touches the visible space. Complex bases reach unitary shadows in the joint large-$d$, zero-reality limit, with the two depolarizing parameters agreeing to a relative $2/d$ already at hardware-relevant sizes. At small $d$, however, the approach is not from below, a fully complex basis measured with orthogonal unitaries can be markedly worse than unitary shadows (\cref{fig:complex_noise}).

The structural reason sharpens a natural expectation about which noise the real construction ``tolerates.'' On a real basis, \cref{prop:global_noisy_channel} (and its local counterpart, \cref{prop:local_noisy}) holds for every trace-preserving or unital $\channel$, because the derivation twirls over $\mbo(d)$ and never invokes any reality or conjugation property of the noise, so the shadow channel is $\mcd_{n,f}\circ(\cdot)_{\mathrm{sym}}$ for any such $\channel$, with all noise dependence carried by the single scalar $\beta$ (as $\alpha=d$). One might expect coherent errors with complex Kraus operators to spoil the real structure by coupling the ``real'' and ``imaginary'' sectors of operator space; \cref{sec:coherent} shows they do not. A coherent over-rotation is absorbed entirely into $\beta=\sum_b\lvert V_{bb}\rvert^2$, and the true compatibility condition is not that $\channel$ commute with complex conjugation but simply that $\beta\neq1$, i.e.\ that the symmetric sector not be fully depolarized. In this sense every CPTP noise is structurally compatible with real-basis shadows. A channel is harmful only quantitatively, through how far it drives $\beta$ from its noiseless value, and coherent errors are notable precisely because they can approach the threshold $\beta=1$ while remaining structurally benign.

The reality of the noise enters in the complex-basis setting of Parts~III--IV. There the sector coupling is real, but it is controlled by the interplay of the basis and the noise rather than by the noise alone, since the basis reality $\alpha_{\mathrm r}$ sets the size of the enlarged visible space, the gap $\tilde\beta-\beta$ measures the departure of the basis from real (for depolarizing noise it is exactly $p(\alpha_{\mathrm r}-d)$, a purely basis-driven quantity), and, in the third moment of \cref{prop:complex_variance}, the imaginary part of the invariant $\delta$ is what couples the real and imaginary operator sectors in the variance. A conjugation-commuting channel keeps $\delta$ real but, through $\alpha_{\mathrm r}$ alone, still gives $\tilde\beta\neq\beta$ on a complex basis; conjugation-mixing noise additionally makes $\mathrm{Im}\,\delta\neq0$. The naive dichotomy ``conjugation-symmetric noise is compatible, complex noise is not'' is therefore not the right one, because on a real basis all noise is compatible, while on a complex basis the relevant quantities are $\alpha_{\mathrm r}$, $\tilde\beta-\beta$, and $\mathrm{Im}\,\delta$.

We conclude by listing the open problems this work leaves.
\begin{enumerate}[leftmargin=*]
\item \textbf{Unknown noise.} Our inversion assumes $\channel$ is characterized in advance. The robust-estimation strategy of Chen \etal~\cite{chen2021robust}, which learns a stochastic noise model, should transfer to the orthogonal ensemble; quantifying the residual bias when the assumed channel $\mathcal F$ differs from the true $\channel$ is a concrete next step. \cref{rem:blind} settles the extreme case $\mathcal F=\id$ exactly, the bias is the multiplicative factor $\tfrac{\beta-1}{d-1}$ of \eqref{eq:blind_bias}, and the same argument gives $f(\channel)/f(\mathcal F)$ for any assumed $\mathcal F$, so the open part is not the bias itself but the sample cost of learning $\beta$ to a given accuracy. The most directly transferable precedent is~\cite{wu2024error}, which calibrates noise of exactly the class assumed here for an ensemble that is likewise governed by orthogonal Haar averages, and which quantifies the calibration overhead as $\tilde{\mathcal O}(\sqrt n)$ measurements; We leave it as an open problem to determine whether a comparable overhead suffices for $\mbo(2^n)$, where the relevant commutant is the order-three Brauer algebra of \cref{app:brauer} rather than its $\mbo(2n)$ counterpart.
\item \textbf{Optimality of the ensemble.} \cref{prop:design_indep} shows every orthogonal $3$-design gives identical predictions, and \cref{sec:opt_thm} exhibits an orthogonal $2$-design at $d=4$ that is not a $3$-design and beats every $3$-design on a particular observable by $20.6\%$, so the $3$-design is \emph{not} variance-optimal among $2$-designs. We leave it as an open problem to decide whether the orthogonal $3$-design is optimal in a minimax (worst-case-observable) or average sense.
\item \textbf{Beyond GTM.} Gate-dependent and non-Markovian noise break the clean depolarizing form. We leave it as an open problem to determine whether a weaker structural statement survives, for instance that the visible space is still preserved even when the channel is no longer a depolarizer on it. The shallow-circuit analysis of~\cite{yu2026light} indicates what to expect, and locates the difference in the placement of the noise rather than in any failure of gate-independence. Their noise is also independent of the sampled gates, but it acts between circuit layers rather than once after the evolution, and the bias is then observable-dependent rather than a single scalar, damping the estimated Pauli coefficient exponentially in the size of a contiguous string. Relaxing spatial uniformity and time-stationarity does not change this: Gaussian site-to-site fluctuations and a deterministic linear drift both leave the exponential law intact, while fluctuations correlated across a whole layer, which couple an entire light-cone slice, can instead generate nonlinear corrections. Recovering the collapse onto $\beta$ at shallow depth therefore needs more than gate-independence. It needs the depth-$L$ twirl to see a single invariant.
\item \textbf{Error mitigation.} Because inversion already implements $\mcm_{\mbo,\channel}^{-1}$ --- the inverse of the \emph{shadow} channel, not of $\channel$ itself --- deterministically in post-processing, combining real shadows with probabilistic error cancellation~\cite{temme2017error} or zero-noise extrapolation may compound the advantages.
\item \textbf{Circuit depth.} Every statement in this paper is for an \emph{exact} orthogonal $3$-design. A circuit family of finite depth realizes at best an approximate one, and whether logarithmic-depth approximate orthogonal designs exist is open~\cite{west2025real,schuster2024random}. We leave it as an open problem to establish or rule this out, and to settle the resulting trade-off between gate cost and the sample-complexity advantage.
\end{enumerate}

\section*{Data and code availability}
The code used to generate every figure in this paper, together with the exact orthogonal
twirls via the Brauer commutant, the symbolic Weingarten computations of \cref{app:brauer},
the protocol simulators they rest on, and the two-design search of \cref{sec:opt_thm}, is openly available at
\url{https://github.com/Atharva-11/Noisy-Real-Classical-Shadows}. No experimental data were used.

\begin{acknowledgments}
This project is supported by the National Research Foundation, Singapore through the National Quantum Office, hosted in A*STAR, under the Advanced Quantum Algorithms and Solutions (AQAS) Funding Initiative (S25Q9DA001).
\end{acknowledgments}

\appendix
\onecolumn

\section{Weingarten Calculus for the Orthogonal Group: Brauer Algebra, Gram Matrix, and the Singularity at \texorpdfstring{$d=2$}{d=2}}
\label{app:brauer}

This appendix develops the algebraic machinery behind every orthogonal Haar average used above, from the Brauer algebra to the order-three Weingarten function. The order-three Gram determinant is factored and its singularity at $d=2$ is explained representation-theoretically, with explicit null vectors and a proof that the linear systems the protocol needs remain solvable there. The treatment draws on the classical literature~\cite{brauer1937algebras,weyl1939classical,goodman2009symmetry}, on orthogonal Weingarten calculus~\cite{collins2006integration,collins2009some}, and on the structure theory of the Brauer algebra~\cite{wenzl1988structure,graham1996cellular,rui2005criterion}, specialized to $k=3$.

\subsection{The Brauer algebra}
\label{app:brauer_def}

\begin{definition}[Brauer diagram and algebra~\cite{brauer1937algebras}]
\label{def:brauer}
A \emph{Brauer diagram} on $2k$ nodes is a perfect matching (partition into $k$ unordered pairs) of a set of $2k$ points arranged in a top row $1,\dots,k$ and a bottom row $1',\dots,k'$. For a parameter $\delta$ in a commutative ring, the \emph{Brauer algebra} $\Brauer_k(\delta)$ is the free module on the set of Brauer diagrams, with multiplication defined on diagrams by vertical concatenation: to form $D_1 D_2$, place $D_1$ above $D_2$, identify the bottom row of $D_1$ with the top row of $D_2$, read off the resulting matching between the remaining outer nodes, and multiply by $\delta^{c}$ where $c$ is the number of closed loops formed in the middle. Extending bilinearly makes $\Brauer_k(\delta)$ an associative unital algebra of dimension $(2k-1)!!$.
\end{definition}

The $(2k-1)!!$ matchings split by the number of ``through-strands'' (top-to-bottom pairs), which for $k=3$ is either $3$ (the $6$ permutation diagrams, indexed by $S_3$) or $1$ (the $9$ diagrams with one top ``cup'' and one bottom ``cap''). The symmetric-group algebra $\Complex[S_k]$ embeds as the span of the through-permutation diagrams. The Brauer algebra is cellular in the sense of Graham and Lehrer~\cite{graham1996cellular}, with cell (standard) modules $W_\lambda$ labelled by partitions $\lambda$ of $k, k-2, k-4,\dots$. For $k=3$ the labels are $\lambda\vdash3$ (through-number $3$) and $\lambda\vdash1$ (through-number $1$). Generically $\Brauer_k(\delta)$ is semisimple and $\dim\Brauer_k(\delta)=\sum_\lambda(\dim W_\lambda)^2$. For $k=3$ this is $1^2+2^2+1^2+3^2=15$, from $\dim W_{(3)}=1$, $\dim W_{(2,1)}=2$, $\dim W_{(1,1,1)}=1$ and $\dim W_{(1)}=3$. Semisimplicity fails at a finite set of integers $\delta$, determined precisely by Rui~\cite{rui2005criterion} (building on Wenzl~\cite{wenzl1988structure}), and we will see the failure at $\delta=2$ directly.

\subsection{Schur--Weyl duality for the orthogonal group}
\label{app:sw}

Let $V=\Complex^d$ with the orthogonal group $\mbo(d)$ acting diagonally on $V^{\otimes k}$ by $U\mapsto U^{\otimes k}$. The natural map sends each Brauer diagram to the tensor contraction that connects tensor legs according to the matching, with each cup $\sum_i\ket{i}\ket{i}$ and each cap $\sum_j\bra j\bra j$. Concretely, for $\sigma\in\Brauer_k(d)$,
\begin{equation}
\phi_d(\sigma)=\sum_{i_1,\dots,i_{2k}=0}^{d-1}\;\ket{i_{k+1}\cdots i_{2k}}\bra{i_1\cdots i_k}\prod_{\{a,b\}\in\sigma}\delta_{i_a i_b}.
\end{equation}
This is an algebra homomorphism $\phi_d:\Brauer_k(d)\to\End(V^{\otimes k})$, the closed-loop factor $\delta=d$ arising from $\tr\id_V=d$.

\begin{theorem}[Schur--Weyl duality for $\mbo(d)$~\cite{brauer1937algebras,goodman2009symmetry}]
\label{thm:sw}
The image of $\phi_d$ is exactly the commutant of the diagonal $\mbo(d)$ action,
\begin{equation}
\im\phi_d=\comm(\mbo(d),k)\;:=\;\End_{\mbo(d)}(V^{\otimes k})=\{X:[X,U^{\otimes k}]=0\ \forall U\in\mbo(d)\}.
\end{equation}
The map $\phi_d$ is an isomorphism onto its image, and $\Brauer_k(d)\to\comm(\mbo(d),k)$ is an algebra isomorphism, whenever $d\ge k$. For $d<k$ it has a nonzero kernel.
\end{theorem}

The dimension of the commutant is a Haar moment. Since $V\cong V^*$ for real orthogonal $U$,
\begin{equation}\label{eq:commdim}
\dim\comm(\mbo(d),k)=\dim\big(V^{\otimes 2k}\big)^{\mbo(d)}=\int_{\mbo(d)}\big(\tr U\big)^{2k}\,\dd U .
\end{equation}
For $d\ge k$ the integral equals $(2k-1)!!$, matching $\dim\Brauer_k(d)$ and confirming injectivity of $\phi_d$. The case $d=2$ is special and central here: parametrizing $\mathrm{SO}(2)$ by rotations $U(\theta)$ with $\tr U(\theta)=2\cos\theta$ and noting the reflection component of $\mbo(2)$ contributes $\tr U=0$,
\begin{equation}\label{eq:o2moment}
\int_{\mbo(2)}(\tr U)^{2k}\,\dd U=\frac12\cdot\frac{1}{2\pi}\int_0^{2\pi}(2\cos\theta)^{2k}\dd\theta=\frac12\,2^{2k}\,\frac{1}{2^{2k}}\binom{2k}{k}=\frac12\binom{2k}{k}.
\end{equation}
For $k=3$ this gives $\tfrac12\binom{6}{3}=10<15$: at $d=2$ the fifteen order-three Brauer diagrams span only a ten-dimensional commutant, so $\ker\phi_2$ has dimension five. This is the source of every $d=2$ degeneracy below.

\subsection{The order-two and order-three commutants}
\label{app:comm23}

For $k=2$ the three diagrams give
\begin{equation}
\comm(\mbo(d),2)=\mathrm{span}\{\id,\ \mbs,\ \om\},
\end{equation}
where $\mbs$ is the swap and $\om=\sum_{i,j}\ket{ii}\bra{jj}$ is the (unnormalized) projector onto the maximally entangled vector $\ket\Omega=\sum_i\ket{ii}$. By \eqref{eq:o2moment} with $k=2$, $\dim\comm(\mbo(2),2)=\tfrac12\binom42=3$, so these remain independent even at $d=2$: the order-two calculus has no singularity.

For $k=3$ the fifteen diagrams give the basis used throughout,
\begin{equation}\label{eq:comm3basis}
\comm(\mbo(d),3)=\mathrm{span}\big\{\underbrace{\mbs_\pi}_{\pi\in S_3}\big\}\cup\big\{\underbrace{\Omega_{ab;xy}}_{9\text{ contractions}}\big\},
\end{equation}
where $\mbs_\pi\ket{i_1i_2i_3}=\ket{i_{\pi^{-1}(1)}i_{\pi^{-1}(2)}i_{\pi^{-1}(3)}}$ is the permutation operator and
\begin{equation}
\Omega_{ab;xy}=\sum_{i,j}\ket{i}_a\ket{i}_b\bra{j}_x\bra{j}_y\ \ (\text{identity on the remaining factor})
\end{equation}
carries a cup on output legs $a,b$ and a cap on input legs $x,y$. We order the basis as
\begin{equation}\label{eq:basisorder}
\begin{aligned}
\big(&\mbs_{e},\mbs_{(23)},\mbs_{(12)},\mbs_{(13)},\mbs_{(132)},\mbs_{(123)};\\
&\ \Omega_{12;12},\Omega_{23;23},\Omega_{13;13},\Omega_{12;23},\Omega_{23;12},\Omega_{13;23},\Omega_{13;12},\Omega_{23;13},\Omega_{12;13}\big),
\end{aligned}
\end{equation}
and denote these $x_1,\dots,x_{15}$. A graphical calculus for composing and tracing these operators is given in~\cite{mele2023introduction} for permutation diagrams and in~\cite{garciamartin2025quantum} for Brauer diagrams, and the partial-trace identities we need are collected in \cref{app:shadow_norm_proof} (\cref{tab:trace_rel}).

\subsection{The Gram matrix and the Weingarten function}
\label{app:gram}

The Haar average $\mct^{(k)}_{\mbo(d)}(A)=\int_{\mbo(d)}U^{\otimes k}A\,U^{\intercal\otimes k}\dd U$ is the orthogonal projector (in the Hilbert--Schmidt inner product) onto $\comm(\mbo(d),k)$~\cite{collins2006integration,garciamartin2025quantum}. Expanding in the basis $\{x_i\}$ and imposing $\langle x_i,\mct^{(k)}(A)\rangle=\langle x_i,A\rangle$ gives the Weingarten formula
\begin{equation}\label{eq:wein}
\mct^{(k)}_{\mbo(d)}(A)=\sum_{i,j}\Wg_{ij}\,\langle x_i,A\rangle\,x_j,\qquad \Wg=G(d)^{+},
\end{equation}
where $G(d)^{+}$ is the Moore--Penrose pseudoinverse (the ordinary inverse when $d\ge k$) of the \emph{Gram matrix}
\begin{equation}\label{eq:gramdef}
G(d)_{ij}=\langle x_i,x_j\rangle_{\mathrm{HS}}=\tr\!\big[\phi_d(x_i)^\dagger\phi_d(x_j)\big]=d^{\,\ell(i,j)},
\end{equation}
with $\ell(i,j)$ the number of closed loops in the composite diagram $\bar x_i\!\circ x_j$~\cite{collins2006integration,collins2009some}. This last equality makes the Gram matrix a matrix of monomials in $d$, and we prove it.
\begin{lemma}[Loop rule]\label{lem:loop}
For Brauer diagrams $D_M,D_N$ on $2k$ nodes, $\tr[\phi_d(D_M)^\dagger\phi_d(D_N)]=d^{\ell(M,N)}$, where $\ell(M,N)$ is the number of connected components (closed loops) of the diagram obtained by stacking $D_M$ (reflected) on $D_N$ and joining matched endpoints.
\end{lemma}
\begin{proof}
Each diagram is a product of cups $\sum_i\ket{ii}$, caps $\sum_j\bra{jj}$, and through-lines $\delta$, so $\phi_d(D)$ is a product of Kronecker deltas indexed by the arcs of $D$. Forming $\tr[\phi_d(D_M)^\dagger\phi_d(D_N)]$ contracts every free index against its partner, since an index is summed over $\{0,\dots,d-1\}$, and the deltas force all indices lying on one connected component of the stacked diagram to be equal. Each component therefore contributes a single free sum $\sum_{a=0}^{d-1}1=d$, and independent components multiply, giving $d^{\ell(M,N)}$. On the diagonal, the stack of a diagram on itself has $k$ loops through the $k$ arcs plus none extra, i.e.\ $\ell(i,i)=k$. Here $k=3$.
\end{proof}
The exponent matrix $\ell$ is
\begin{equation}\label{eq:loops}
\ell=\left(\begin{smallmatrix}
3&2&2&2&1&1&2&2&2&1&1&1&1&1&1\\
2&3&1&1&2&2&1&2&1&1&1&1&2&1&2\\
2&1&3&1&2&2&2&1&1&1&1&2&1&2&1\\
2&1&1&3&2&2&1&1&2&2&2&1&1&1&1\\
1&2&2&2&3&1&1&1&1&2&1&1&2&2&1\\
1&2&2&2&1&3&1&1&1&1&2&2&1&1&2\\
2&1&2&1&1&1&3&1&1&2&2&1&2&1&2\\
2&2&1&1&1&1&1&3&1&2&2&2&1&2&1\\
2&1&1&2&1&1&1&1&3&1&1&2&2&2&2\\
1&1&1&2&2&1&2&2&1&3&1&2&1&1&2\\
1&1&1&2&1&2&2&2&1&1&3&1&2&2&1\\
1&1&2&1&1&2&1&2&2&2&1&3&2&1&1\\
1&2&1&1&2&1&2&1&2&1&2&2&3&1&1\\
1&1&2&1&2&1&1&2&2&1&2&1&1&3&2\\
1&2&1&1&1&2&2&1&2&2&1&1&1&2&3
\end{smallmatrix}\right),
\end{equation}
in the ordering \eqref{eq:basisorder} (the diagonal is $\ell(i,i)=3$ since $\tr\phi_d(x_i)^\dagger\phi_d(x_i)=d^3$). We compute every orthogonal third moment in this paper from \eqref{eq:wein}--\eqref{eq:loops}.

\subsection{The determinant and its representation-theoretic content}
\label{app:det}

Over the integers, the Gram determinant of $G(d)=\big(d^{\ell(i,j)}\big)$ factors as
\begin{equation}\label{eq:detfac}
\det G(d)=d^{15}\,(d-1)^{14}\,(d-2)^{5}\,(d+2)^{10}\,(d+4),
\end{equation}
which we verified by exact symbolic computation from \eqref{eq:loops}. Each root of \eqref{eq:detfac} is a value of the loop parameter $\delta=d$ at which the trace bilinear form on $\Brauer_3(\delta)$ degenerates, i.e.\ at which $\Brauer_3(\delta)$ fails to be semisimple~\cite{wenzl1988structure,rui2005criterion}. For the physically relevant integers $d\ge2$ the only vanishing factor is $(d-2)^5$, and the negative roots $d=-2,-4$ (and $d=1$) lie outside the tensor-space regime and reflect the general non-semisimplicity locus of $\Brauer_3$~\cite{rui2005criterion}. Its exponent $5$ of $(d-2)$ is exactly $\dim\ker\phi_2=15-\tfrac12\binom63=5$ from \eqref{eq:o2moment}: it is the dimension of the radical of the form, equivalently the number of independent linear relations among the fifteen diagrams at $d=2$.

\subsection{The singularity at \texorpdfstring{$d=2$}{d=2}}
\label{app:d2}

At $d=2$ the rank of $G(2)$ is $10$ and $\ker G(2)=\ker\phi_2=\Rad$, the radical of the trace form, of dimension $5$ (\cref{eq:detfac}, \cref{thm:radical}). We identify this kernel completely.

\paragraph{The antisymmetric relation.}
Under $\mbo(d)$, $V^{\otimes3}$ contains the totally antisymmetric subspace $\Alt^3(\Complex^d)$ of dimension $\binom d3$. The Brauer antisymmetrizer
\begin{equation}\label{eq:A3}
A_3=\sum_{\pi\in S_3}\sgn(\pi)\,\mbs_\pi=\id-\mbs_{(12)}-\mbs_{(13)}-\mbs_{(23)}+\mbs_{(123)}+\mbs_{(132)}
\end{equation}
maps under $\phi_d$ to $3!$ times the projector onto $\Alt^3(\Complex^d)$. At $d=2$, $\binom23=0$, so $\Alt^3(\Complex^2)=\{0\}$ and $\phi_2(A_3)=0$, since the wedge $e_{i_1}\wedge e_{i_2}\wedge e_{i_3}$ of three vectors in a two-dimensional space vanishes. In the ordering \eqref{eq:basisorder} this is the null vector
\begin{equation}
\mathbf n_1=(1,-1,-1,-1,1,1,\,0,0,0,0,0,0,0,0,0)^\intercal .
\end{equation}

\paragraph{The four contraction relations.}
Remaining relations are specific to $d=2$ and originate in the $\mbo(2)$ Levi-Civita tensor $\epsilon_{ab}$ ($\epsilon_{01}=-\epsilon_{10}=1$), which satisfies the Fierz identity
\begin{equation}\label{eq:fierz}
\epsilon_{ab}\,\epsilon_{cd}=\delta_{ac}\delta_{bd}-\delta_{ad}\delta_{bc}\qquad(d=2).
\end{equation}
Substituting \eqref{eq:fierz} for a cup--cap pair rewrites each $\Omega$ diagram as a combination of a permutation diagram and an $\epsilon\epsilon$ term with no counterpart for $d\ge3$; these $\epsilon\epsilon$ terms cancel in exactly four independent combinations, collapsing the fifteen diagrams onto the ten-dimensional commutant. Computing the null space of $G(2)$ by exact rational arithmetic yields the following basis of $\ker G(2)$ (verified by direct substitution of the matrix elements \eqref{eq:mat-elts}):
\begin{align}
\mathbf n_1&=(1,-1,-1,-1,1,1,\,0,0,0,0,0,0,0,0,0)^\intercal,\nn
\mathbf n_2&=(1,0,0,-1,0,0,\,-1,-1,0,1,1,0,0,0,0)^\intercal,\nn
\mathbf n_3&=(0,0,1,0,-1,0,\,-1,0,0,1,0,-1,1,0,0)^\intercal,\nn
\mathbf n_4&=(1,0,-1,0,0,0,\,0,-1,-1,0,0,1,0,1,0)^\intercal,\nn
\mathbf n_5&=(1,-1,-1,0,1,0,\,0,0,-1,-1,0,1,0,0,1)^\intercal .\label{eq:nullvecs}
\end{align}
Each $\mathbf n_i$ encodes an operator identity in $\End((\Complex^2)^{\otimes3})$. For instance $\mathbf n_2$ states
\begin{equation}
\id-\mbs_{(13)}-\Omega_{12;12}-\Omega_{23;23}+\Omega_{12;23}+\Omega_{23;12}=0 ,
\end{equation}
which we check on the matrix elements
\begin{equation}\label{eq:mat-elts}
\bra{j_1j_2j_3}\mbs_\pi\ket{i_1i_2i_3}=\prod_r\delta_{j_r,i_{\pi^{-1}(r)}},\qquad
\bra{j_1j_2j_3}\Omega_{ab;xy}\ket{i_1i_2i_3}=\delta_{j_a j_b}\,\delta_{i_x i_y}\,\delta_{j_e i_f},
\end{equation}
where $e$ is the leftover output leg and $f$ the leftover input leg. Note that the cup, sitting on the output legs $a,b$, constrains the \emph{output} indices $j$, and the cap on the input legs $x,y$ the \emph{input} indices $i$, in accordance with the definition of $\Omega_{ab;xy}$ in \cref{app:comm23}. We check the identity by exhausting the finitely many $i,j\in\{0,1\}^3$ and using that any term with a factor $\delta_{j_a j_b}$, $j_a\neq j_b$, vanishes. A single index pair shows the bookkeeping. Take $j=(1,1,0)$ and $i=(0,0,0)$. Then $\id$ and $\mbs_{(13)}$ both need $j_1=i_1$ and vanish; $\Omega_{23;23}$ and $\Omega_{23;12}$ need $j_2=j_3$ and vanish; while
\begin{equation}\label{eq:mat-elts-instance}
\bra{110}\Omega_{12;12}\ket{000}=\delta_{j_1j_2}\delta_{i_1i_2}\delta_{j_3i_3}=1,
\qquad
\bra{110}\Omega_{12;23}\ket{000}=\delta_{j_1j_2}\delta_{i_2i_3}\delta_{j_3i_1}=1 ,
\end{equation}
the leftover legs being $(e,f)=(3,3)$ in the first case and $(3,1)$ in the second. These enter $\mathbf n_2$ with opposite signs, so the entry cancels. For the all-zero pair $j=i=(0,0,0)$ every term equals $1$, and the cancellation is $1-1-1-1+1+1=0$.

\begin{theorem}[Radical of $\Brauer_3(2)$]
\label{thm:radical}
$\Brauer_3(2)$ is not semisimple, and the radical of its trace form has dimension $5$ and equals $\ker\phi_2$, so $\rank G(2)=10=\dim\comm(\mbo(2),3)$.
\end{theorem}
\begin{proof}
Two bounds meet at $5$. For the upper bound, the entries of $G(d)$ are polynomials in $d$ and $\det G(d)\not\equiv0$, so $G(d)$ has a Smith normal form $\diag(f_1,\dots,f_{15})$ over the principal ideal domain $\Complex[d]$, with $\det G(d)=\prod_i f_i$ up to a nonzero constant. Each $f_i$ vanishing at $d=2$ contributes at least $1$ to the order of vanishing of $\det G(d)$ there, and that order is $5$ by \eqref{eq:detfac}, so at most five of the $f_i$ vanish at $d=2$, and $\dim\ker G(2)\le5$. For the lower bound, the five vectors \eqref{eq:nullvecs} are linearly independent and each lies in $\ker G(2)$, by the operator identity it encodes, so $\dim\ker G(2)\ge5$. Together the two bounds give $\dim\ker G(2)=5$ and $\rank G(2)=10$. Since $\ker G(2)=\ker\phi_2$ is the radical of the trace form, that radical has dimension $5$, and $\rank G(2)=\dim\comm(\mbo(2),3)$ by \cref{thm:sw}.
\end{proof}
\noindent Both halves are ours, and neither appeals to the classification of the parameters at which $\Brauer_k(\delta)$ fails to be semisimple. We read the upper bound off the determinant \eqref{eq:detfac}. We read the lower bound off the explicit basis \eqref{eq:nullvecs}, obtained by exact rational row reduction of $G(2)$ and confirmed by verifying each of the five operator identities it encodes, as illustrated above for $\mathbf n_2$. This agrees with the semisimplicity criteria of~\cite{rui2005criterion,wenzl1988structure}, cited here for priority. It is easy to invert the direction of the determinant bound, so we state it explicitly. Vanishing order of $\det G(d)$ bounds the nullity from \emph{above}, not below, as $\diag\big((d-2)^5,1,1,1,1\big)$ shows---its determinant carries the same factor while its nullity at $d=2$ is $1$. Instead the explicit basis supplies the matching lower bound.

\begin{figure}[t]
\centering
\includegraphics[width=0.59\linewidth]{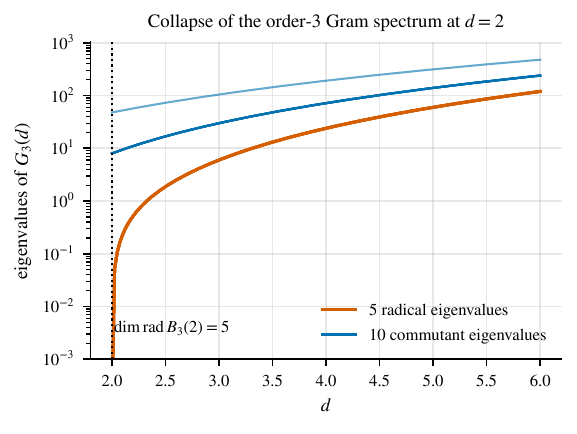}
\caption{The fifteen eigenvalues of the order-three Gram matrix $G_3(d)=\big(d^{\ell(i,j)}\big)$. Exactly five (red) vanish as $d\to2$, the analytic signature of the five-dimensional radical of $\Brauer_3(2)$ and of the factor $(d-2)^5$ in $\det G_3(d)$, and the other ten (blue) stay bounded away from zero and span $\comm(\mbo(2),3)$.}
\label{fig:gram_spectrum}
\end{figure}

\paragraph{Solvability of the physical systems.}
Although $G(2)$ is singular, the linear systems arising in the protocol are always consistent, and for a reason that requires no property whatsoever of the target. Let $T\in\End((\Complex^2)^{\otimes3})$ be arbitrary and $c_i=\langle x_i,T\rangle$ the corresponding right-hand side of \eqref{eq:wein}. If $\mathbf n\in\ker G(2)$ then, as established above, $\ker G(2)=\ker\phi_2$, so $\sum_i n_i x_i=0$ as an operator, and therefore
\begin{equation}\label{eq:consistency}
\mathbf n\cdot c=\sum_i n_i\langle x_i,T\rangle=\Big\langle \sum_i n_i x_i,\,T\Big\rangle=\langle 0,T\rangle=0 .
\end{equation}
Thus $c\perp\ker G(2)$, since $G(2)$ is symmetric, $\im G(2)=(\ker G(2))^{\perp}$, so $c\in\im G(2)$ and the system $G(2)x=c$ is solvable for every $T$. With the pseudoinverse, \eqref{eq:wein} returns the minimum-norm solution. Any other solution differs by an element of $\ker\phi_2$ and hence gives the same operator $\mct^{(3)}(T)$, the unique element of $\comm(\mbo(2),3)$ with the prescribed overlaps. All $d=2$ (single-qubit, local) results are therefore well defined.

It is nonetheless instructive to see what \eqref{eq:consistency} says about the targets that actually arise. Both the channel derivation and the seminorm derivation use $T=A\otimes B\otimes B$ with $A=\channel^*(\Pi_b)$, $B=\Pi_b$ (for $k=2$ the target is $A\otimes B$). Evaluating $\mathbf n_2\cdot c$ and $\mathbf n_4\cdot c$ (the only two of the five that involve the $\Omega$ diagrams non-trivially in this way) against the trace data of \cref{tab:trace_rel} turns \eqref{eq:consistency} into the identity
\begin{equation}\label{eq:solv}
s_5=s_2+s_3-s_4-(1-\alpha_b)\,s_1,
\end{equation}
in the notation $s_1=\tr[\channel^*(\Pi_b)]$, $s_2=\bra b\channel^*(\Pi_b)\ket b$, $s_3=\bra{b^*}\channel^*(\Pi_b)\ket{b^*}$, $s_4=\tr[\channel^*(\Pi_b)\Pi_b^\intercal\Pi_b]$, $s_5=\tr[\channel^*(\Pi_b)\Pi_b\Pi_b^\intercal]$, $\alpha_b=|\braket{b}{b^*}|^2$. For a computational-basis measurement $\Pi_b^\intercal=\Pi_b$, $\alpha_b=1$, and \eqref{eq:solv} holds identically ($s_4=s_5$, $s_2=s_3$).

\paragraph{Finiteness of the coefficients at $d=2$.}
Consistency alone does not guarantee that the twirl coefficients stay finite as $d\to2$, and we show they do. Writing $t_1=\tr[\channel^*(\Pi_b)]$, $t_2=\bra b\channel^*(\Pi_b)\ket b$, the overlap vector for $T=\channel^*(\Pi_b)\otimes\Pi_b\otimes\Pi_b$ is, by \cref{tab:trace_rel}, $r_i=t_1$ for $x_i\in\{\id,\mbs_{(23)},\Omega_{23;23}\}$ and $r_i=t_2$ for the remaining twelve $x_i$. This two-block partition is preserved by the symmetry of the target, since $T$ is fixed by $\mbs_{(23)}$ and by the cup--cap exchange on the last two legs, and these operations permute the fifteen diagrams while preserving both blocks, so the overlap vector $r$ and each row sum of $G(d)$ restricted to a block are constant across the block. (The two blocks are not a single group orbit, and they mix the through-numbers, but they are the coarsest partition on which $r$ and the relevant Gram row sums are simultaneously constant, which is all the ansatz needs.) We may therefore seek a solution constant on each block, $\mct^{(3)}(T)=c_A\!\!\sum_{x_i\in\text{block }A}\!\!x_i+c_B\!\!\sum_{x_i\in\text{block }B}\!\!x_i$. Substituting into $G(d)\,x=r$, every $A$-row of $G(d)$ has the same sum over the $A$-columns and over the $B$-columns, and likewise every $B$-row, so the fifteen equations reduce to the $2\times2$ block system
\begin{equation}\label{eq:block2x2}
\begin{pmatrix}S_{AA}&S_{AB}\\ S_{BA}&S_{BB}\end{pmatrix}\!\begin{pmatrix}c_A\\ c_B\end{pmatrix}=\begin{pmatrix}t_1\\ t_2\end{pmatrix},\qquad
\begin{aligned}
S_{AA}&=d^2(d+2), & S_{AB}&=4d(d+2),\\
S_{BA}&=d(d+2), & S_{BB}&=d(d+2)(d+3),
\end{aligned}
\end{equation}
where $S_{XY}$ is the (block-independent) sum of the Gram entries from a row in block $X$ over the columns in block $Y$, computed from $G(d)=(d^{\ell(i,j)})$. Its determinant is $S_{AA}S_{BB}-S_{AB}S_{BA}=d^2(d+2)^2(d-1)(d+4)$, so Cramer's rule gives, for all fifteen rows simultaneously,
\begin{equation}\label{eq:cAcB}
c_A=\frac{(d+3)\,t_1-4\,t_2}{p_D(d)},\qquad
c_B=\frac{d\,t_2-t_1}{p_D(d)},\qquad
p_D(d)=d(d-1)(d+2)(d+4),
\end{equation}
with exactly zero residual, verified symbolically. Its denominator is $p_D(d)$, which contains no factor $(d-2)$, because although $\det G(d)\propto(d-2)^5$ vanishes at $d=2$, the physical solution is obtained by projecting $r$ onto $\mathrm{row}(G)=\comm$, and this projection is smooth through the degeneracy. Hence the coefficients are finite at $d=2$,
\begin{equation}\label{eq:cAcB2}
c_A(2)=\frac{5t_1-4t_2}{48}=\frac{5t_1}{48}-\frac{t_2}{12},\qquad
c_B(2)=\frac{2t_2-t_1}{48}=-\frac{t_1}{48}+\frac{t_2}{24},
\end{equation}
and substituting \eqref{eq:cAcB2} back reproduces $r(2)$ exactly. So the $d=2$ (single-qubit, hence local) shadow channel and seminorm are given by explicit finite formulas, on which the local results of Part~II rely.

Our two-block argument uses the reality of the basis twice: through $\Omega_{23;23}\mapsto t_1\tr[\Pi_b\Pi_b^\intercal]=t_1$, placing $\Omega_{23;23}$ in block $A$, and through $s_2=s_3$, $s_4=s_5$. For a complex basis, neither holds, and the overlap vector is no longer two-valued, and \eqref{eq:cAcB}--\eqref{eq:cAcB2} do not apply. Since Part~IV needs precisely that case, we treat it separately next.

\subsection{The \texorpdfstring{$d=2$}{d=2} third moment for a complex basis}
\label{app:d2complex}

Part~IV requires the single-qubit third moment for a basis $\mcw=\{\ket{w_1},\ket{w_2}\}$ that need not be real. This is the one place where the closed form \eqref{eq:complex_R} of \cref{app:complex_third_moment} is unavailable, since its coefficients carry the denominator $D=d(d-2)(d-1)(d+2)(d+4)$, which vanishes at $d=2$. We therefore compute the object Part~IV actually needs, the contracted operator $R$, not the fifteen coefficients---directly at $d=2$, and find that it is finite and, remarkably, a multiple of the identity.

Three single-qubit facts drive the computation.

\begin{lemma}[Equal reality at $d=2$]\label{lem:equalreality}
For any orthonormal basis $\{\ket{w_1},\ket{w_2}\}$ of $\Complex^2$ we have $\alpha_{w_1}=\alpha_{w_2}=\tfrac12\alpha_{\mathrm r}$. Consequently, for any trace-preserving or unital single-qubit channel $\channel_1$,
\begin{equation}\label{eq:gamma_d2}
\gamma_1\;=\;\sum_w\tr[\channel_1^*(\Pi_w)]\,\alpha_w\;=\;\tfrac12\alpha_{\mathrm r}\,\tr[\channel_1(\id)]\;=\;\alpha_{\mathrm r},
\end{equation}
whether or not $\channel_1$ is unital.
\end{lemma}
\begin{proof}
Collect the basis vectors as the columns of a unitary $W$, so $\ket{w_k}=W\ket k$. Then $\braket{w_k}{w_k^*}=\sum_j\overline{W_{jk}}\,\overline{W_{jk}}=\overline{(W^{\intercal}W)_{kk}}$, so $\alpha_{w_k}=|M_{kk}|^2$ with $M=W^\intercal W$. Here $M$ is unitary and symmetric, $M=\big(\begin{smallmatrix}a&b\\ b&c\end{smallmatrix}\big)$; unitarity of its two columns gives $|a|^2+|b|^2=1=|b|^2+|c|^2$, hence $|a|=|c|$ and $\alpha_{w_1}=\alpha_{w_2}$. Their sum is $\alpha_{\mathrm r}$, so each equals $\tfrac12\alpha_{\mathrm r}$. For \eqref{eq:gamma_d2}, $\tr[\channel_1^*(\Pi_w)]=\tr[\Pi_w\channel_1(\id)]=\bra w\channel_1(\id)\ket w$, so pulling out the common factor $\alpha_w=\tfrac12\alpha_{\mathrm r}$ leaves $\sum_w\bra w\channel_1(\id)\ket w=\tr[\channel_1(\id)]=\alpha=2$.
\end{proof}
This is special to $d=2$, since for $d\ge3$ the individual $\alpha_w$ differ and $\gamma$ is an independent invariant. It collapses the six invariants $\alpha,\beta,\tilde\beta,\gamma,\delta,\bar\delta$ of \eqref{eq:scalars}--\eqref{eq:complex_invariants} to fewer at $d=2$: summing \eqref{eq:solv} over $w$ gives, in addition, the linear relation
\begin{equation}\label{eq:d2_delta}
2\,\mathrm{Re}\,\delta\;=\;\beta+\tilde\beta-\alpha+\gamma\qquad(d=2),
\end{equation}
which has no analogue for $d\ge3$ and which is exactly the shadow cast by the five-dimensional radical of \cref{thm:radical} on the invariants.

\begin{lemma}[The inverse acts by a scalar on single-qubit Paulis]\label{lem:lambda}
Let $\mcm_1=\mcm_{\mbo(2),\channel_1,\mcw}=\mcd_{1,f_1}\circ\Phi_{q}$ be the channel of \cref{prop:local_complex}, with $\Phi_q(A)=qA+(1-q)A^\intercal$, $q=q_{\beta_1}$, and $f_1\neq0$, $q\neq\tfrac12$. Let $P\in\{X,Y,Z\}$, so $P$ is traceless with $P^2=\id$ and $P^\intercal=\eta P$, $\eta=+1$ for $P\in\{X,Z\}$ and $\eta=-1$ for $P=Y$. Then
\begin{equation}\label{eq:lambda}
\hat P\;:=\;\mcm_1^{-1,\dagger}(P)\;=\;\lambda_\eta\,P,\qquad
\lambda_{+1}=\frac{1}{f_1},\qquad \lambda_{-1}=\frac{1}{f_1(2q-1)},
\end{equation}
and $2q-1=(\beta_1-\tilde\beta_1)/(2f_1)$.
\end{lemma}
\begin{proof}
For all $X,A$ we have $\tr[X^\dagger A^\intercal]=\tr[(X^\dagger A^\intercal)^\intercal]=\tr[\overline X A]=\langle X^\intercal,A\rangle$, so $\Phi_q^\dagger=\Phi_q$ for real $q$; and $\mcd_{1,f}^\dagger=\mcd_{1,f}$. Hence $\mcm_1^\dagger=\Phi_q\circ\mcd_{1,f_1}$ and $\mcm_1^{-1,\dagger}=(\mcm_1^\dagger)^{-1}=\mcd_{1,1/f_1}\circ\Phi_q^{-1}$. Since $P^\intercal=\eta P$,
\begin{equation}
\Phi_q^{-1}(P)=\frac{qP-(1-q)P^\intercal}{2q-1}=\frac{q-\eta(1-q)}{2q-1}\,P,
\end{equation}
which is $P$ for $\eta=+1$ and $\tfrac{1}{2q-1}P$ for $\eta=-1$. Applying $\mcd_{1,1/f_1}$ to the traceless operator $P$ multiplies by $1/f_1$. Finally $q-\tfrac12=\tfrac{\beta_1-\tilde\beta_1}{\beta_1+\tilde\beta_1-2}$ from \eqref{eq:Atilde} at $d=2$, and $\beta_1+\tilde\beta_1-2=4f_1$.
\end{proof}

Together the two lemmas give the single-qubit factor Part~IV needs, and at $d=2$ the answer is
finite where \cref{prop:complex_variance} is singular.

\begin{proposition}[Single-qubit complex-basis seminorm factors]\label{prop:d2complex}
Let $\channel_1$ be a trace-preserving or unital single-qubit channel, $\mcw$ a basis of reality $\alpha_{\mathrm r}$ with $f_1\neq0$ and $\beta_1\neq\tilde\beta_1$, and $P\in\{X,Y,Z\}$ with $P^\intercal=\eta P$. Then the contracted third moment
\begin{equation}\label{eq:Rj_d2}
R\;=\;\sum_w\tr_{23}\!\Big[\mct^{(3)}_{\mbo(2)}\!\big(\channel_1^*(\Pi_w)\otimes\Pi_w\otimes\Pi_w\big)\big(\id\otimes\hat P\otimes\hat P\big)\Big]
\end{equation}
is a multiple of the identity, $R=c_1(P;\mcw)\,\id$, with the finite value
\begin{equation}\label{eq:c1_closed}
c_1(P;\mcw)=
\begin{cases}
\dfrac{\alpha_{\mathrm r}}{4\,f_1^{2}}, & P^\intercal=+P \quad (P\in\{X,Z\}),\\[3mm]
\dfrac{2-\alpha_{\mathrm r}}{2\,f_1^{2}\,(2q_{\beta_1}-1)^{2}}, & P^\intercal=-P \quad (P=Y).
\end{cases}
\end{equation}
In particular $c_1$ is independent of the state, so the maximization over $\sigma$ in \eqref{eq:seminorm_def} is vacuous and \eqref{eq:c1_closed} is exact rather than an optimized bound.
\end{proposition}
\begin{proof}
Write $\Sigma=\sum_w\mct^{(3)}_{\mbo(2)}(\channel_1^*(\Pi_w)\otimes\Pi_w^{\otimes2})=\sum_i c_i x_i$ in the Brauer basis \eqref{eq:basisorder}, and the coefficients $c_i$ exist and are finite by the solvability argument above, and are independent of which particular solution is chosen because $\ker G(2)=\ker\phi_2$. By \cref{lem:lambda}, $\hat P=\lambda_\eta P$, so the right-hand column of \cref{tab:trace_rel} with $A=\id$ and $B=\hat P=\lambda_\eta P$ evaluates every diagram using only $\tr\hat P=0$, $\hat P^2=\lambda_\eta^2\id$, $\hat P^\intercal=\eta\hat P$:
\begin{equation}
\begin{aligned}
\id&\mapsto\id\,\tr[\hat P]^2=0, &
\mbs_{(23)}&\mapsto\id\,\tr[\hat P^2]=2\lambda_\eta^2\,\id,\\
\mbs_{(12)},\mbs_{(13)}&\mapsto\hat P\,\tr[\hat P]=0, &
\mbs_{(123)},\mbs_{(132)}&\mapsto\hat P^2=\lambda_\eta^2\id,\\
\Omega_{12;12},\Omega_{13;13}&\mapsto\tr[\hat P]\,\hat P^\intercal=0, &
\Omega_{23;23}&\mapsto\id\,\tr[\hat P\hat P^\intercal]=2\eta\lambda_\eta^2\id,\\
\Omega_{12;23},\Omega_{13;23}&\mapsto\hat P^\intercal\hat P=\eta\lambda_\eta^2\id, &
\Omega_{23;12},\Omega_{23;13}&\mapsto\hat P\hat P^\intercal=\eta\lambda_\eta^2\id,\\
\Omega_{13;12},\Omega_{12;13}&\mapsto(\hat P^\intercal)^2=\lambda_\eta^2\id.
\end{aligned}
\end{equation}
Every surviving term is a multiple of $\id$, using only that a single-qubit Pauli squares to the identity and is an eigenvector of the transpose; hence $R$ is state-independent at $d=2$. Collecting them,
\begin{equation}\label{eq:R_SpSm}
R=\lambda_\eta^{2}\big(S_{+}+\eta\,S_{-}\big)\,\id,\qquad
\begin{aligned}
S_{+}&=2c_{2}+c_{5}+c_{6}+c_{13}+c_{15},\\
S_{-}&=2c_{8}+c_{10}+c_{11}+c_{12}+c_{14},
\end{aligned}
\end{equation}
in the numbering of \eqref{eq:basisorder}. Because $R$ is determined by $\Sigma$ and \eqref{eq:R_SpSm} holds for both $\eta=\pm1$, the two combinations $S_\pm$ are themselves independent of the choice of solution. In the notation of \eqref{eq:solv}, the overlap vector for a complex basis is
\begin{equation}\label{eq:r_d2}
r=\big(s_1,\,s_1,\,s_2,\,s_2,\,s_2,\,s_2;\ s_3,\,\alpha_ws_1,\,s_3,\,s_4,\,s_5,\,s_4,\,s_3,\,s_5,\,s_3\big)
\end{equation}
in the ordering \eqref{eq:basisorder}: the two diagrams with no measurement leg contracted give $s_1$, the four remaining permutations $s_2$, the four $\Omega$ diagrams that transpose both measurement legs $s_3$, the cap on the two measurement legs $\alpha_ws_1$, and the four diagrams coupling slot $1$ to a single measurement leg $s_4$ or $s_5$ according to orientation. Solving $G(2)x=r$ and forming the two combinations of \eqref{eq:R_SpSm} gives
\begin{equation}\label{eq:SpSm_raw}
S_{+}=\frac{s_1}{6}+\frac{s_2+s_3-s_4-s_5}{12}-\frac{\alpha_ws_1}{24},\qquad
S_{-}=-\frac{s_1}{12}-\frac{s_2+s_3-s_4-s_5}{6}+\frac{5\,\alpha_ws_1}{24}.
\end{equation}
Summing over $w$ replaces $s_1,s_2,s_3,s_4,s_5,\alpha_ws_1$ by $\alpha,\beta_1,\tilde\beta_1,\delta,\bar\delta,\gamma_1$ respectively, so that with $\alpha=2$,
\begin{equation}\label{eq:SpSm_mid}
S_{+}=\frac13+\frac{\beta_1+\tilde\beta_1-2\,\mathrm{Re}\,\delta}{12}-\frac{\gamma_1}{24},\qquad
S_{-}=-\frac16-\frac{\beta_1+\tilde\beta_1-2\,\mathrm{Re}\,\delta}{6}+\frac{5\gamma_1}{24}.
\end{equation}
Now the noise scalars cancel, and it is \eqref{eq:d2_delta} that cancels them: substituting $2\,\mathrm{Re}\,\delta=\beta_1+\tilde\beta_1-2+\gamma_1$ leaves
\begin{equation}\label{eq:SpSm}
S_{+}=\frac{4-\gamma_1}{8},\qquad S_{-}=\frac{3\gamma_1-4}{8}\qquad(\alpha=2),
\end{equation}
so the contracted third moment at $d=2$ depends on the noise and the basis only through $\gamma_1$. This is not an accident of the single-qubit case so much as the radical of \cref{thm:radical} asserting itself, since \eqref{eq:d2_delta} is a relation among the invariants that holds only at $d=2$, and it is exactly what removes $\beta_1$ and $\tilde\beta_1$ here. With $\gamma_1=\alpha_{\mathrm r}$ from \cref{lem:equalreality},
\begin{equation}
S_{+}+S_{-}=\frac{\alpha_{\mathrm r}}{4},\qquad S_{+}-S_{-}=\frac{2-\alpha_{\mathrm r}}{2}.
\end{equation}
Substituting into \eqref{eq:R_SpSm} together with \eqref{eq:lambda} gives \eqref{eq:c1_closed}. Every denominator is a product of $f_1$ and $2q_{\beta_1}-1$, both nonzero by hypothesis, so $c_1$ is finite. No factor of $(d-2)$ appears.
\end{proof}

Two special cases fix the result. A real basis has $\alpha_{\mathrm r}=2$, whence $c_1(X)=c_1(Z)=\tfrac{2}{4f_1^2}=(2f_1^2)^{-1}$, exactly \cref{prop:pauli_seminorm}, while the numerator $2-\alpha_{\mathrm r}$ of the $Y$ case vanishes together with $2q_{\beta_1}-1$, and the correct reading there is not $0/0$ but that $Y$ lies in the kernel, so $\mcm_1^{-1,\dagger}(Y)$ does not exist and $c_1(Y)=\infty$ by the convention of \cref{sec:local_complex_seminorm}. And for single-qubit depolarizing noise $\mcd_{1,p}$ on a basis of reality $\alpha_{\mathrm r}$ we have $\beta_1=1+p$ and $\tilde\beta_1=p\alpha_{\mathrm r}+1-p$, hence $\beta_1-\tilde\beta_1=p(2-\alpha_{\mathrm r})$, $f_1=\tfrac{p\,\alpha_{\mathrm r}}{4}$ and $2q_{\beta_1}-1=\tfrac{\beta_1-\tilde\beta_1}{2f_1}=\tfrac{2(2-\alpha_{\mathrm r})}{\alpha_{\mathrm r}}$, and \eqref{eq:c1_closed} collapses to the fully explicit pair
\begin{equation}\label{eq:c1_depol}
c_1(X)=c_1(Z)=\frac{4}{p^{2}\alpha_{\mathrm r}},\qquad
c_1(Y)=\frac{2}{p^{2}\,(2-\alpha_{\mathrm r})} ,
\end{equation}
which exhibits the trade-off of Part~IV in closed form, in that as the per-qubit basis is made less real ($\alpha_{\mathrm r}$ decreasing from $2$) the symmetric factors $c_1(X),c_1(Z)$ grow like $\alpha_{\mathrm r}^{-1}$ while the previously infinite $c_1(Y)$ falls like $(2-\alpha_{\mathrm r})^{-1}$, and the real-basis values $2/p^2$ and $\infty$ are recovered at $\alpha_{\mathrm r}=2$. \cref{prop:d2complex} and \eqref{eq:c1_depol} were checked against the exact $\mbo(2)$ commutant twirl for depolarizing, amplitude-damping and Haar-random CPTP single-qubit channels on random complex bases.

\section{Proof of \texorpdfstring{\cref{prop:global_noisy_channel}}{the global channel}}
\label{app:global_channel_proof}

\propglobal*

\begin{proof}
Write $\channel^*$ for the bilinear trace dual of \cref{sec:scalars}, $\tr[\channel^*(X)\,Y]=\tr[X\,\channel(Y)]$ for all $X,Y$; this exists for any linear superoperator and needs no Kraus representation, matching the hypothesis of the proposition. (It has to be this dual and not the Hilbert--Schmidt adjoint, for which $\langle A,\channel(B)\rangle=\langle\channel^*(A),B\rangle$ gives $\tr[\channel(X)\Pi_b]=\tr[\channel^*(\Pi_b)^\dagger X]$ and so requires $\channel^*(\Pi_b)$ to be Hermitian --- automatic when $\channel$ preserves Hermiticity, but not for a merely linear $\channel$. The two coincide on the completely positive channels of \cref{sec:noise_model}, so nothing elsewhere in the paper is affected.) From \eqref{eq:noisy_channel} with a real basis and orthogonal $U$ ($U^\dagger=U^\intercal$),
\begin{align}
\mcm_{\mbo,\channel}(A)
&=\int_{\mbo(d)}\!\dd\mu(U)\sum_b \tr\!\big[\channel(UAU^\intercal)\Pi_b\big]\,U^\intercal\Pi_b U
=\int\!\dd\mu(U)\sum_b \tr\!\big[UAU^\intercal\channel^*(\Pi_b)\big]\,U^\intercal\Pi_b U\nn
&=\sum_b\int\!\dd\mu(U)\,\tr\!\big[U^\intercal\channel^*(\Pi_b)U\,A\big]\,U^\intercal\Pi_b U
=\sum_b\tr_1\!\Big[\mct^{(2)}_{\mbo(d)}\!\big(\channel^*(\Pi_b)\otimes\Pi_b\big)(A\otimes\id)\Big],
\label{eq:chan_start}
\end{align}
using cyclicity and the adjoint property in the first line. That last equality is the tensor lift: for any operators $C,B$ and any $U$,
\begin{equation}\label{eq:tensorlift}
\tr\!\big[U^\intercal C U\,A\big]\;U^\intercal B U=\tr_1\!\Big[\big((U^\intercal C U)\otimes(U^\intercal B U)\big)(A\otimes\id)\Big],
\end{equation}
where $\tr_1$ traces out the first tensor factor; averaging over $U$ and using
$(U^\intercal C U)\otimes(U^\intercal B U)=(U^{\intercal})^{\otimes2}(C\otimes B)U^{\otimes2}$
turns the Haar integral into the order-two twirl
\[
\mct^{(2)}_{\mbo(d)}(C\otimes B)=\int\dd\mu(U)\,(U^\intercal)^{\otimes2}(C\otimes B)U^{\otimes2},
\]
with $C=\channel^*(\Pi_b)$, $B=\Pi_b$.

By \eqref{eq:wein}, $\mct^{(2)}_{\mbo(d)}$ projects onto the commutant $\comm(\mbo(d),2)=\mathrm{span}\{\id,\mbs,\om\}$ (the image of the Brauer algebra $\Brauer_2(d)$, of dimension $3$), where $\mbs$ is the swap and $\om=\ketbra{\Omega}{\Omega}$ with $\ket{\Omega}=\sum_i\ket{ii}$ the unnormalized maximally entangled vector. Hence $\mct^{(2)}_{\mbo(d)}(\channel^*(\Pi_b)\otimes\Pi_b)=c_{\id}\id+c_{\mbs}\mbs+c_{\om}\om$, with the coefficients fixed by pairing both sides against the three basis operators under $\langle X,Y\rangle=\tr[X^\dagger Y]$. Pairing (right-hand) traces are
\begin{equation}
\langle\id,\channel^*(\Pi_b)\otimes\Pi_b\rangle=\tr[\channel^*(\Pi_b)]\tr[\Pi_b]=\tr[\channel^*(\Pi_b)],\quad
\langle\mbs,\cdot\rangle=\tr[\channel^*(\Pi_b)\Pi_b]=\bra b\channel^*(\Pi_b)\ket b,
\end{equation}
\begin{equation}
\langle\om,\cdot\rangle=\bra\Omega\big(\channel^*(\Pi_b)\otimes\Pi_b\big)\ket\Omega=\tr[\channel^*(\Pi_b)\Pi_b^\intercal]=\bra b\channel^*(\Pi_b)\ket b,
\end{equation}
where $\langle\mbs,X\otimes Y\rangle=\tr[\mbs(X\otimes Y)]=\tr[XY]$, $\langle\om,X\otimes Y\rangle=\bra\Omega(X\otimes Y)\ket\Omega=\tr[XY^\intercal]$, and the last equality uses the real basis $\Pi_b^\intercal=\Pi_b$. Its order-two Gram matrix has entries $\langle\id,\id\rangle=\tr[\id]=d^2$, $\langle\id,\mbs\rangle=\tr[\mbs]=d$, $\langle\id,\om\rangle=\tr[\om]=\braket{\Omega}{\Omega}=d$, $\langle\mbs,\mbs\rangle=\tr[\mbs^2]=\tr[\id]=d^2$, $\langle\mbs,\om\rangle=\tr[\mbs\,\om]=\bra\Omega\mbs\ket\Omega=\braket{\Omega}{\Omega}=d$ (as $\mbs\ket\Omega=\ket\Omega$), and $\langle\om,\om\rangle=\tr[\om^2]=d\,\tr[\om]=d^2$ (as $\om^2=d\,\om$), i.e.\ in the basis $(\id,\mbs,\om)$,
\begin{equation}
G_2=\begin{pmatrix}d^2&d&d\\ d&d^2&d\\ d&d&d^2\end{pmatrix},
\end{equation}
whose inverse is the second-moment orthogonal Weingarten matrix
\begin{equation}\label{eq:wg2}
\Wg_2=G_2^{-1}=\frac{1}{d(d-1)(d+2)}\begin{pmatrix}d+1&-1&-1\\-1&d+1&-1\\-1&-1&d+1\end{pmatrix}\qquad(d\neq0,1),
\end{equation}
so that, writing the pairing vector as $(t_1,t_2,t_2)$ with $t_1=\tr[\channel^*(\Pi_b)]$ and $t_2=\bra b\channel^*(\Pi_b)\ket b$, the coefficient vector $(c_{\id},c_{\mbs},c_{\om})^\intercal=\Wg_2\,(t_1,t_2,t_2)^\intercal$ is
\begin{equation}\label{eq:coeffs_b}
c_{\mbs}=c_{\om}=\frac{d\,t_2-t_1}{d(d-1)(d+2)}=\frac{d\bra b\channel^*(\Pi_b)\ket b-\tr[\channel^*(\Pi_b)]}{d(d-1)(d+2)},\quad
c_{\id}=\frac{(d+1)t_1-2t_2}{d(d-1)(d+2)}.
\end{equation}
Since $\tr_1[\id(A\otimes\id)]=\tr(A)\id$, $\tr_1[\mbs(A\otimes\id)]=A$, and $\tr_1[\om(A\otimes\id)]=A^\intercal$,
\begin{equation}
\mcm_{\mbo,\channel}(A)=\sum_b\big(c_{\id}\tr(A)\id+c_{\mbs}A+c_{\om}A^\intercal\big)=\sum_b\big(c_{\id}\tr(A)\id+2c_{\mbs}A_{\mathrm{sym}}\big).
\end{equation}
Summing the numerators with \cref{lem:adjoint} ($\sum_b\tr[\channel^*(\Pi_b)]=\alpha$, $\sum_b\bra b\channel^*(\Pi_b)\ket b=\beta$) gives \eqref{eq:global_general}. For trace-preserving or unital $\channel$, $\alpha=d$, so the $A_{\mathrm{sym}}$ coefficient is $2\sum_b c_{\mbs}=\tfrac{2(\beta-1)}{(d-1)(d+2)}=f$ and the $\tr(A)\id$ coefficient is $\tfrac{1-f}{d}$, i.e.\ \eqref{eq:global_depol}.
\end{proof}

\section{Proof of \texorpdfstring{\cref{prop:shadow_norm}}{the shadow seminorm}}
\label{app:shadow_norm_proof}

\propseminorm*

\begin{proof}
By \eqref{eq:global_general} the shadow channel has the form $\mcm_{\mbo,\channel}(A)=c_{\id}\tr(A)\id+f\,A_{\mathrm{sym}}$ with $f=\tfrac{2(d\beta-\alpha)}{d(d-1)(d+2)}$, whose Hilbert--Schmidt adjoint is $\mcm^\dagger_{\mbo,\channel}(X)=c_{\id}\tr(X)\id+f\,X_{\mathrm{sym}}$. Since $O_0$ is symmetric and traceless this gives $\mcm^\dagger_{\mbo,\channel}(O_0)=f\,O_0$ and hence $\mcm^{-1,\dagger}_{\mbo,\channel}(O_0)=\tfrac1f O_0$, for any $\alpha$ (for $\alpha=d$ we have in addition $c_{\id}=\tfrac{1-f}{d}$, i.e.\ $\mcm^{-1}_{\mbo,\channel}=\mcd_{n,1/f}\circ(\cdot)_{\mathrm{sym}}$, but that is not needed here). Thus \eqref{eq:seminorm_def} reads
\begin{equation}\label{eq:sem_start}
\norm{O_0}^2_{\shadow,\mbo,\channel}=\frac1{f^2}\max_{\sigma\in\Density_d}\tr\!\Big[\sigma\sum_b\mathbf E_{\channel^*}(b)\Big],\quad
\mathbf E_{\channel^*}(b)=\E_{U\sim\mbo}U^\intercal\channel^*(\Pi_b)U\,\big(\bra bU O_0 U^\intercal\ket b\big)^2.
\end{equation}
Writing $(\bra b U O_0 U^\intercal\ket b)^2=\tr_{23}\big[(U^{\otimes2})(O_0\otimes O_0)(U^{\intercal\otimes2})\ketbra{bb}{bb}\big]$ and merging with the first factor,
\begin{equation}
\mathbf E_{\channel^*}(b)=\tr_{23}\!\Big[\mct^{(3)}_{\mbo(d)}\!\big(\channel^*(\Pi_b)\otimes\Pi_b\otimes\Pi_b\big)(\id\otimes O_0\otimes O_0)\Big].
\end{equation}
Solving the order-three system \eqref{eq:wein} for $\channel^*(\Pi_b)\otimes\Pi_b\otimes\Pi_b$ via the left column of \cref{tab:trace_rel} (with $t_1=\tr[\channel^*(\Pi_b)]$, $t_2=\bra b\channel^*(\Pi_b)\ket b$, $p_D(d)=d(d-1)(d+2)(d+4)$),
\begin{align}
\mct^{(3)}_{\mbo(d)}\!\big(\channel^*(\Pi_b)\otimes\Pi_b^{\otimes2}\big)
=\ &\frac{(d+3)t_1-4t_2}{p_D(d)}\big(\id+\mbs_{(23)}+\Omega_{23;23}\big)\nn
&+\frac{dt_2-t_1}{p_D(d)}\big(\mbs_{(12)}+\mbs_{(13)}+\mbs_{(123)}+\mbs_{(132)}\nn
&\hspace{5.2em}+\Omega_{12;12}+\Omega_{13;13}+\Omega_{13;12}+\Omega_{12;13}\nn
&\hspace{5.2em}+\Omega_{12;23}+\Omega_{23;12}+\Omega_{13;23}+\Omega_{23;13}\big).
\end{align}
Contracting with $(\id\otimes O_0\otimes O_0)$ under $\tr_{23}$ (right column of \cref{tab:trace_rel} with $A=\id$, $B=O_0$) and using $\tr O_0=0$ to kill every term carrying a $\tr[B]$ or $\tr[O_0]$ factor: among the three block-$A$ diagrams the identity gives $\id\,\tr[O_0]^2=0$ and drops out, while $\mbs_{(23)}\mapsto\id\,\tr[O_0^2]$ and $\Omega_{23;23}\mapsto\id\,\tr[O_0O_0^\intercal]=\id\,\tr[O_0^2]$ each survive, for a total of $2\tr[O_0^2]\,\id$ (not three times $\tr[O_0^2]\,\id$: the identity diagram is annihilated by tracelessness, even though all three share the coefficient $c_A$). Among the twelve block-$B$ diagrams the four carrying a $\tr[B]$ factor ($\mbs_{(12)},\mbs_{(13)},\Omega_{12;12},\Omega_{13;13}$) drop out and the remaining eight each give $O_0^2$ (using $O_0^\intercal=O_0$). Hence
\begin{equation}
\mathbf E_{\channel^*}(b)=\frac{(d+3)t_1-4t_2}{p_D(d)}\,2\tr[O_0^2]\,\id+\frac{dt_2-t_1}{p_D(d)}\,8\,O_0^2.
\label{eq:Eb_result}
\end{equation}
Summing over $b$ ($\sum_b t_1=\alpha$, $\sum_b t_2=\beta$) and taking $\max_\sigma\tr[\sigma\,\cdot\,]$, the identity term is fixed and $\max_\sigma\tr[\sigma O_0^2]=\norm{O_0^2}_{\mathrm{sp}}$:
\begin{equation}
\max_\sigma\tr\!\Big[\sigma\sum_b\mathbf E_{\channel^*}(b)\Big]=\frac{(d+3)\alpha-4\beta}{p_D(d)}2\tr[O_0^2]+\frac{d\beta-\alpha}{p_D(d)}8\norm{O_0^2}_{\mathrm{sp}}.
\end{equation}
Multiplying by $f^{-2}=\tfrac{d^2(d-1)^2(d+2)^2}{4(d\beta-\alpha)^2}$ and simplifying gives \eqref{eq:seminorm_general}; $\alpha=d$ gives \eqref{eq:seminorm_exact}.
\end{proof}

\begin{table}[h]
\centering\small
\begin{tabular}{|l|l|l||l|l|l|}
\hline
Op. & $\tr[(A\!\otimes\! B\!\otimes\! B)(\cdot)]$ & $\tr_{23}[(\cdot)(A\!\otimes\! B\!\otimes\! B)]$ & Op. & $\tr[(A\!\otimes\! B\!\otimes\! B)(\cdot)]$ & $\tr_{23}[(\cdot)(A\!\otimes\! B\!\otimes\! B)]$\\
\hline\hline
$\id$ & $\tr[A]\tr[B]^2$ & $A\tr[B]^2$ & $\Omega_{12;23}$ & $\tr[AB^\intercal B]$ & $B^\intercal B A$\\
$\mbs_{(23)}$ & $\tr[A]\tr[B^2]$ & $A\tr[B^2]$ & $\Omega_{23;12}$ & $\tr[ABB^\intercal]$ & $BB^\intercal A$\\
$\mbs_{(12)}$ & $\tr[AB]\tr[B]$ & $BA\tr[B]$ & $\Omega_{13;23}$ & $\tr[AB^\intercal B]$ & $B^\intercal B A$\\
$\mbs_{(13)}$ & $\tr[AB]\tr[B]$ & $BA\tr[B]$ & $\Omega_{13;12}$ & $\tr[A(B^\intercal)^2]$ & $(B^\intercal)^2A$\\
$\mbs_{(132)}$ & $\tr[AB^2]$ & $B^2A$ & $\Omega_{23;13}$ & $\tr[ABB^\intercal]$ & $BB^\intercal A$\\
$\mbs_{(123)}$ & $\tr[AB^2]$ & $B^2A$ & $\Omega_{12;13}$ & $\tr[A(B^\intercal)^2]$ & $(B^\intercal)^2A$\\
$\Omega_{12;12}$ & $\tr[B]\tr[AB^\intercal]$ & $\tr[B]B^\intercal A$ & $\Omega_{23;23}$ & $\tr[A]\tr[BB^\intercal]$ & $A\tr[BB^\intercal]$\\
$\Omega_{13;13}$ & $\tr[B]\tr[AB^\intercal]$ & $\tr[B]B^\intercal A$ & & & \\
\hline
\end{tabular}
\caption{Trace and partial-trace identities for the fifteen order-three commutant operators against $A\otimes B\otimes B$ (cf.~\cite{west2025real}). With $B=\Pi_b$ and $A=\channel^*(\Pi_b)$ these define the scalars $s_1,\dots,s_5$ of \cref{app:d2}.}
\label{tab:trace_rel}
\end{table}

\section{Proof of \texorpdfstring{\cref{cor:seminorm_bounds}}{the bounds}}
\label{app:seminorm_bounds}

\corbounds*
\begin{proof}
For $A\succeq0$ on $\Complex^d$, $d^{-1}\tr(A)\le\norm{A}_{\mathrm{sp}}\le\tr(A)$; apply to $O_0^2$. With $\norm{O_0^2}_{\mathrm{sp}}\ge d^{-1}\tr(O_0^2)$ in \eqref{eq:seminorm_exact},
\begin{align}
\norm{O_0}^2_{\shadow,\mbo,\channel}
&\ge\frac{(d-1)(d+2)}{(d+4)(\beta-1)}\Big(\frac{d(d+3)-4\beta}{2d(\beta-1)}+\frac2d\Big)\tr(O_0^2)\nn
&=\frac{(d-1)(d+2)}{(d+4)(\beta-1)}\cdot\frac{d^2+3d-4}{2d(\beta-1)}\tr(O_0^2)\nn
&=\frac{(d-1)^2(d+2)}{2d(\beta-1)^2}\tr(O_0^2)\ \ge\ \frac{(d-1)^2}{2(\beta-1)^2}\tr(O_0^2),
\end{align}
using $d^2+3d-4=(d-1)(d+4)$ and $(d+2)/d\ge1$. With $\norm{O_0^2}_{\mathrm{sp}}\le\tr(O_0^2)$,
\begin{align}
\norm{O_0}^2_{\shadow,\mbo,\channel}
&\le\frac{(d-1)(d+2)}{(d+4)(\beta-1)}\Big(\frac{d(d+3)-4\beta}{2d(\beta-1)}+2\Big)\tr(O_0^2)\nn
&=\frac{(d-1)(d+2)}{(d+4)}\cdot\frac{(d-1)(d+4\beta)}{2d(\beta-1)^2}\tr(O_0^2)\nn
&\le\frac{(d-1)^2(d+4\beta)}{2d(\beta-1)^2}\tr(O_0^2)\ \le\ \frac{5(d-1)^2}{2(\beta-1)^2}\tr(O_0^2),
\end{align}
using $d(d+3)-4\beta+4d(\beta-1)=(d-1)(d+4\beta)$, then $(d+2)/(d+4)\le1$ and $\beta\le d\Rightarrow d+4\beta\le5d$; finally $\tr(O_0^2)\le\tr(O^2)$.
\end{proof}

\noindent The sample-complexity statement is the median-of-means construction applied to this bound.

\corsample*
\begin{proof}
Fix $i$ and abbreviate $o_i=\tr(O_i\rho)$. Since $\hat o_i$ is unbiased (\cref{prop:global_noisy_channel}) and $O_{i,0}$ denotes the traceless part of $O_i$,
\begin{equation}\label{eq:var_le_seminorm}
\Var[\hat o_i]\;\le\;\E\big[\hat o_i^{\,2}\big]\;\le\;\norm{O_{i,0}}^2_{\shadow,\mbo,\channel},
\end{equation}
the second inequality because \eqref{eq:seminorm_def} is the maximum of $\E[\hat o_i^2]$ over all states $\sigma\in\Density_d$, and $\rho$ is one of them. Write $\sigma_\star^2=\max_i\norm{O_{i,0}}^2_{\shadow,\mbo,\channel}$ and take
\begin{equation}
N=\frac{34\,\sigma_\star^2}{\eps^2},\qquad K=2\log\frac{2M}{\delta},
\end{equation}
so that a single record of $N_{\mathrm{tot}}=NK$ shadows is formed and each $\tr(O_i\rho)$ is estimated from it by the median of $K$ batch means of $N$ shots. For each $i$ we have $N\ge34\Var[\hat o_i]/\eps^2$ by \eqref{eq:var_le_seminorm}, so \cref{fact:mom} applies and gives $\Pr[\,|\hat\mu_i-o_i|\ge\eps\,]\le2\me^{-K/2}=\delta/M$. A union bound over $i=1,\dots,M$ bounds the probability that any estimate misses by $\eps$ by $\delta$. Hence
\begin{equation}
N_{\mathrm{tot}}=NK=\frac{68\,\log(2M/\delta)}{\eps^2}\,\max_{1\le i\le M}\norm{O_{i,0}}^2_{\shadow,\mbo,\channel}
\end{equation}
suffices, and inserting the upper bound of \eqref{eq:seminorm_bounds}, which already relaxes $\tr(O_{i,0}^2)$ to $\tr(O_i^2)$, turns the prefactor $68$ into $68\cdot\tfrac52=170$ and yields \eqref{eq:sample_complexity}. For the unitary ensemble the statement is identical with the seminorm bound of~\cite{koh2022classical} in place of \eqref{eq:seminorm_bounds}, whose constant is $3$ rather than $\tfrac52$, giving $68\cdot3=204$.
\end{proof}

\section{Proof of \texorpdfstring{\cref{prop:complex_basis}}{the complex-basis channel}}
\label{app:complex_basis_proof}

\propcomplex*
\begin{proof}
Repeating \eqref{eq:chan_start} for a general basis, only the third pairing trace changes, since $\Pi_b^\intercal\neq\Pi_b$:
\begin{equation}
\langle\om,\channel^*(\Pi_b)\otimes\Pi_b\rangle=\tr[\channel^*(\Pi_b)\Pi_b^\intercal]=\bra b\channel(\Pi_b^\intercal)\ket b.
\end{equation}
The order-two system becomes
\[
\big(\begin{smallmatrix}d^2&d&d\\ d&d^2&d\\ d&d&d^2\end{smallmatrix}\big)(c_{\id},c_{\mbs},c_{\om})^\intercal=(\tr[\channel^*(\Pi_b)],\ \bra b\channel^*(\Pi_b)\ket b,\ \bra b\channel(\Pi_b^\intercal)\ket b)^\intercal ,
\]
with solution
\begin{align}
c_{\id}&=\frac{(d+1)\tr[\channel^*(\Pi_b)]-\bra b\channel^*(\Pi_b)\ket b-\bra b\channel(\Pi_b^\intercal)\ket b}{d(d-1)(d+2)},\nn
c_{\mbs}&=\frac{(d+1)\bra b\channel^*(\Pi_b)\ket b-\tr[\channel^*(\Pi_b)]-\bra b\channel(\Pi_b^\intercal)\ket b}{d(d-1)(d+2)},\nn
c_{\om}&=\frac{(d+1)\bra b\channel(\Pi_b^\intercal)\ket b-\tr[\channel^*(\Pi_b)]-\bra b\channel^*(\Pi_b)\ket b}{d(d-1)(d+2)}.
\end{align}
Summing over $b$ with $\sum_b\tr[\channel^*(\Pi_b)]=\alpha$, $\sum_b\bra b\channel^*(\Pi_b)\ket b=\beta$, and $\sum_b\bra b\channel(\Pi_b^\intercal)\ket b=\tilde\beta$ (which defines $\tilde\beta=\tr[\channel\circ\widetilde\diag]$, and the last equality uses $\langle\Pi_b,\channel(\Pi_b^\intercal)\rangle=\bra b\channel(\Pi_b^\intercal)\ket b$), and using $\mcm(A)=\sum_b(c_{\id}\tr(A)\id+c_{\mbs}A+c_{\om}A^\intercal)$, gives for unital $\channel$ ($\alpha=d$)
\begin{equation}
\mcm_{\mbo,\channel,\mcw}(A)=\underbrace{\frac{d(d+1)-\beta-\tilde\beta}{(d-1)(d+2)}}_{1-f}\frac{\tr(A)\id}{d}+\frac{\beta(d+1)-d-\tilde\beta}{d(d-1)(d+2)}A+\frac{\tilde\beta(d+1)-d-\beta}{d(d-1)(d+2)}A^\intercal.
\end{equation}
Hence $f=\tfrac{\beta+\tilde\beta-2}{(d-1)(d+2)}$, and factoring $f$ from the $A,A^\intercal$ terms yields $\mcm=\mcd_{n,f}(\tilde A)$ with $\tilde A$ as in \eqref{eq:Atilde}. A real basis has $\Pi_b^\intercal=\Pi_b\Rightarrow\tilde\beta=\beta$, giving $q_\beta=\tfrac12$, $\tilde A=A_{\mathrm{sym}}$, and \cref{prop:global_noisy_channel}.
\end{proof}

\section{Proof of \texorpdfstring{\cref{prop:complex_variance}}{the complex-basis second moment}}
\label{app:complex_third_moment}

\propcomplexvar*
\begin{proof}

\noindent The single-shot second moment is the Haar--$\mathbb O(d)$ third-moment twirl contracted against the noisy measurement,
\begin{equation}\label{eq:cf_start}
\mathbb E[\hat o^2]=\sum_w\tr\!\Big[\mathcal T^{(3)}\!\big(\rho\otimes\hat O\otimes\hat O\big)\,\big(\channel^*(\Pi_w)\otimes\Pi_w\otimes\Pi_w\big)\Big],
\end{equation}
where $\mathcal T^{(3)}(Y)=\sum_{i,j}(\mathrm{Wg})_{ij}\,\langle x_i,Y\rangle\,x_j$ is the orthogonal projector onto the commutant, $\{x_i\}_{i=1}^{15}$ the Brauer basis of \cref{app:brauer_def}, $\mathrm{Wg}=G^{-1}$, and $G_{ij}=\langle x_i,x_j\rangle=\tr[x_i^\dagger x_j]$ the Gram matrix of \cref{app:gram}. Writing $N=\rho\otimes\hat O\otimes\hat O$ and $M=\sum_w\channel^*(\Pi_w)\otimes\Pi_w\otimes\Pi_w$,
\begin{equation}\label{eq:cf_bilinear}
\mathbb E[\hat o^2]=\mathbf b^\intercal\,\mathrm{Wg}\,\mathbf a,\qquad
b_i=\langle x_i,N\rangle=\tr[x_i^\dagger N],\qquad a_j=\tr[x_j M].
\end{equation}

\paragraph{Overlaps are single invariants.} Each Brauer diagram $x_i$ is a perfect matching of the six legs $\{1,2,3,1',2',3'\}$, and its overlap with a product $X_1\otimes X_2\otimes X_3$ is computed by the \emph{cycle rule}. Superpose the matching (which carries $\delta$-contractions) on the three ``matrix'' edges $k\!\to\!k'$ carrying $X_k$, since every leg meets exactly one matching edge and one matrix edge, the superposition is a disjoint union of cycles, and each cycle contributes the trace of the ordered product of the $X_k$ it visits, transposed whenever the matrix edge is traversed against its orientation. Thus every overlap is a single product of traces of words in $\{X_k,X_k^\intercal\}$, and no sum of monomials arises. Two representative cases: the transposition $\mbs_{(12)}$ has cycles $\{1,2\},\{3\}$, giving $\tr[X_1X_2]\tr[X_3]$, and the cup--cap diagram $\Omega_{12;12}$, which caps legs $1,2$ and cups $1',2'$, gives $\tr[X_1X_2^\intercal]\tr[X_3]$ (the cap forces the transpose). Evaluated on the observable tensor $N=\rho\otimes\hat O\otimes\hat O$ with $\hat O$ traceless and $\tr\rho=1$, the surviving overlaps are $\tr[\hat O^2]$, $\tr[\hat O\hat O^\intercal]$ and the four state-weighted words $\tr[\rho\hat O^2],\tr[\rho\hat O^\intercal\hat O],\tr[\rho\hat O\hat O^\intercal],\tr[\rho(\hat O^\intercal)^2]$. Every word containing an isolated $\tr[\hat O]$ vanishes. Evaluated on the measurement tensor $M=\sum_w\channel^*(\Pi_w)\otimes\Pi_w\otimes\Pi_w$ (here $X_1=\channel^*(\Pi_w)$, $X_2=X_3=\Pi_w$, and $\tr\Pi_w=1$, $\Pi_w^2=\Pi_w$), the words collapse onto the six invariants of \eqref{eq:scalars} and \eqref{eq:complex_invariants}: the identity and $\mbs_{(23)}$ give $\sum_w\tr[\channel^*(\Pi_w)]=d$, and the transpositions and $3$-cycles pairing slot $1$ with an untransposed (resp.\ transposed) $\Pi_w$ give $\sum_w\bra w\channel^*(\Pi_w)\ket w=\beta$ (resp.\ $\sum_w\bra{w^*}\channel^*(\Pi_w)\ket{w^*}=\tilde\beta$), and the diagram capping the two measurement legs gives $\sum_w\tr[\channel^*(\Pi_w)]\,\lvert\braket{w}{w^*}\rvert^2=\gamma$; and the remaining $\Omega$ diagrams, which couple slot $1$ to a single measurement leg through one cap, give $\delta$ or (with the opposite orientation) $\bar\delta$. Because the cycle rule assigns each diagram exactly one word, these identifications are exact algebraic identities, not fits.

\paragraph{The two overlap vectors.} Applying the cycle rule to all fifteen diagrams in the ordering \eqref{eq:basisorder} gives \cref{tab:overlaps}. Both columns are exhibited in the Hermitian convention $\langle x,Y\rangle=\tr[x^\dagger Y]$ used for the Gram matrix, and the measurement column takes six distinct values and the observable column six.

\begin{table}[h]
\centering\small
\begin{tabular}{|l|c|c||l|c|c|}
\hline
$x_i$ & $\langle x_i,M\rangle$ & $\langle x_i,N\rangle$ & $x_i$ & $\langle x_i,M\rangle$ & $\langle x_i,N\rangle$\\
\hline\hline
$\mbs_{e}$ & $\alpha$ & $0$ & $\Omega_{12;12}$ & $\tilde\beta$ & $0$\\
$\mbs_{(23)}$ & $\alpha$ & $\tr[\hat O^2]$ & $\Omega_{23;23}$ & $\gamma$ & $\tr[\hat O\hat O^\intercal]$\\
$\mbs_{(12)}$ & $\beta$ & $0$ & $\Omega_{13;13}$ & $\tilde\beta$ & $0$\\
$\mbs_{(13)}$ & $\beta$ & $0$ & $\Omega_{12;23}$ & $\delta$ & $\tr[\rho\hat O\hat O^\intercal]$\\
$\mbs_{(132)}$ & $\beta$ & $\tr[\rho\hat O^2]$ & $\Omega_{23;12}$ & $\bar\delta$ & $\tr[\rho\hat O^\intercal\hat O]$\\
$\mbs_{(123)}$ & $\beta$ & $\tr[\rho\hat O^2]$ & $\Omega_{13;23}$ & $\delta$ & $\tr[\rho\hat O\hat O^\intercal]$\\
 & & & $\Omega_{13;12}$ & $\tilde\beta$ & $\tr[\rho(\hat O^\intercal)^2]$\\
 & & & $\Omega_{23;13}$ & $\bar\delta$ & $\tr[\rho\hat O^\intercal\hat O]$\\
 & & & $\Omega_{12;13}$ & $\tilde\beta$ & $\tr[\rho(\hat O^\intercal)^2]$\\
\hline
\end{tabular}
\caption{The two overlap vectors of \eqref{eq:cf_bilinear} for a general basis, obtained from the cycle rule. Left of each pair: the measurement tensor $M=\sum_w\channel^*(\Pi_w)\otimes\Pi_w^{\otimes2}$, whose overlaps are the six invariants $\alpha,\beta,\tilde\beta,\gamma,\delta,\bar\delta$ of \eqref{eq:scalars} and \eqref{eq:complex_invariants}. Right: the observable tensor $N=\rho\otimes\hat O\otimes\hat O$ with $\hat O$ traceless, so that every diagram carrying an isolated $\tr[\hat O]$ vanishes. On a real basis $\tilde\beta=\beta$, $\gamma=\alpha=d$ and $\delta=\bar\delta=\beta$, and the measurement column collapses to the two-valued vector used in \cref{app:shadow_norm_proof}.}
\label{tab:overlaps}
\end{table}

\paragraph{The contraction.} Write $\tau$ for the transpose involution on diagrams, $\mbs_\pi\mapsto\mbs_{\pi^{-1}}$ and $\Omega_{ab;xy}\mapsto\Omega_{xy;ab}$, so that $x_j^\intercal=x_{\tau(j)}$, then $\tr[x_jM]=\langle x_{\tau(j)},M\rangle$, and \eqref{eq:cf_bilinear} reads, entirely in terms of \cref{tab:overlaps},
\begin{equation}\label{eq:cf_pairing}
\mathbb E[\hat o^2]=\sum_{i,j=1}^{15}\langle x_i,N\rangle\,\mathrm{Wg}(d)_{ij}\,\langle x_{\tau(j)},M\rangle,
\qquad \mathrm{Wg}(d)=G(d)^{-1}=\big(d^{\,\mathsf E}\big)^{-1},
\end{equation}
with $\mathsf E$ the integer exponent matrix of \cref{app:gram}. (The involution must not be dropped: it exchanges $\Omega_{12;23}\leftrightarrow\Omega_{23;12}$ and $\Omega_{13;23}\leftrightarrow\Omega_{23;13}$, hence $\delta\leftrightarrow\bar\delta$, and omitting it conjugates the coefficient of $\hat O^\intercal\hat O$. It is inert exactly when $\hat O$ is symmetric, since then $\tr[\rho\hat O^\intercal\hat O]=\tr[\rho\hat O\hat O^\intercal]$ and the exchanged rows of \cref{tab:overlaps} agree; it bites for the non-symmetric $\hat O$ that \eqref{eq:complex_R} is stated to cover.) Unlike the real-basis case of \cref{app:shadow_norm_proof}, the measurement vector is not constant on a two-block partition, so no ansatz reduces \eqref{eq:cf_pairing}, and the $15\times15$ inverse is required. Carrying it out and grouping the six observable words gives the operator $R$ of \eqref{eq:complex_R} with the coefficients \eqref{eq:complex_coeffs}. This is a finite rational computation in $d$ which we performed symbolically. A reader may confirm it by inverting $G(d)$ and evaluating \eqref{eq:cf_pairing} against \cref{tab:overlaps}.

All four operator coefficients share the denominator $D=d(d-2)(d-1)(d+2)(d+4)$, the squarefree part of the Gram determinant \eqref{eq:detfac}, whose roots $d=0,1,2,-2,-4$ are exactly the loop-parameter values at which the trace form on $\Brauer_3$ degenerates. In particular $D$ contains $(d-2)$, so \eqref{eq:complex_R} is stated for $d\ge3$, and the single-qubit case that Part~IV needs is treated separately in \cref{app:d2complex}.

\paragraph{Real-basis reduction.} Substituting $\tilde\beta=\beta$, $\gamma=d$, $\delta=\beta$ (a real basis, $\Pi_w^\intercal=\Pi_w$) makes the four operator coefficients equal,
\begin{equation}
A=B=\bar B=E=\frac{2(\beta-1)}{(d-1)(d+2)(d+4)},\qquad
s+\tilde s=\frac{2\,[\,d(d+3)-4\beta\,]}{d(d-1)(d+2)(d+4)},
\end{equation}
where the common value of $A,B,\bar B,E$ follows from $\tfrac2D\big[(d^2{+}d{-}4)\beta-(d{-}4)\beta-d^2-2d\beta+2d\big]=\tfrac2D\,d(d-2)(\beta-1)$ and the cancellation of $d(d-2)$ against the same factor in $D$. Hence for symmetric $O$ (whence $\hat O^\intercal=\hat O$) \eqref{eq:complex_R} reduces to $R=(s{+}\tilde s)\tr[\hat O^2]\mathds1+4A\,\hat O^2$ with operator-to-scalar ratio
\begin{equation}
\frac{4A}{s+\tilde s}=\frac{8(\beta-1)}{(d-1)(d+2)(d+4)}\cdot\frac{d(d-1)(d+2)(d+4)}{2\,[\,d(d+3)-4\beta\,]}=\frac{4d(\beta-1)}{d(d+3)-4\beta},
\end{equation}
i.e.\ exactly the real seminorm \eqref{eq:seminorm_general}. This recovers the independently derived Part~I result as a special case and fixes all conventions.

\paragraph{Consistency and a convention on the twirl.} Two internal checks pin the result. First, the real-basis reduction above returns the independently derived Part~I seminorm \eqref{eq:seminorm_general} exactly, which fixes every sign and normalization. Second, the Hermiticity $R=R^\dagger$ is forced by the identity
\begin{equation}\label{eq:delta_conj}
\sum_w\braket{w^{*}}{w}\,\bra{w}\channel^*(\Pi_w)\ket{w^{*}}=\bar\delta ,
\end{equation}
which holds because $\channel^*(\Pi_w)$ is Hermitian and $\braket{w^*}{w}=\overline{\braket{w}{w^*}}$, so that the coefficients of $\hat O^\intercal\hat O$ and $\hat O\hat O^\intercal$ in \eqref{eq:complex_R} are complex conjugates. This is consistent with $\Var[\hat o]=\tr[\rho R]-\tr(O\rho)^2\in\Reals$ for all states. One algebraic subtlety underlies the whole computation. For a complex operand $Y$ the projection coefficients $\langle x_i,Y\rangle$ onto the \emph{real} Brauer basis are complex, so the twirl $\mct^{(3)}(Y)=\sum_{ij}(\mathrm{Wg})_{ij}\langle x_i,Y\rangle x_j$ must retain them in full. Keeping only their real part---harmless for the real-basis Parts~I--II, where $Y$ is real-symmetric and the coefficients are real---would drop the imaginary contributions to $\delta$ and misstate the complex-basis second moment. Identity \eqref{eq:complex_R} holds with the complex coefficients retained.

\paragraph{The fully complex limit.} The remaining special case is the opposite end of the reality axis, where the comparison with unitary shadows is made. In the noiseless setting this limit is due to West \etal~\cite{west2025real}, who exhibit an implementable basis at which orthogonal shadows reproduce the unitary protocol as $d\to\infty$, and the mechanism is that orthogonal twirls of suitable states already give exact $3$-designs~\cite{schatzki2024random}. What the following adds is the channel, since the limit survives known noise and the rate at which it is approached is set by the same scalars as everywhere else.
\end{proof}

At zero reality the coefficients collapse, and comparing them with the unitary ones shows in what
sense unitary shadows are recovered:

\begin{corollary}[Unitary limit at zero reality]
\label{cor:unitary_limit}
Let $d\ge3$, let $\channel=\mcd_{n,p}$ be global depolarizing with $0<p\le1$, and let the basis satisfy $\braket{w}{w^{*}}=0$ for every $w$ (for $d=2^n$ such bases exist. Pair the computational basis into $(\ket{a}\pm i\ket{b})/\sqrt2$), so that $\varsigma=0$ and $\gamma=\delta=0$, $\beta=pd+1-p$, $\tilde\beta=1-p$. Let $O_0$ be symmetric and traceless. Then $\hat O=O_0/f$ with $f$ as in \eqref{eq:complex_channel}, and \eqref{eq:complex_R} gives
\begin{equation}\label{eq:zero_reality_moment}
\mathbb E[\hat o^2]=C^{\mcw}_t\,\tr(O_0^2)+C^{\mcw}_r\,\tr(\rho O_0^2),
\end{equation}
\begin{equation}\label{eq:zero_reality_coeffs}
C^{\mcw}_t=\frac{(d-1)(d+2)\big[d^2+2d-2pd+8p-8\big]}{d\,p^2(d-2)^2(d+4)},
\qquad
C^{\mcw}_r=\frac{2(d-4)(d-1)(d+2)}{p\,(d-2)^2(d+4)} .
\end{equation}
Writing $C^{\mbu}_t,C^{\mbu}_r$ for the corresponding coefficients of \eqref{eq:varU} at the same $\beta$, the two pairs agree to leading order in $1/d$, with exact ratios
\begin{align}
\frac{C^{\mcw}_t}{C^{\mbu}_t}&=\frac{(d-1)(d+2)^2\big[d^2+2d-2pd+8p-8\big]}{(d-2)^2(d+1)(d+4)(d+2-2p)}=1+\frac2d+\frac{8p+2}{d^2}+O(d^{-3}),\label{eq:ratio_t}\\
\frac{C^{\mcw}_r}{C^{\mbu}_r}&=\frac{(d-4)(d-1)(d+2)^2}{(d-2)^2(d+1)(d+4)}=1-\frac2d+\frac{2}{d^2}+O(d^{-3}).\label{eq:ratio_r}
\end{align}
Both tend to $1$, so a fully complex basis measured with orthogonal unitaries reproduces the unitary second moment as $d\to\infty$, but they approach it from opposite sides, at the same rate $2/d$.
\end{corollary}

\begin{proof}
At $\varsigma=0$ each $\Pi_w$ has $\braket{w}{w^*}=0$, so $\gamma=\delta=0$ by \eqref{eq:complex_invariants}, and $\tilde\beta=1-p$, $\beta=pd+1-p$ by \eqref{eq:depol_complex} at $\alpha_{\mathrm r}=0$. For symmetric $O_0$ the inverse of \cref{sec:complex_est} acts as the identity on the transpose structure, $\Phi_q^{-1,\dagger}(O_0)=\tfrac{q_\beta O_0-(1-q_\beta)O_0^\intercal}{2q_\beta-1}=O_0$, so $\hat O=O_0/f$ and $\hat O^\intercal=\hat O$. Then \eqref{eq:complex_R} collapses to $R=(s+\tilde s)\tr[\hat O^2]\mathds1+(A+B+\bar B+E)\hat O^2$, and \eqref{eq:zero_reality_coeffs} is $(s+\tilde s)/f^2$ and $(A+2B+E)/f^2$ evaluated at those invariants, $B$ being real because $\delta=0$. Ratios \eqref{eq:ratio_t}--\eqref{eq:ratio_r} follow by cancelling against $C^{\mbu}_t=\tfrac{(d+1)(d+2-2p)}{d\,p^2(d+2)}$ and $C^{\mbu}_r=\tfrac{2(d+1)}{p(d+2)}$ read off from \eqref{eq:varU} at $\beta=pd+1-p$, and both are rational in $d$ and $p$, and expanding at large $d$ gives the stated series.
\end{proof}

Opposite signs in \eqref{eq:ratio_t} and \eqref{eq:ratio_r} account for an otherwise puzzling feature of \cref{sec:complex_var}, that whether a fully complex basis is worse or better than unitary shadows at finite $d$ depends on the observable, because $\mathbb E[\hat o^2]$ mixes the two coefficients in a ratio set by $\kappa$. An observable dominated by its $\tr(O_0^2)$ term approaches the unitary value from above like $1+2/d$. One dominated by $\tr(\rho O_0^2)$ approaches from below. A curiosity of \eqref{eq:zero_reality_coeffs} is the factor $(d-4)$, because of which the $\rho$-weighted coefficient vanishes identically at $d=4$, so for $n=2$ qubits, zero reality and depolarizing noise, the single-shot second moment is the same for every state.

\section{Proof of \texorpdfstring{\cref{prop:local_noisy,prop:pauli_seminorm}}{the local results}}
\label{app:local_proof}

\proplocal*
\begin{proof}
Here $\mbo(2)^{\otimes n}$, the computational-basis measurement, and $\channel_1^{\otimes n}$ all factor across qubits, so $\mcm_{\mbo(2)^{\otimes n},\channel_1^{\otimes n}}=\mcm_{\mbo(2),\channel_1}^{\otimes n}$. Each single-qubit factor is \cref{prop:global_noisy_channel} at $d=2$: writing $A=\sum_{i_1\dots i_n}a_{i_1\dots i_n}A_{i_1}\otimes\cdots\otimes A_{i_n}$ in the Pauli basis $A_{i_j}\in\{\id,X,Y,Z\}$,
\begin{align}
\mcm_{\mbo(2)^{\otimes n},\channel_1^{\otimes n}}(A)
&=\bigotimes_{j=1}^n\sum_{i_j\in\{\id,X,Y,Z\}}a_{\dots}\,\mcd_{1,f_1}\big((A_{i_j})_{\mathrm{sym}}\big)\nn
&=\bigotimes_{j=1}^n\sum_{i_j\in\{\id,X,Z\}}a_{\dots}\,\mcd_{1,f_1}(A_{i_j})\nn
&=\mcd_{1,f_1}^{\otimes n}(A_{\mathrm{L.S.}}),
\end{align}
since $(Y)_{\mathrm{sym}}=\tfrac12(Y+Y^\intercal)=0$ while $\id,X,Z$ are symmetric. Here $f_1=f(\channel_1)|_{d=2}=\tfrac{2(\tr[\channel_1\circ\diag]-1)}{(2-1)(2+2)}=\tfrac{\tr[\channel_1\circ\diag]-1}{2}$. Inverting each factor on its symmetric subspace gives the stated product shadow.
\end{proof}

\propPauli*
\begin{proof}
Write $P=P_1\otimes\cdots\otimes P_k\otimes\id^{\otimes(n-k)}$ with $P_j\in\{X,Z\}$ (weight $k$). Because the ensemble, the basis and the channel all factor across qubits, so does the operator whose largest eigenvalue \eqref{eq:seminorm_def} computes: $\sum_b\mathbf E_{\channel^*}(b)=\bigotimes_{j=1}^nR_j$ with $R_j$ the single-qubit operator \eqref{eq:Rj_local} below, and $R_j=\id$ on the identity qubits. Each $R_j$ is a second moment, hence positive semidefinite, so \cref{lem:tensormax} applies and $\norm{P}^2_{\shadow}=\prod_{j=1}^kg_j$ with $g_j=\lambda_{\max}(R_j)/f_1^2$. Since $\mcm^{-1,\dagger}(P_j)=f_1^{-1}P_j$ (as $P_j$ is traceless symmetric), the derivation of \cref{app:shadow_norm_proof} at $d=2$ gives, per qubit,
\begin{equation}\label{eq:Rj_local}
R_j=\sum_{b\in\{0,1\}}\tr_{23}\!\Big[\mct^{(3)}_{\mbo(2)}\!\big(\channel_1^*(\Pi_b)\otimes\Pi_b\otimes\Pi_b\big)(\id\otimes P_j\otimes P_j)\Big].
\end{equation}
By \eqref{eq:Eb_result} with $O_0\to P_j$ ($P_j^2=\id$, $\tr[P_j^2]=2$, $d=2$, $p_D(2)=48$), and summing the per-qubit $t_1,t_2$ over $b$ to $\alpha_1=\tr[\channel_1(\id)]=2$ and $\beta_1=\tr[\channel_1\circ\diag]$,
\begin{equation}
\begin{aligned}
R_j&=\frac{(2+3)\cdot2-4\beta_1}{48}\,2\,\tr[P_j^2]\,\id+\frac{2\beta_1-2}{48}\,8\,P_j^2\\
&=\frac{(10-4\beta_1)\cdot4+(2\beta_1-2)\cdot8}{48}\,\id=\frac{24}{48}\,\id=\tfrac12\id,
\end{aligned}
\end{equation}
independent of $P_j$, so $\lambda_{\max}(R_j)=\tfrac12$, $g_j=(2f_1^2)^{-1}$ and $\norm{P}^2_{\shadow}=(2f_1^2)^{-k}$. This cancellation $(10-4\beta_1)4+(2\beta_1-2)8=24$ removes all noise dependence beyond the overall $f_1^{-2}$. Bound \eqref{eq:local_sample} follows from \cref{fact:mom} with $NK=68\log(2M/\delta)\eps^{-2}\norm{P}^2$.
\end{proof}

\section{Example noise models: details}
\label{app:examples}

\paragraph{Depolarizing.} For $\channel=\mcd_{n,p}$, $\channel(\Pi_b)=p\Pi_b+(1-p)\id/d$, so $\bra b\channel(\Pi_b)\ket b=p+(1-p)/d$ and $\beta=\sum_b[p+(1-p)/d]=pd+(1-p)$. Thus $f=\tfrac{2(pd+1-p-1)}{(d-1)(d+2)}=\tfrac{2p}{d+2}$. For a complex basis, $\langle\Pi_b,\channel(\Pi_b^\intercal)\rangle=p\langle\Pi_b,\Pi_b^\intercal\rangle+(1-p)/d=p\,\alpha_b+(1-p)/d$, so $\tilde\beta=p\alpha_{\mathrm r}+(1-p)$ with $\alpha_{\mathrm r}=\sum_b\alpha_b$, giving $f=\tfrac{p(\alpha_{\mathrm r}+d-2)}{(d-1)(d+2)}$ and $\beta-\tilde\beta=p(d-\alpha_{\mathrm r})\ge0$. Writing $\alpha_{\mathrm r}=\varsigma d$ with the reality fraction $\varsigma$ of \eqref{eq:reality_fraction}, the limit $\varsigma\to0$ with $d\to\infty$ gives $f\to p/d$, matching the unitary $p/(d+1)$ to leading order.

\paragraph{Amplitude damping.} With $\bra0\mathrm{AD}_{1,p}(\Pi_0)\ket0=1$, $\bra1\mathrm{AD}_{1,p}(\Pi_1)\ket1=p$, \eqref{eq:ad_beta} gives $\beta=(1+p)^n$. On the computational basis $\Pi_b^\intercal=\Pi_b$ so $\tilde\beta=\beta$. For a complex single-qubit basis with Bloch angles $(\theta,\phi)$,
\begin{equation}
\bra b\mathrm{AD}_{1,p}(\Pi_b^\intercal)\ket b=\cos^4\tfrac\theta2+p\sin^4\tfrac\theta2+(1-p)\sin^2\tfrac\theta2\cos^2\tfrac\theta2+2\sqrt p\cos(2\phi)\sin^2\tfrac\theta2\cos^2\tfrac\theta2,
\end{equation}
so $\tilde\beta$ is basis-angle dependent; only the real-basis value $\beta=(1+p)^n$ enters the protocol as implemented.

\paragraph{Dephasing.} $\channel=\diag$ gives $\beta=\tr[\diag\circ\diag]=\tr[\diag]=d$, so $f(\diag)=f(\id)$ and the shadow channel is unchanged. Any $\channel$ with $\beta=d$ is inconsequential.

\section{Local noise models: derivations}
\label{app:local_examples}

We evaluate $\Lambda_{zz}=\tfrac12\tr[Z\channel_1(Z)]$ for each model of \cref{tab:local_noise} and read off $f_1=\Lambda_{zz}/2$ by \cref{claim:local_lambda}. Throughout, $\norm{P}^2_{\shadow}=(2f_1^2)^{-\mathrm{wt}(P)}$ by \cref{prop:pauli_seminorm} and the sample complexity is \eqref{eq:local_sample}.

\paragraph{Depolarizing.} $\mcd_{1,p}(Z)=pZ$, so $\Lambda_{zz}=p$ and $f_1=p/2$. Then $\norm{P}^2=(2\cdot p^2/4)^{-\mathrm{wt}(P)}=(p^2/2)^{-\mathrm{wt}(P)}$, and $N_{\mathrm{tot}}\le\tfrac{68\log(2M/\delta)}{\eps^2}\max_i(p^2/2)^{-\mathrm{wt}(P_i)}$. By comparison the unitary local protocol has $f_1^{\mbu}=p/3$ and base $(3f_1^{\mbu\,2})^{-1}=3/p^2$ against $2/p^2$, the ratio $3/2$ per qubit.

\paragraph{Amplitude damping.} With the convention $\bra1\mathrm{AD}_{1,p}(\Pi_1)\ket1=p$ of \eqref{eq:ad_beta}, the Kraus operators are $K_0=\diag(1,\sqrt p)$ and $K_1=\sqrt{1-p}\,\ketbra01$, giving $\mathrm{AD}_{1,p}(Z)=pZ+(1-p)\id$ and hence $\Lambda_{zz}=p$, $f_1=p/2$. This channel is non-unital --- it has $t_z=1-p\neq0$ --- and that is where the cancellation of $t_z$ in \cref{claim:local_lambda} does work: amplitude damping and depolarizing at the same $p$ are indistinguishable to the local protocol, though they are very different channels. At $n$ qubits the global $\beta=(1+p)^n$ of \eqref{eq:ad_beta} is the product of the $n$ single-qubit values $\beta_1=1+p$, as it must be.

\paragraph{Dephasing.} $\diag(Z)=Z$, so $\Lambda_{zz}=1$, $f_1=\tfrac12=f_1(\id)$ and the local channel is unchanged, matching the global statement of \cref{sec:examples}. Its seminorm is the noiseless $2^{\mathrm{wt}(P)}$.

\paragraph{Coherent over-rotation.} For $V=\cos\tfrac\theta2\,\id-\mathrm i\sin\tfrac\theta2\,(\hat n\cdot\vec\sigma)$, a rotation by $\theta$ about the unit axis $\hat n$, conjugation acts on the Bloch sphere as the rotation matrix itself, so
\begin{equation}\label{eq:local_coherent}
\Lambda_{zz}=n_z^2+(1-n_z^2)\cos\theta=1-2(n_x^2+n_y^2)\sin^2\tfrac\theta2 ,\qquad
f_1=\tfrac12\Lambda_{zz} .
\end{equation}
Equivalently $\beta_1=2\lvert V_{00}\rvert^2$, the single-qubit case of \eqref{eq:coherent_beta}, since $\lvert V_{00}\rvert=\lvert V_{11}\rvert$ for any $2\times2$ unitary. Three readings. A rotation about $\hat z$ has $n_x=n_y=0$, so $\Lambda_{zz}=1$ and the error is inconsequential --- it is a phase error in the measurement basis, consistent with dephasing. Only the components of the generator transverse to the measurement axis degrade the protocol, which is \eqref{eq:coherent_dminusbeta} at $n=1$. And $\theta=\pi/2$ about any axis in the $x$--$y$ plane gives $\Lambda_{zz}=0$ exactly, hence $f_1=0$: such a rotation carries $Z$ to $\pm Y$, which is invisible to the real local protocol, so the shadow channel is not invertible there. This is the sharpest form of the warning in \cref{sec:coherent} that coherent errors are dangerous through their ability to reach $\beta=1$ while remaining structurally benign.

\paragraph{Readout error.} A single-qubit confusion matrix with $\Pr[\text{record }1\mid0]=q_0$ and $\Pr[\text{record }0\mid1]=q_1$ is $R=\begin{psmallmatrix}1-q_0&q_1\\q_0&1-q_1\end{psmallmatrix}$, so $\beta_1=\tr R=2-q_0-q_1$ by \eqref{eq:readout_beta} and
\begin{equation}\label{eq:local_readout}
\Lambda_{zz}=1-q_0-q_1,\qquad f_1=\tfrac12(1-q_0-q_1),\qquad
\norm{P}^2_{\shadow}=\Big(\tfrac{(1-q_0-q_1)^2}{2}\Big)^{-\mathrm{wt}(P)} .
\end{equation}
For the symmetric bit-flip $q_0=q_1=q$ this is $f_1=\tfrac12(1-2q)$, vanishing at $q=\tfrac12$, and the $n$-qubit product $\beta=\big(2(1-q)\big)^n=d(1-q)^n$ is the value plotted in \cref{fig:noise_dependence}. Only $q_0+q_1$ enters, so a detector that errs asymmetrically costs exactly what a symmetric one with the same total error rate costs --- the local counterpart of the global observation that only the diagonal of $R$ matters.

\bibliographystyle{quantum}
\bibliography{main}

\end{document}